\documentclass[fleqn,reqno]{amsart}
\usepackage[T1]{fontenc}
\usepackage[latin9]{inputenc}
\usepackage{verbatim}
\usepackage{enumitem}
\usepackage{amstext}
\usepackage{amsthm}
\usepackage{amssymb}

\makeatletter
\numberwithin{equation}{section}
\numberwithin{figure}{section}
\theoremstyle{plain}
\newtheorem{thm}{\protect\theoremname}[section]
\theoremstyle{plain}
\newtheorem{cor}[thm]{\protect\corollaryname}
\theoremstyle{plain}
\newtheorem{lem}[thm]{\protect\lemmaname}
\theoremstyle{plain}
\newtheorem{prop}[thm]{\protect\propositionname}
\theoremstyle{plain}
\newtheorem{fact}[thm]{\protect\factname}
\theoremstyle{definition}
\newtheorem{defn}[thm]{\protect\definitionname}
\theoremstyle{remark}
\newtheorem{claim}[thm]{\protect\claimname}
\theoremstyle{remark}
\newtheorem{rem}[thm]{\protect\remarkname}
\theoremstyle{plain}
\newtheorem*{thm*}{\protect\theoremname}

\usepackage{amsmath}
\usepackage{sherstov}
\usepackage{small-caption}
\usepackage{setspace}
\usepackage{centernot}
\usepackage{bbm}

\newcommand{\moo}{\{-1,+1\}}
\newcommand{\classAC}{\operatorname{\mathsf{AC}}}
\newcommand{\orth}{\operatorname{orth}}

\usepackage{hyperref}
\hypersetup{
pdfauthor = {A. A. Sherstov},
plainpages = false,
bookmarks = true,
bookmarksopen = true,
pdftex,
colorlinks=true,
citecolor=black,
linkcolor=black,
urlcolor=black,
allcolors=black
}

\newtheoremstyle{myplain}      {10pt}{10pt}{\itshape}{}{\scshape}{.}{.5em}{}
\newtheoremstyle{mydefinition} {10pt}{10pt}{}{}{\scshape}{.}{.5em}{}
\newtheoremstyle{myremark} {10pt}{10pt}{}{}{\itshape}{.}{.5em}{}

\@namedef{thm}{\@thm{\let \thm@swap \@gobble \th@myplain }{thm}{Theorem}}
\@namedef{lem}{\@thm{\let \thm@swap \@gobble \th@myplain }{thm}{Lemma}}
\@namedef{cor}{\@thm{\let \thm@swap \@gobble \th@myplain }{thm}{Corollary}}
\@namedef{prop}{\@thm{\let \thm@swap \@gobble \th@myplain }{thm}{Proposition}}
\@namedef{claim}{\@thm{\let \thm@swap \@gobble \th@myplain }{thm}{Claim}}
\@namedef{fact}{\@thm{\let \thm@swap \@gobble \th@myplain }{thm}{Fact}}
\@namedef{defn}{\@thm{\let \thm@swap \@gobble \th@mydefinition }{thm}{Definition}}
\@namedef{rem}{\@thm{\let \thm@swap \@gobble \th@mydefinition }{thm}{Remark}}
\@namedef{example}{\@thm{\let \thm@swap \@gobble \th@mydefinition }{thm}{Example}}

\@namedef{thm*}{\@thm {\th@myplain }{}{Theorem}}
\@namedef{lem*}{\@thm {\th@myplain }{}{Lemma}}
\@namedef{cor*}{\@thm {\th@myplain }{}{Corollary}}
\@namedef{prop*}{\@thm {\th@myplain }{}{Proposition}}
\@namedef{claim*}{\@thm {\th@myplain }{}{Claim}}
\@namedef{fact*}{\@thm {\th@myplain }{}{Fact}}
\@namedef{defn*}{\@thm {\th@mydefinition }{}{Definition}}
\@namedef{rem*}{\@thm {\th@mydefinition }{}{Remark}}
\@namedef{example*}{\@thm {\th@mydefinition }{}{Example}}

\renewcommand{\mathcal}[1]{\mathscr{#1}}

\makeatletter
\def\@seccntformat#1{%
  \protect\textup{%
    \protect\@secnumfont
    \expandafter\protect\csname format#1\endcsname 
    \csname the#1\endcsname
    \protect\@secnumpunct
  }%
}

\makeatletter
\newcommand \SparseDotfill {\leavevmode \leaders \hb@xt@ .7em{\hss .\hss }\hfill \kern \z@}
\makeatother	

\makeatletter
\def\@tocline#1#2#3#4#5#6#7{\relax
  \ifnum #1>\c@tocdepth 
  \else
    \par \addpenalty\@secpenalty\addvspace{\ifnum #1=1 2mm \else #2\fi}%
    \begingroup \hyphenpenalty\@M
    \@ifempty{#4}{%
      \@tempdima\csname r@tocindent\number#1\endcsname\relax
    }{%
      \@tempdima#4\relax
    }%
    \parindent\z@ \leftskip#3\relax \advance\leftskip\@tempdima\relax
    \rightskip\@pnumwidth plus4em \parfillskip-\@pnumwidth
          \ifnum #1=1 \bfseries #5\else #5\fi 
   \leavevmode\hskip-\@tempdima
      \ifcase #1
       \or\or \hskip 1em \or \hskip 2em \else \hskip 3em \fi%
#6     \nobreak\relax
{\ifnum #1=1\hfill \else \SparseDotfill\fi}
 \hbox to\@pnumwidth{\@tocpagenum{
    \ifnum #1=1 \bfseries \fi #7}}\par
    \nobreak
    \endgroup
  \fi}
\makeatother

\usepackage{enumitem}
\setenumerate{label=\rm (\roman*)}

\newcommand{\iu}{\mathbf{i}}
\DeclareMathOperator{\realpart}{Re}
\DeclareMathOperator{\imagpart}{Im}

\usepackage{graphicx}

\providecommand{\noopsort}[1]{}

\newcommand{\pomo}{\{-1,+1\}}

\newcommand{\onedeg}{\deg^+}

\DeclareMathOperator{\iterdisc}{rdisc}

\AtBeginDocument{
  
}

\makeatother

\providecommand{\claimname}{Claim}
\providecommand{\corollaryname}{Corollary}
\providecommand{\definitionname}{Definition}
\providecommand{\factname}{Fact}
\providecommand{\lemmaname}{Lemma}
\providecommand{\propositionname}{Proposition}
\providecommand{\remarkname}{Remark}
\providecommand{\theoremname}{Theorem}

\begin{document}
\title[The Approximate Degree of DNF and CNF Formulas]{The Approximate Degree of DNF and CNF Formulas$^{*}$}
\author{Alexander A. Sherstov}
\thanks{$^{*}$ This manuscript is a much-expanded version of the STOC '22
paper, with several new results. Work supported by NSF grants CCF-1814947
and CCF-2220232. }
\thanks{\emph{\hspace{7mm}Author affiliation:} Computer Science Department,
UCLA, Los Angeles, CA~90095. \emph{Email:} \texttt{sherstov@cs.ucla.edu}}
\begin{abstract}
The \emph{approximate degree} of a Boolean function $f\colon\zoon\to\zoo$
is the minimum degree of a real polynomial $p$ that approximates
$f$ pointwise: $|f(x)-p(x)|\leq1/3$ for all $x\in\zoon.$ For every
$\delta>0,$ we construct CNF and DNF formulas of polynomial size
with approximate degree $\Omega(n^{1-\delta}),$ essentially matching
the trivial upper bound of $n.$ This improves polynomially on previous
lower bounds and fully resolves the approximate degree of constant-depth
circuits $(\classAC^{0}),$ a question that has seen extensive research
over the past 10~years. Prior to our work, an $\Omega(n^{1-\delta})$
lower bound was known only for $\classAC^{0}$ circuits of depth that
grows with $1/\delta$ (Bun and Thaler, FOCS~2017). Furthermore,
the CNF and DNF formulas that we construct are the simplest possible
in that they have \emph{constant} width. Our result holds even for
\emph{one-sided} approximation: for any~$\delta>0$, we construct
a polynomial-size constant-width CNF formula with one-sided approximate
degree $\Omega(n^{1-\delta})$. 

Our work has the following consequences.
\begin{enumerate}
\item We essentially settle the communication complexity of $\classAC^{0}$
circuits in the bounded-error quantum model, $k$-party number-on-the-forehead
randomized model, and $k$-party number-on-the-forehead nondeterministic
model: we prove that for every $\delta>0$, these models require $\Omega(n^{1-\delta})$,
$\Omega(n/4^{k}k^{2})^{1-\delta}$, and $\Omega(n/4^{k}k^{2})^{1-\delta}$,
respectively, bits of communication even for polynomial-size constant-width
CNF formulas.
\item In particular, we show that the multiparty communication class $\coNP_{k}$
can be separated essentially optimally from $\NP_{k}$ and $\BPP_{k}$
by a particularly simple function, a polynomial-size constant-width
CNF formula.
\item We give an essentially tight separation, of $O(1)$ versus $\Omega(n^{1-\delta})$,
for the one-sided versus two-sided approximate degree of a function;
and $O(1)$ versus $\Omega(n^{1-\delta})$ for the one-sided approximate
degree of a function $f$ versus its negation $\neg f$.
\end{enumerate}
\medskip{}
Our proof departs significantly from previous approaches and contributes
a novel, number-theoretic method for amplifying approximate degree.
\end{abstract}

\maketitle
\belowdisplayskip=8pt plus 3pt minus 2pt 
\abovedisplayskip=8pt plus 3pt minus 2pt 
\thispagestyle{empty}
\allowdisplaybreaks[1]

\newpage\thispagestyle{empty}
\hypersetup{linkcolor=black} \tableofcontents{}\newpage{}

\hypersetup{linkcolor=black} 
\thispagestyle{empty}

\hyphenation{com-po-nent-wise}

\pagebreak{}

\section{Introduction}

Representations of Boolean functions by real polynomials play a central
role in theoretical computer science. Our focus in this paper is on
\emph{approximate degree}, a particularly natural and useful complexity
measure. Formally, the $\epsilon$-approximate degree of a Boolean
function $f\colon\zoon\to\zoo$ is denoted \emph{$\deg_{\epsilon}(f)$}
and defined as the minimum degree of a real polynomial $p$ that approximates
$f$ within $\epsilon$ pointwise: $|f(x)-p(x)|\leq\epsilon$ for
all $x\in\zoon.$ The standard choice of the error parameter is $\epsilon=1/3,$
which is a largely arbitrary setting that can be replaced by any other
constant in $(0,1/2)$ without affecting the approximate degree by
more than a multiplicative constant. Since every function $f\colon\zoon\to\zoo$
can be computed with zero error by a polynomial of degree at most
$n,$ the $\epsilon$-approximate degree is always at most $n.$

The notion of approximate degree originated three decades ago in the
pioneering work of Nisan and Szegedy~\cite{nisan-szegedy94degree}
and has since proved to be a powerful tool in theoretical computer
science. Upper bounds on approximate degree have algorithmic applications,
whereas lower bounds are a staple in complexity theory. On the algorithmic
side, approximate degree underlies many of the strongest results obtained
to date in computational learning, differentially private data release,
and algorithm design in general. In complexity theory, the notion
of approximate degree has produced breakthroughs in quantum query
complexity, communication complexity, and circuit complexity. A detailed
bibliographic overview of these applications can be found in~\cite{sherstov17algopoly,bun-thaler17adeg-ac0}.

Approximate degree has been particularly prominent in the study of
$\classAC^{0},$ the class of polynomial-size constant-depth circuits
with gates $\vee,\wedge,\neg$ of unbounded fan-in. The simplest functions
in $\classAC^{0}$ are conjunctions and disjunctions, which have depth~$1$,
followed by polynomial-size CNF and DNF formulas, which have depth~$2$,
followed in turn by higher-depth circuits. Lower bounds on the approximate
degree of $\classAC^{0}$ functions have been used to settle the quantum
query complexity of Grover search~\cite{beals-et-al01quantum-by-polynomials},
element distinctness~\cite{aaronson-shi04distinctness}, and a host
of other problems~\cite{BKT17poly-strikes-back}; resolve the communication
complexity of set disjointness in the two-party quantum model~\cite{razborov02quantum,sherstov07quantum}
and number-on-the-forehead multiparty model~\cite{sherstov07ac-majmaj,sherstov07quantum,lee-shraibman08disjointness,chatt-ada08disjointness,dual-survey,beame-huyn-ngoc09multiparty-focs,sherstov12mdisj,sherstov13directional};
separate the communication complexity classes $\PP$ and $\UPP$~\cite{buhrman-dewolf01polynomials,sherstov07ac-majmaj};
and separate the polynomial hierarchy in communication complexity
from the communication class $\UPP$~\cite{RS07dc-dnf}. Despite
this array of applications and decades of study, our understanding
of the approximate degree of $\classAC^{0}$ has remained surprisingly
fragmented and incomplete. In this paper, we set out to resolve this
question in full.

In more detail, previous work on the approximate degree of $\classAC^{0}$
started with the seminal 1994 paper of Nisan and Szegedy~\cite{nisan-szegedy94degree},
who proved that the OR function on $n$ bits has approximate degree
$\Theta(\sqrt{n}).$ This was the best result until Aaronson and Shi's
celebrated lower bound of $\Omega(n^{2/3})$ for the element distinctness
problem~\cite{aaronson-shi04distinctness}. In a beautiful paper
from 2017, Bun and Thaler~\cite{bun-thaler17adeg-ac0} showed that
$\classAC^{0}$ contains functions in $n$ variables with approximate
degree $\Omega(n^{1-\delta})$, where the constant $\delta>0$ can
be made arbitrarily small at the expense of increasing the depth of
the circuit. In follow-up work, Bun and Thaler~\cite{BT18ac0-large-error}
proved an $\Omega(n^{1-\delta})$ lower bound for approximating $\classAC^{0}$
circuits even with error exponentially close to $1/2,$ where once
again the circuit depth grows with $1/\delta$. A stronger yet result
was obtained by Sherstov and Wu~\cite{sherstov-wu18sign-ac0}, who
showed that $\classAC^{0}$ has essentially the maximum possible \emph{threshold
degree} (defined as the limit of $\epsilon$-approximate degree as
$\epsilon\nearrow1/2$) and \emph{sign-rank} (a generalization of
threshold degree to arbitrary bases rather than just the basis of
monomials). Quantitatively, the authors of~\cite{sherstov-wu18sign-ac0}
proved a lower bound of $\Omega(n^{1-\delta})$ for threshold degree
and $\exp(\Omega(n^{1-\delta}))$ for sign-rank, essentially matching
the trivial upper bounds. As before, $\delta>0$ can be made arbitrarily
small at the expense of increasing the circuit depth. In particular,
$\classAC^{0}$ requires a polynomial of degree $\Omega(n^{1-\delta})$
even for approximation to error doubly (triply, quadruply, quintuply\ldots )
exponentially close to $1/2$.

The lower bounds of~\cite{bun-thaler17adeg-ac0,BT18ac0-large-error,sherstov-wu18sign-ac0}
show that $\classAC^{0}$ functions have essentially the maximum possible
complexity\textemdash but only if one is willing to look at circuits
of \emph{arbitrarily large} constant depth. What happens at \emph{small}
depths has been a wide open problem, with no techniques to address
it. Bun and Thaler observe that their $\classAC^{0}$ circuit in~\cite{bun-thaler17adeg-ac0}
with approximate degree $\Omega(n^{1-\delta})$ can be flattened to
produce a DNF formula of size $\exp(\log^{O(\log(1/\delta))}n)$,
but this is superpolynomial and thus no longer in $\classAC^{0}$.
The only progress of which we are aware is an $\Omega(n^{3/4-\delta})$
lower bound obtained for polynomial-size DNF formulas in~\cite{BKT17poly-strikes-back,mande-thaler-zhu20distinctness}.
This leaves a polynomial gap in the approximate degree for small depth
versus arbitrary constant depth. Our main contribution is to definitively
resolve the approximate degree of $\classAC^{0}$ by constructing,
for any constant $\delta>0,$ a polynomial-size DNF formula with approximate
degree $\Omega(n^{1-\delta})$. We now describe our main result and
its generalizations and applications.

\subsection{Approximate degree of DNF and CNF formulas}

Recall that a \emph{literal} is a Boolean variable $x_{1},x_{2},\ldots,x_{n}$
or its negation $\overline{x_{1}},\overline{x_{2}},\ldots,\overline{x_{n}}$.
A conjunction of literals is called a \emph{term}, and a disjunction
of literals is called a \emph{clause}. The \emph{width} of a term
or clause is the number of literals that it contains. A \emph{DNF
formula} is a disjunction of terms, and analogously a \emph{CNF formula}
is a conjunction of clauses. The \emph{width }of a DNF or CNF formula
is the maximum width of a term or clause in it, respectively. One
often refers to DNF and CNF formulas of width $k$ as $k$-DNF and
$k$-CNF formulas. The \emph{size }of a DNF or CNF formula is the
total number of terms or clauses, respectively, that it contains.
Thus, $\classAC^{0}$ circuits of depth~$1$ correspond precisely
to clauses and terms, whereas $\classAC^{0}$ circuits of depth~$2$
correspond precisely to polynomial-size DNF and CNF formulas. Our
main result on approximate degree is as follows.
\begin{thm}[Main result]
\label{thm:MAIN-dnf}Let $\delta>0$ be any constant. Then for each
$n\geq1,$ there is an $($explicitly given$)$ function $f\colon\zoon\to\zoo$
that has approximate degree
\[
\deg_{1/3}(f)=\Omega(n^{1-\delta})
\]
and is computable by a DNF formula of size $n^{O(1)}$ and width $O(1).$
\end{thm}

\noindent Theorem~\ref{thm:MAIN-dnf} almost matches the trivial
upper bound of $n$ on the approximate degree of any function. Thus,
the theorem shows that $\classAC^{0}$ circuits of depth~$2$ already
achieve essentially the maximum possible approximate degree. This
depth cannot be reduced further because $\classAC^{0}$ circuits of
depth~$1$ have approximate degree~$O(\sqrt{n}).$ Finally, the
DNF formulas constructed in Theorem~\ref{thm:MAIN-dnf} are the simplest
possible in that they have \emph{constant }width.

Recall that previously, a lower bound of $\Omega(n^{1-\delta})$ for
$\classAC^{0}$ was known only for circuits of large constant depth
that grows with $1/\delta$. The lack of progress on small-depth $\classAC^{0}$
prior to this paper had experts seriously entertaining~\cite{BT18ac0-large-error}
the possibility that $\classAC^{0}$ circuits of any given depth $d$
have approximate degree $O(n^{1-\delta_{d}})$, for some constant
$\delta_{d}=\delta_{d}(d)>0$. Such an upper bound would have far-reaching
consequences in computational learning and circuit complexity. Theorem~\ref{thm:MAIN-dnf}
rules it out.

\subsection{Large-error approximation}

Any Boolean function can be approximated pointwise within $1/2$ in
a trivial manner, by a constant polynomial. Approximation within $\frac{1}{2}-o(1),$
on the other hand, is a meaningful and extremely useful notion. We
obtain the following strengthening of our main result, in which the
approximation error is relaxed from $1/3$ to an optimal $\frac{1}{2}-\frac{1}{n^{\Theta(1)}}.$
\begin{thm}[Main result for large error]
\label{thm:MAIN-dnf-low-error}Let $\delta>0$ and $C\geq1$ be any
constants. Then for each $n\geq1,$ there is an $($explicitly given$)$
function $f\colon\zoon\to\zoo$ that has approximate degree
\[
\deg_{\frac{1}{2}-\frac{1}{n^{C}}}(f)=\Omega(n^{1-\delta})
\]
and is computable by a DNF formula of size $n^{O(1)}$ and width $O(1).$
\end{thm}

\noindent To rephrase Theorem~\ref{thm:MAIN-dnf-low-error}, polynomial-size
DNF formulas require degree $\Omega(n^{1-\delta})$ for approximation
not only to constant error but even to error $\frac{1}{2}-\frac{1}{n^{C}},$
where $C\geq1$ is an arbitrarily large constant. The error parameter
in Theorem~\ref{thm:MAIN-dnf-low-error} cannot be relaxed further
to $\frac{1}{2}-\frac{1}{n^{\omega(1)}}$ because any DNF formula
with $m$ terms can be approximated to error $\frac{1}{2}-\Omega(\frac{1}{m})$
by a polynomial of degree $O(\sqrt{n\log m})$.

Negating a function has no effect on the approximate degree. Indeed,
if $f$ is approximated to error $\epsilon$ by a polynomial $p$,
then the negated function $\neg f=1-f$ is approximated to the same
error $\epsilon$ by the polynomial $1-p.$ With this observation,
Theorems~\ref{thm:MAIN-dnf} and~\ref{thm:MAIN-dnf-low-error} carry
over to CNF formulas:
\begin{cor}
\label{cor:MAIN-cnf}Let $\delta>0$ and $C\geq1$ be any constants.
Then for each $n\geq1,$ there is an $($explicitly given$)$ function
$g\colon\zoon\to\zoo$ that has approximate degree
\[
\deg_{\frac{1}{2}-\frac{1}{n^{C}}}(g)=\Omega(n^{1-\delta})
\]
and is computable by a CNF formula of size $n^{O(1)}$ and width $O(1).$
\end{cor}

\subsection{One-sided approximation}

There is a natural notion of \emph{one-sided} approximation for Boolean
functions. Specifically, the \emph{one-sided $\epsilon$-approximate
degree }of a function $f\colon\zoon\to\zoo$ is defined as the minimum
degree of a real polynomial $p$ such that
\begin{align*}
f(x)=0 & \qquad\Rightarrow\qquad p(x)\in[-\epsilon,\epsilon],\\
f(x)=1 & \qquad\Rightarrow\qquad p(x)\in[1-\epsilon,+\infty)
\end{align*}
for every $x\in\zoon.$ This complexity measure is denoted $\onedeg_{\epsilon}(f)$.
It plays a considerable role~\cite{GS09npconp,bun-thaler13and-or-tree,sherstov13and-or,bun-thaler13amplification,sherstov14sign-deg-ac0,sherstov12mdisj,sherstov13directional}
in the area, both in its own right and due to its applications to
other asymmetric notions of computation such as nondeterminism and
Merlin\textendash Arthur protocols. One-sided approximation is meaningful
for any error parameter $\epsilon\in[0,1/2),$ and as before the standard
setting is $\epsilon=1/3$. By definition, one-sided approximate degree
is always at most $n$. Observe that the definitions of $\epsilon$-approximate
degree $\deg_{\epsilon}(f)$ and its one-sided variant $\onedeg_{\epsilon}(f)$
impose the same requirement for inputs $x\in f^{-1}(0)$: the approximating
polynomial must approximate $f$ within $\epsilon$ at each such $x$.
For inputs $x\in f^{-1}(1),$ on the other hand, the definitions of
$\deg_{\epsilon}(f)$ and $\onedeg_{\epsilon}(f)$ diverge dramatically,
with one-sided $\epsilon$-approximate degree not requiring any upper
bound on the approximating polynomial $p$. As a result, one always
has $\onedeg_{\epsilon}(f)\leq\deg_{\epsilon}(f)$, and it is reasonable
to expect a large gap between the two quantities for some $f$. Moreover,
the one-sided approximate degree of a function is in general not equal
to that of its negation: $\onedeg_{\epsilon}(f)\ne\onedeg_{\epsilon}(\neg f).$
This contrasts with the equality $\deg_{\epsilon}(f)=\deg_{\epsilon}(\neg f)$
for two-sided approximation.

In this light, there are three particularly natural questions to ask
about one-sided approximate degree:
\begin{enumerate}
\item \label{enu:problem-onedeg-AC0}What is the one-sided approximate degree
of $\classAC^{0}$ circuits?
\item \label{enu:problem-onedeg-deg}What is the largest possible gap between
approximate degree and one-sided approximate degree?
\item \label{enu:problem-onedeg-f-onedeg-neg-f}What is the largest possible
gap between the one-sided approximate degree of a function $f$ and
that of its negation $\neg f$?
\end{enumerate}
In this paper, we resolve all three questions in detail. For question~\ref{enu:problem-onedeg-AC0},
we prove that polynomial-size CNF formulas achieve essentially the
maximum possible one-sided approximate degree. In fact, our result
holds even for approximation to error vanishingly close to random
guessing, $\frac{1}{2}-o(1)$:
\begin{thm}
\label{thm:MAIN-one-sided}Let $\delta>0$ and $C\geq1$ be any constants.
Then for each $n\geq1,$ there is an $($explicitly given$)$ function
$g\colon\zoon\to\zoo$ that has one-sided approximate degree
\[
\onedeg_{\frac{1}{2}-\frac{1}{n^{C}}}(g)=\Omega(n^{1-\delta})
\]
and is computable by a CNF formula of size $n^{O(1)}$ and width $O(1).$
\end{thm}

\noindent Theorem~\ref{thm:MAIN-one-sided} essentially settles the
one-sided approximate degree of $\classAC^{0}$. The theorem is optimal
with respect to circuit depth; recall that depth-$1$ circuits have
approximate degree $O(\sqrt{n})$ and hence also one-sided approximate
degree $O(\sqrt{n})$. Previous work on the one-sided approximate
degree of $\classAC^{0}$ was suboptimal with respect to the degree
bound and/or circuit depth. Specifically, the best previous lower
bounds were $\Omega(n/\log n)^{2/3}$ due to Bun and Thaler~\cite{bun-thaler13amplification}
for a polynomial-size CNF formula, and $\Omega(n^{1-\delta})$ due
to Sherstov and Wu~\cite{sherstov-wu18sign-ac0} for $\classAC^{0}$
circuits of depth that grows with $1/\delta.$

As an application of Theorem~\ref{thm:MAIN-one-sided}, we resolve
questions~\ref{enu:problem-onedeg-deg} and~\ref{enu:problem-onedeg-f-onedeg-neg-f}
in full, establishing a gap of $O(1)$ versus $\Omega(n^{1-\delta})$
in each case. Moreover, we prove that these gaps remain valid well
beyond the standard error regime of $\epsilon=1/3$. A detailed statement
of our separations follows.
\begin{cor}
\label{cor:deg-onedeg-onedeg-neg-separations} Let $\delta>0$ and
$C\geq1$ be any constants. Then for each $n\geq1,$ there is an $($explicitly
given$)$ function $f\colon\zoon\to\zoo$ with
\begin{equation}
\onedeg_{0}(f)=O(1)\label{eq:onedeg-g}
\end{equation}
but 
\begin{align}
 & \deg_{\frac{1}{2}-\frac{1}{n^{C}}}(f)=\Omega(n^{1-\delta}),\label{eq:deg-g}\\
 & \onedeg_{\frac{1}{2}-\frac{1}{n^{C}}}(\neg f)=\Omega(n^{1-\delta}).\label{eq:onedeg-neg-g}
\end{align}
Moreover, $f$ is computable by a DNF formula of size $n^{O(1)}$
and width $O(1).$
\end{cor}

\noindent Equations~(\ref{eq:onedeg-g}) and~(\ref{eq:deg-g}) in
this result give the promised $O(1)$ versus $\Omega(n^{1-\delta})$
separation for question~\ref{enu:problem-onedeg-deg}. Analogously,
(\ref{eq:onedeg-g}) and~(\ref{eq:onedeg-neg-g}) give an $O(1)$
versus $\Omega(n^{1-\delta})$ separation for question~\ref{enu:problem-onedeg-f-onedeg-neg-f}.
Of particular note in both separations is the error regime: the upper
bound remains valid even under the stronger requirement of zero error,
whereas the lower bounds remain valid even under the weaker requirement
of error $\frac{1}{2}-o(1)$. Our separations improve on previous
work. For question~\ref{enu:problem-onedeg-deg}, the best previous
separation was $(\log n)^{O_{\delta}(1)}$ versus $\Omega(n^{1-\delta})$
for any fixed $\delta>0$, implicit in~\cite{bun-thaler17adeg-ac0}.
For the harder question~\ref{enu:problem-onedeg-f-onedeg-neg-f},
the best previous separation~\cite{bun-thaler13amplification} was
$O(\log n)$ versus $\Omega(n/\log n)^{2/3}$, which is polynomially
weaker than ours.

The derivation of Corollary~\ref{cor:deg-onedeg-onedeg-neg-separations}
from Theorem~\ref{thm:MAIN-one-sided} is short and illustrative,
and we include it here.
\begin{proof}[Proof of Corollary~\emph{\ref{cor:deg-onedeg-onedeg-neg-separations}.}]
 Let $g$ be the function from Theorem~\ref{thm:MAIN-one-sided},
and set $f=\neg g$. Then~(\ref{eq:onedeg-neg-g}) is immediate.
Equation~(\ref{eq:deg-g}) follows from~(\ref{eq:onedeg-neg-g})
in light of the basic relations $\deg_{\epsilon}(f)=\deg_{\epsilon}(\neg f)\geq\onedeg_{\epsilon}(\neg f)$,
valid for all $f$ and $\epsilon.$ Finally,~(\ref{eq:onedeg-g})
can be seen as follows. Since $g$ is a CNF formula of width $O(1),$
its negation $f$ is a DNF formula of width $O(1).$ Thus, every term
of $f$ can be represented exactly by a polynomial of degree $O(1)$.
Summing these polynomials gives a $0$-error one-sided approximant
for $f.$
\end{proof}
We now discuss applications of our results on approximate degree and
one-sided approximate degree to fundamental questions in communication
complexity. 

\subsection{Randomized multiparty communication}

We adopt the \emph{number-on-the-forehead} model of Chandra, Furst,
and Lipton~\cite{cfl83multiparty}, which is the most powerful formalism
of multiparty communication. The model features $k$ communicating
players and a Boolean function $F\colon X_{1}\times X_{2}\times\cdots\times X_{k}\to\zoo$
with $k$ arguments. An input $(x_{1},x_{2},\dots,x_{k})$ is distributed
among the $k$ players by giving the $i$-th player the arguments
$x_{1},\dots,x_{i-1},x_{i+1},\dots,x_{k}$ but not $x_{i}$. This
arrangement can be visualized as having the $k$ players seated in
a circle with $x_{i}$ written on the $i$-th player's forehead, whence
the name of the model. Number-on-the-forehead is the canonical model
in the area because any other way of assigning arguments to players
results in a less powerful model\textemdash provided of course that
one does not assign all the arguments to some player, in which case
there is never a need to communicate.

The players communicate according to a protocol agreed upon in advance.
The communication occurs in the form of broadcasts, with a message
sent by any given player instantly reaching everyone else. The players'
objective is to compute $F$ on any given input with minimal communication.
To this end, the players have access to an unbounded supply of shared
random bits which they can use in deciding what message to send at
any given point in the protocol. The\emph{ cost }of a protocol is
the total bit length of all the messages broadcast in a worst-case
execution. The\emph{ $\epsilon$-error randomized communication complexity
$R_{\epsilon}(F)$} of a given function $F$ is the least cost of
a protocol that computes $F$ with probability of error at most $\epsilon$
on every input. As with approximate degree, the standard setting of
the error parameter is $\epsilon=1/3.$

The number-on-the-forehead communication complexity of constant-depth
circuits is a challenging question that has been the focus of extensive
research, e.g.,~\cite{beame06corruption,lee-shraibman08disjointness,chatt-ada08disjointness,dual-survey,beame-huyn-ngoc09multiparty-focs,sherstov12mdisj,sherstov13directional,bun-thaler17adeg-ac0}.
In contrast to the two-party model, where a lower bound of $\Omega(\sqrt{n})$
for $\classAC^{0}$ circuits is straightforward to prove from first
principles~\cite{BFS86cc}, the first $n^{\Omega(1)}$ multiparty
lower bound~\cite{sherstov12mdisj} for $\classAC^{0}$ was obtained
only in~2012. The strongest known multiparty lower bounds for $\classAC^{0}$
are obtained using the \emph{pattern matrix method} of~\cite{sherstov13directional},
which transforms approximate degree lower bounds in a black-box manner
into communication lower bounds. In the most recent application of
this method, Bun and Thaler~\cite{bun-thaler17adeg-ac0} gave a $k$-party
communication problem $F\colon(\zoo^{n})^{k}\to\zoo$ in $\classAC^{0}$
with communication complexity $\Omega(n/4^{k}k^{2})^{1-\delta},$
where the constant $\delta>0$ can be taken arbitrarily small at the
expense of increasing the depth of the $\classAC^{0}$ circuit. This
shows that $\classAC^{0}$ has essentially the maximum possible multiparty
communication complexity\textemdash as long as one is willing to use
circuits of arbitrarily large constant depth. For circuits of small
depth, the best lower bound is polynomially weaker: $\Omega(n/4^{k}k^{2})^{3/4-\delta}$
for the $k$-party communication complexity of polynomial-size DNF
formulas, which can be proved by applying the pattern matrix method
to the approximate degree lower bounds in~\cite{BKT17poly-strikes-back,mande-thaler-zhu20distinctness}.
This fragmented state of the art closely parallels that for approximate
degree prior to our work.

We resolve the multiparty communication complexity of $\classAC^{0}$
in detail in the following theorem.
\begin{thm}
\label{thm:MAIN-multiparty}Fix any constants $\delta\in(0,1]$ and
$C\geq1$. Then for all integers $n,k\geq2,$ there is an $($explicitly
given$)$ $k$-party communication problem $F_{n,k}\colon(\zoon)^{k}\to\zoo$
with
\begin{align*}
 & R_{1/3}(F_{n,k})\geq\left(\frac{n}{c'4^{k}k^{2}}\right)^{1-\delta},\\
 & R_{\frac{1}{2}-\frac{1}{n^{C}}}(F_{n,k})\geq\frac{n^{1-\delta}}{c'4^{k}},
\end{align*}
where $c'\geq1$ is a constant independent of $n$ and $k.$ Moreover,
$F_{n,k}$ is computable by a DNF formula of size $n^{c'}$ and width
$c'k$.
\end{thm}

\noindent Theorem~\ref{thm:MAIN-multiparty} essentially represents
the state of the art for multiparty communication lower bounds. Indeed,
the best communication lower bound to date for any explicit function
$F\colon(\zoon)^{k}\to\zoo,$ whether or not $F$ is computable by
an $\classAC^{0}$ circuit, is $\Omega(n/2^{k})$~\cite{bns92}.
Theorem~\ref{thm:MAIN-multiparty} comes close to matching the trivial
upper bound of $n+1$ for any communication problem, thereby showing
that $\classAC^{0}$ circuits of depth~$2$ achieve nearly the maximum
possible communication complexity. Moreover, our result holds not
only for bounded-error communication but also for communication with
error $\frac{1}{2}-\frac{1}{n^{C}}$ for any $C\geq1.$ The error
parameter in Theorem~\ref{thm:MAIN-multiparty} is optimal and cannot
be further increased to $\frac{1}{2}-\frac{1}{n^{\omega(1)}}$; indeed,
it is straightforward to see that any DNF formula with $m$ terms
has a communication protocol with error $\frac{1}{2}-\Omega(\frac{1}{m})$
and cost $2$~bits. Theorem~\ref{thm:MAIN-multiparty} is also optimal
with respect to circuit depth because the multiparty communication
complexity of $\classAC^{0}$ circuits of depth~$1$ is at most $2$
bits.

Since randomized communication complexity is invariant under function
negation, Theorem~\ref{thm:MAIN-multiparty} remains valid with the
word ``DNF'' replaced with ``CNF.''

\subsection{Nondeterministic and Merlin\textendash Arthur multiparty communication}

Here again, we adopt the $k$-party number-on-the-forehead model of
Chandra, Furst, and Lipton~\cite{cfl83multiparty}. Nondeterministic
communication is defined in complete analogy with computational complexity.
Specifically, a nondeterministic protocol starts with a guess string,
whose length counts toward the protocol's communication cost, and
proceeds deterministically thenceforth. A nondeterministic protocol
for a given communication problem $F\colon X_{1}\times X_{2}\times\cdots\times X_{k}\to\zoo$
is required to output the correct answer for \emph{all} guess strings
when presented with a negative instance of $F,$ and for \emph{some}
guess string when presented with a positive instance. We further consider
\emph{Merlin\textendash Arthur protocols}~\cite{babai85arthur-merlin,bm88am},
a communication model that combines the power of randomization and
nondeterminism. As before, a Merlin\textendash Arthur protocol for
a given problem $F$ starts with a guess string, whose length counts
toward the communication cost. From then on, the parties run an ordinary
randomized protocol. The randomized phase in a Merlin\textendash Arthur
protocol must produce the correct answer with probability at least
$2/3$ for \emph{all} guess strings when presented with a negative
instance of $F,$ and for \emph{some} guess string when presented
with a positive instance. Thus, the cost of a nondeterministic or
Merlin\textendash Arthur protocol is the sum of the costs of the guessing
phase and communication phase. The minimum cost of a valid protocol
for $F$ in these models is called the \emph{nondeterministic communication
complexity of $F$}, denoted $N(F)$, and \emph{Merlin\textendash Arthur
communication complexity of $F,$} denoted $\ma_{1/3}(F)$. The quantity
$N(\neg F)$ is called the \emph{co-nondeterministic communication
complexity of $F$.}

Nondeterministic and Merlin\textendash Arthur protocols have been
extensively studied for $\mbox{\ensuremath{k=2}}$ parties but are
much less understood in the multiparty setting~\cite{BDPW10d-bpp,GS09npconp,sherstov12mdisj,sherstov13directional}.
Prior to our paper, the best lower bounds in these models for an $\classAC^{0}$
circuit $F\colon(\zoon)^{k}\to\zoo$ were $\Omega(\sqrt{n}/2^{k}k)$
for nondeterministic communication and $\Omega(\sqrt{n}/2^{k}k)^{1/2}$
for Merlin\textendash Arthur communication, obtained in~\cite{sherstov13directional}
for the set disjointness problem. We give a quadratic improvement
on these lower bounds. In particular, our result for nondeterminism
essentially matches the trivial upper bound. Moreover, we obtain our
result for a particularly simple function in $\classAC^{0}$, namely,
a polynomial-size CNF formula of constant width. A detailed statement
follows.
\begin{thm}
\label{thm:MAIN-N-MA} Let $\delta>0$ be arbitrary. Then for all
integers $n,k\geq2,$ there is an $($explicitly given$)$ $k$-party
communication problem $G_{n,k}\colon(\zoon)^{k}\to\zoo$ with
\[
N(\neg G_{n,k})\leq c\log n
\]
 but
\begin{align}
N(G_{n,k}) & \geq\left(\frac{n}{c4^{k}k^{2}}\right)^{1-\delta},\label{eq:N-dnf-1}\\
R_{1/3}(G_{n,k}) & \geq\left(\frac{n}{c4^{k}k^{2}}\right)^{1-\delta},\label{eq:N-cnf-1}\\
\ma_{1/3}(G_{n,k}) & \geq\left(\frac{n}{c4^{k}k^{2}}\right)^{\frac{1-\delta}{2}},\label{eq:ma-cnf-1}
\end{align}
where $c\geq1$ is a constant independent of $n$ and $k.$ Moreover,
$G_{n,k}$ is computable by a CNF formula of width $ck$ and size
$n^{c}$.
\end{thm}

\noindent This result can be viewed as a far-reaching generalization
of Theorem~\ref{thm:MAIN-multiparty} to nondeterministic and Merlin\textendash Arthur
protocols. To obtain Theorem~\ref{thm:MAIN-N-MA}, we adapt the pattern
matrix method~\cite{sherstov13directional} to be able to transform
any lower bound on one-sided approximate degree into a multiparty
communication lower bound in the nondeterministic and Merlin\textendash Arthur
models. With this tool in hand, we obtain Theorem~\ref{thm:MAIN-N-MA}
from our one-sided approximate degree lower bound (Theorem~\ref{thm:MAIN-one-sided}).

\subsection{Multiparty communication classes}

Theorem~\ref{thm:MAIN-N-MA} sheds new light on communication complexity
classes, defined in the seminal work of Babai, Frankl, and Simon~\cite{BFS86cc}.
An infinite family $\{F_{n}\}_{n=1}^{\infty},$ where each $F_{n}\colon(\zoon)^{k}\to\zoo$
is a $k$-party number-on-the-forehead communication problem, is said
to be \emph{efficiently solvable} in a given model of communication
if $F_{n}$ has communication complexity at most $\log^{c}n$ in that
model, for a large enough constant $c>1$ and all $n>c.$ One defines
$\BPP_{k},$ $\NP_{k},$ $\coNP_{k},$ and $\MA_{k}$ as the classes
of families that are efficiently solvable in the randomized, nondeterministic,
co-nondeterministic, and Merlin\textendash Arthur models, respectively.
In particular, $\MA_{k}$ is a superset of $\NP_{k}$ and $\BPP_{k}$.
In these definitions, $k=k(n)$ can be any function of $n,$ including
constant functions such as $k=3.$ The relations among these multiparty
classes have been actively studied over the past decade~\cite{BDPW07p-rp,lee-shraibman08disjointness,chatt-ada08disjointness,david-pitassi-viola08bpp-np,beame-huyn-ngoc09multiparty-focs,BDPW10d-bpp,GS09npconp,sherstov12mdisj,sherstov13directional}.
It particular, for $k\leq\Theta(\log n),$ it is known that $\coNP_{k}$
is not contained in $\BPP_{k},$ $\NP_{k},$ or even $\MA_{k}$. Quantitatively,
these results can be summarized as follows. \sloppy
\begin{enumerate}[itemsep=1mm]
\item  Prior to our work, the strongest $k$-party separation of co-nondeterministic
versus randomized communication complexity was $O(\log n)$ versus
$\Omega(\sqrt{n}/2^{k}k)$, proved in~\cite{sherstov13directional}
for the set disjointness function.
\item The best previous $k$-party separations of co-nondeterministic versus
nondeterministic communication complexity were: $O(\log n)$ versus
$\Omega(n),$ proved~in~\cite{sherstov13directional} nonconstructively
by the probabilistic method; and $O(\log n)$ versus $\Omega(\sqrt{n}/2^{k}k)$,
proved in~\cite{sherstov13directional} for the set disjointness
problem.
\item The best previous $k$-party separation of co-nondeterministic versus
Merlin\textendash Arthur communication complexity was $O(\log n)$
versus $\Omega(\sqrt{n}/2^{k}k)^{1/2}$, proved in~\cite{sherstov13directional}
for the set disjointness problem.
\end{enumerate}
Theorem~\ref{thm:MAIN-N-MA} gives a quadratic improvement on these
previous separations, excluding the nonconstructive separation of
$\coNP_{k}$ from $\NP_{k}$ in~\cite{BDPW10d-bpp}. Moreover, our
quadratically improved separations are achieved for a particularly
simple function, namely, the polynomial-size constant-width CNF formula
$G_{n,k}$. In the regime $k\leq\Theta(\log n),$ our separations
of $\coNP_{k}$ from $\BPP_{k}$ and $\NP_{k}$ are essentially optimal,
and our separation of $\coNP_{k}$ from $\MA_{k}$ is within a square
of optimal. Recall that no explicit lower bounds at all are currently
known in the regime $k\geq\log n,$ even for deterministic communication.
We state our contributions for communication complexity classes as
a corollary below.
\begin{cor}
Let $k=k(n)$ be a function with $k(n)\leq(\frac{1}{2}-\epsilon)\log n$
for some constant $\epsilon>0$.\emph{ }Then the communication problem
$G_{n,k}$ from Theorem~\emph{\ref{thm:MAIN-N-MA}} satisfies
\begin{align*}
 & \{G_{n,k}\}_{n=1}^{\infty}\in\coNP_{k}\setminus\BPP_{k},\\
 & \{G_{n,k}\}_{n=1}^{\infty}\in\coNP_{k}\setminus\NP_{k},\\
 & \{G_{n,k}\}_{n=1}^{\infty}\in\coNP_{k}\setminus\MA_{k}.
\end{align*}
Analogously, the communication problem $F_{n,k}$ from Theorem~\emph{\ref{thm:MAIN-multiparty}}
satisfies 
\[
\{F_{n,k}\}_{n=1}^{\infty}\in\NP_{k}\setminus\BPP_{k}.
\]
\end{cor}

\begin{proof}
The claims for $G_{n,k}$ are immediate from Theorem~\ref{thm:MAIN-N-MA}
and the definitions of $\NP_{k},\coNP_{k},\BPP_{k},\MA_{k}.$ For
the remaining separation, we need only prove the upper bound $N(F_{n,k})=O(\log n).$
Recall from Theorem~\ref{thm:MAIN-multiparty} that $F_{n,k}$ is
a DNF formula with $n^{c'}$ terms. This gives the desired nondeterministic
protocol: the parties ``guess'' one of the terms in $F_{n,k}$ (for
a cost of $\lceil\log n^{c'}\rceil$ bits), evaluate it (using another
$2$ bits of communication), and output the result.
\end{proof}

\subsection{Quantum communication complexity}

We adopt the standard model of quantum communication, where two parties
exchange quantum messages according to an agreed-upon protocol in
order to solve a two-party communication problem $F\colon X\times Y\to\zoo$.
As usual, an input $(x,y)\in X\times Y$ is split between the parties,
with one party knowing only $x$ and the other party knowing only
$y.$ We allow arbitrary prior entanglement at the start of the communication.
A measurement at the end of the protocol produces a single-bit answer,
which is interpreted as the protocol output. An\emph{ $\epsilon$-error
protocol for $F$} is required to output, on every input $(x,y)\in X\times Y,$
the correct value $F(x,y)$ with probability at least $1-\epsilon.$
The \emph{cost} of a quantum protocol is the total number of quantum
bits exchanged in the worst case on any input. The \emph{$\epsilon$-error
quantum communication complexity of $F$}, denoted $Q_{\epsilon}^{*}(F),$
is the least cost of an $\epsilon$-error quantum protocol for $F.$
The asterisk in $Q_{\epsilon}^{*}(F)$ indicates that the parties
share arbitrary prior entanglement. The standard setting of the error
parameter is $\epsilon=1/3,$ which is as usual without loss of generality.
For a detailed formal description of the quantum model, we refer the
reader to~\cite{dewolf-thesis,razborov02quantum,sherstov07quantum}.

Proving lower bounds for bounded-error quantum communication is significantly
more challenging than for randomized communication. An illustrative
example is the set disjointness problem on $n$ bits. Babai, Frankl,
and Simon~\cite{BFS86cc} obtained an $\Omega(\sqrt{n})$ randomized
communication lower bound for this function in 1986 using a short
and elementary proof, which was later improved to a tight $\Omega(n)$
in~\cite{KS92disj,razborov90disj,baryossef04info-complexity}. This
is in stark contrast with the quantum model, where the best lower
bound for set disjointness was for a long time a trivial $\Omega(\log n)$
until a tight $\Omega(\sqrt{n})$ was proved by Razborov~\cite{razborov02quantum}
in 2002.

A completely different proof of the $\Omega(\sqrt{n})$ lower bound
for set disjointness was given in~\cite{sherstov07quantum} by introducing
the pattern matrix method. Since then, the method has produced the
strongest known quantum lower bounds for $\classAC^{0}$. Of these,
the best lower bound prior to our work was $\Omega(n^{1-\delta})$
due to Bun and Thaler~\cite{bun-thaler17adeg-ac0}, where the constant
$\delta>0$ can be taken arbitrarily small at the expense of circuit
depth. In the following theorem, we resolve the quantum communication
complexity of $\classAC^{0}$ in full by proving that polynomial-size
DNF formulas achieve near-maximum communication complexity.
\begin{thm}
\label{thm:MAIN-quantum}Let $\delta>0$ and $C\geq1$ be any constants.
Then for each $n\geq1,$ there is an $($explicitly given$)$ two-party
communication problem $F\colon\zoo^{n}\times\zoo^{n}\to\zoo$ that
has quantum communication complexity
\[
Q_{\frac{1}{2}-\frac{1}{n^{C}}}^{*}(F)=\Omega(n^{1-\delta})
\]
and is representable by a DNF formula of size $n^{O(1)}$ and width
$O(1).$
\end{thm}

\noindent This theorem remains valid for CNF formulas since quantum
communication complexity is invariant under function negation. As
in all of our results, Theorem~\ref{thm:MAIN-quantum} essentially
matches the trivial upper bound, showing that $\classAC^{0}$ circuits
of depth~$2$ achieve nearly the maximum possible complexity. Again
analogous to our other results, Theorem~\ref{thm:MAIN-quantum} holds
not only for bounded-error communication but also for communication
with error $\frac{1}{2}-\frac{1}{n^{C}}$ for any $C\geq1.$ The error
parameter in Theorem~\ref{thm:MAIN-quantum} is optimal and cannot
be further increased to $\frac{1}{2}-\frac{1}{n^{\omega(1)}}$: as
remarked above, any DNF formula with $m$ terms has a classical communication
protocol with error $\frac{1}{2}-\Omega(\frac{1}{m})$ and cost $2$~bits.
Lastly, Theorem~\ref{thm:MAIN-quantum} is optimal with respect to
circuit depth because $\classAC^{0}$ circuits of depth~$1$ have
communication complexity at most $2$ bits even in the classical deterministic
model.

In our overview so far, we have separately considered the classical
multiparty model and the quantum two-party model. By combining the
features of these models, one arrives at the $k$-party number-on-the-forehead
model with quantum players. Our results readily generalize to this
setting. Specifically, for any constants $\delta>0$ and $C\geq1,$
we give an explicit DNF formula $F_{n,k}\colon(\zoon)^{k}\to\zoo$
of size $n^{O(1)}$ and width $O(k)$ such that computing $F_{n,k}$
in the $k$-party quantum number-on-the-forehead model with error
$\frac{1}{2}-\frac{1}{n^{C}}$ requires $\Omega(n^{1-\delta}/4^{k}k)$
quantum bits. For more details, see Remark~\ref{rem:quantum-NOF}.

\subsection{Previous approaches}

In the remainder of the introduction, we sketch our proof of Theorem~\ref{thm:MAIN-dnf}.
To properly set the stage for our work, we start by reviewing the
relevant background and previous approaches. The notation that we
adopt below is standard, and we defer its formal review to Section~\ref{sec:Preliminaries}.

\subsubsection*{Dual view of approximation}

Let $f\colon X\to\zoo$ be a Boolean function of interest, where $X$
is an arbitrary finite subset of Euclidean space. The approximate
degree of $f$ is defined analogously to functions on the Boolean
hypercube: $\deg_{\epsilon}(f)$ is the minimum degree of a real polynomial
$p$ such that $|f(x)-p(x)|\leq\epsilon$ for every $x\in X.$ A valuable
tool in the analysis of approximate degree is linear programming duality,
which gives a powerful \emph{dual} view of approximation~\cite{sherstov07quantum}.
This dual characterization states that $\degeps(f)\geq d$ if and
only if there is a function $\phi\colon X\to\Re$ with the following
two properties: $\langle\phi,f\rangle>\epsilon\|\phi\|_{1}$; and
$\langle\phi,p\rangle=0$ for every polynomial $p$ of degree less
than $d$. Rephrasing, $\phi$ must be correlated with $f$ but completely
uncorrelated with any polynomial of degree less than $d.$ Such a
function $\phi$ is variously referred to in the literature as a ``dual
object,'' ``dual polynomial,'' or ``witness'' for $f.$ The dual
characterization makes it possible to prove any approximate degree
lower bound by constructing the corresponding witness $\phi.$ This
good news comes with a caveat: for all but the simplest functions,
the construction of $\phi$ is very demanding, and linear programming
duality gives no guidance in this regard.

\subsubsection*{Componentwise composition}

The construction of a dual object is more approachable for composed
functions since one can hope to break them up into constituent parts,
construct a dual object for each, and recombine these results. Formally,
define the \emph{componentwise composition} of functions $f\colon\zoon\to\zoo$
and $g\colon X\to\zoo$ as the Boolean function $f\circ g\colon X^{n}\to\zoo$
given by $\text{(}f\circ g)(x_{1},\ldots,x_{n})=f(g(x_{1}),\ldots,g(x_{n})).$
To construct a dual object for $f\circ g,$ one starts by obtaining
dual objects $\phi$ and $\psi$ for the constituent functions $f$
and $g$, respectively, either by direct construction or by appeal
to linear programming duality. They are then combined to yield a dual
object $\Phi$ for the composed function, using \emph{dual componentwise
composition~\cite{sherstov09hshs,lee09formulas}}:
\begin{equation}
\!\!\!\Phi(x_{1},x_{2},\ldots,x_{n})=\phi(\I[\psi(x_{1})>0],\ldots,\I[\psi(x_{n})>0])\prod_{i=1}^{n}|\psi(x_{i})|.\label{eq:dual-block-compose}
\end{equation}
This composed dual object typically requires additional work to ensure
strong enough correlation with the composed function $f\circ g$.
Among the generic tools available to assist in this process is a ``corrector''
object $\zeta$ due to Razborov and Sherstov~\cite{RS07dc-dnf},
with the following four properties: (i)~$\zeta$ is orthogonal to
low-degree polynomials; (ii)~$\zeta$ takes on~$1$ at a prescribed
point of the hypercube; (iii)~$\zeta$ is bounded at inputs of low
Hamming weight; and (iv)~$\zeta$ vanishes at all other points of
the hypercube. Using $\zeta$, suitably shifted and scaled, one can
surgically correct the behavior of a given dual object $\Phi$ at
a substantial fraction of the inputs without affecting $\Phi$'s orthogonality
to low-degree polynomials. This technique played an important role
in previous work, e.g.,~\cite{bun-thaler17adeg-ac0,BKT17poly-strikes-back,BT18ac0-large-error,sherstov-wu18sign-ac0}.

Componentwise composition by itself does not allow one to construct
hard-to-approximate functions from easy ones. To see why, consider
arbitrary functions $f\colon\zoo^{n_{1}}\to\zoo$ and $g\colon\zoo^{n_{2}}\to\zoo$
with approximate degrees at most $n_{1}^{\alpha}$ and $n_{2}^{\alpha},$
respectively, for some $0<\alpha<1$. It is well-known~\cite{sherstov12noisy}
that the composed function $f\circ g$ on $n_{1}n_{2}$ variables
has approximate degree $O(n_{1}^{\alpha}n_{2}^{\alpha})=O(n_{1}n_{2})^{\alpha}.$
This means that relative to the new number of variables, the composed
function $f\circ g$ is asymptotically no harder to approximate than
the constituent functions $f$ and $g$. In particular, one cannot
use componentwise composition to transform functions on $n$ bits
with $1/3$-approximate degree at most $n^{\alpha}$ into functions
on $N$ bits with $1/3$-approximate degree $\omega(N^{\alpha}).$

\subsubsection*{Previous best bound for $\classAC^{0}$}

In the previous best result on the $1/3$-approximate degree of $\classAC^{0}$,
Bun and Thaler~\cite{bun-thaler17adeg-ac0} approached the componentwise
composition $f\circ g$ in an ingenious way to amplify the approximate
degree for a careful choice of $g$. Let $f\colon\zoon\to\zoo$ be
given, with $1/3$-approximate degree $n^{\alpha}$ for some $0\leq\alpha<1$.
Bun and Thaler consider the componentwise composition $F=f\circ(\AND_{\Theta(\log m)}\circ\OR_{m})$,
for a small enough parameter $m=\poly(n).$ It was shown in earlier
work~\cite{sherstov09hshs,bun-thaler13amplification} that dual componentwise
composition witnesses the lower bound $\deg_{1/3}(F)=\Omega(\deg_{1/3}(\OR_{m})\deg_{1/3}(f))=\Omega(\sqrt{m}\deg_{1/3}(f)).$
Bun and Thaler make the crucial observation that the dual object for
$\OR_{m}$ has most of its $\ell_{1}$ mass on inputs of Hamming weight
$O(1)$, which in view of~(\ref{eq:dual-block-compose}) implies
that the dual object for $F$ places most of its $\ell_{1}$ mass
on inputs of Hamming weight $\tilde{O}(n).$ The authors of~\cite{bun-thaler17adeg-ac0}
then use the Razborov\textendash Sherstov corrector object to transfer
the small amount of $\ell_{1}$ mass that the dual object for $F$
places on inputs of high Hamming weight, to inputs of low Hamming
weight. The resulting dual object is supported entirely on inputs
of low Hamming weight and therefore witnesses a lower bound on the
approximate degree of the \emph{restriction} $F'$ of $F$ to inputs
of low Hamming weight.

The restriction $F'$ takes as input $N:=\Theta(nm\log m)$ variables
but is defined only when its input string has Hamming weight $\tilde{O}(n).$
This makes it possible to represent the input to $F'$ more economically,
by specifying the locations of the $\tilde{O}(n)$ nonzero bits inside
the array of $N$ variables. Since each such location can be specified
using $\lceil\log N\rceil$ bits, the entire input to $F'$ can be
specified using $\lceil\log N\rceil\cdot\tilde{O}(n)=\tilde{O}(n)$
bits. This yields a function $F''$ on $\tilde{O}(n)$ variables.
A careful calculation shows that this ``input compression'' does
not hurt the approximate degree. Thus, the approximate degree of $F''$
is at least the approximate degree of $F',$ which as discussed above
is $\Omega(\sqrt{m}\deg_{1/3}(f)).$ With $m$ set appropriately,
the approximate degree of $F''$ is polynomially larger than that
of $f.$ 

This passage from $f$ to $F''$ is the desired hardness amplification
for approximate degree. To obtain an $\Omega(n^{1-\delta})$ lower
bound on the approximate degree of $\classAC^{0}$, the authors of~\cite{bun-thaler17adeg-ac0}
start with a trivial circuit and apply the hardness amplification
step a constant number of times, until approximate degree $\Omega(n^{1-\delta})$
is reached.

\subsubsection*{Limitations of previous approaches to $\classAC^{0}$}

Bun and Thaler's hardness amplification for approximate degree rests
on two pillars. The first is componentwise composition, whereby the
given function $f\colon\zoon\to\zoo$ is composed componentwise with
$n$ independent copies of the gadget $\AND_{\Theta(\log m)}\circ\OR_{m}.$
In this gadget, the $\AND_{\Theta(\log m)}$ gate is necessary to
control the accumulation of error and to ensure the correlation property
of the dual polynomial. The resulting composed function $F=f\circ(\AND_{\Theta(\log m)}\circ\OR_{m})$
is defined on $N=\Theta(nm\log m$) variables. The second pillar of~\cite{bun-thaler17adeg-ac0}
is input compression, where the length-$N$ input to $F$ is represented
compactly as an array of $\tilde{O}(n)$ strings of length $\lceil\log N\rceil$
each. The circuitry to implement these two pillars is expensive, requiring
in both cases a polynomial-size DNF formula of width $\Theta(\log n+\log m)$.
As a result, even a \emph{single} iteration of the Bun\textendash Thaler
hardness amplification cannot be implemented as a polynomial-size
DNF or CNF formula.

To prove an $\Omega(n^{1-\delta})$ approximate degree lower bound
for small $\delta>0$ in the framework of~\cite{bun-thaler17adeg-ac0},
one needs a number of iterations that grows with $1/\delta$. Thus,
the overall circuit produced in~\cite{bun-thaler17adeg-ac0} has
a large constant number of alternating layers of $\AND$ and $\OR$
gates of logarithmic and polynomial fan-in, respectively, and in particular
cannot be flattened into a polynomial-size DNF or CNF formula. Proving
Theorem~\ref{thm:MAIN-dnf} within this framework would require reducing
the fan-in of the AND gates from $\Theta(\log n+\log m)$ to $O(1),$
which would completely destroy the componentwise composition and input
compression pillars of~\cite{bun-thaler17adeg-ac0}. These pillars
are present in all follow-up papers~\cite{bun-thaler17adeg-ac0,BKT17poly-strikes-back,BT18ac0-large-error,sherstov-wu18sign-ac0}
and seem impossible to get around, prompting the authors of~\cite[p.~14]{BT18ac0-large-error}
to entertain the possibility that the approximate degree of $\classAC^{0}$
at any given depth is much smaller than once conjectured. We show
that this is not the case.

\subsection{Our proof}

In this paper, we design hardness amplification from first principles,
without using componentwise composition or input compression. Our
approach efficiently amplifies the approximate degree even for functions
with sparse input, while ensuring that each hardness amplification
stage is implementable by a monotone circuit of constant depth with
AND gates of constant fan-in and OR gates of polynomial fan-in. As
a result, repeating our process any constant number of times produces
a polynomial-size DNF formula of constant width.

\subsubsection*{Our approach at a high level}

Let $f\colon\zoo^{N}\to\zoo$ be a given function. Let $f|_{\leq\theta}$
denote the restriction of $f$ to inputs of Hamming weight at most
$\theta,$ and let $d=\deg_{1/3}(f|_{\leq\theta})$ be the approximate
degree of this restriction. The total number of variables $N$ can
be vastly larger than $\theta$; in the actual proof, we will set
$N=\theta^{C}$ for a constant $C\geq1.$ Since an input $y\in\zoo^{N}$
to $f|_{\leq\theta}$ is guaranteed to have Hamming weight at most
$\theta,$ we can think of $y$ as the disjunction of $\theta$ vectors
of Hamming weight at most $1$ each:
\[
y=y_{1}\vee y_{2}\vee\cdots\vee y_{\theta},
\]
where each $y_{i}$ is either the zero vector $0^{N}$ or a basis
vector $e_{1},e_{2},\ldots,e_{N}$, and the disjunction on the right-hand
side is applied coordinate-wise. Our approach centers around encoding
each $y_{i}$ as a string of $n\ll N$ bits so as to make the decoding
difficult for polynomials but easy for circuits. Ideally, we would
like a decoding function $h\colon\zoon\to\zoo^{N}$ with the following
properties:
\begin{enumerate}
\item the sets $h^{-1}(v)$ for $v\in\{e_{1},e_{2},\ldots,e_{N},0^{N}\}$
are indistinguishable by polynomials of degree up to $D$, for some
parameter $D$;
\item the sets $h^{-1}(v)$ for $v\in\{e_{1},e_{2},\ldots,e_{N},0^{N}\}$
contain only strings of Hamming weight $O(1);$
\item $h$ is computable by a constant-depth monotone circuit with AND gates
of constant fan-in and OR gates of polynomial fan-in.
\end{enumerate}
With such $h$ in hand, define $F\colon(\zoon)^{\theta}\to\zoo$ by
\[
F(x_{1},x_{2},\ldots,x_{\theta})=f\left(\bigvee_{i=1}^{\theta}h(x_{i})\right).
\]
Then, one can reasonably expect that approximating $F$ is harder
than approximating $f|_{\leq\theta}.$ Indeed, an approximating polynomial
has access only to the encoded input $(x_{1},x_{2},\ldots,x_{\theta})$.
Decoding this input presumably involves computing $(x_{1},x_{2},\ldots,x_{\theta})\mapsto(h(x_{1}),h(x_{2}),\ldots,h(x_{\theta}))$
one way or another, which by property~(i) requires a polynomial of
degree greater than $D$. Once the decoded string $h(x_{1})\vee h(x_{2})\vee\cdots\vee h(x_{\theta})$
is available, the polynomial supposedly needs to compute $f$ on that
input, which in and of itself requires degree $d.$ Altogether, we
expect $F$ to have approximate degree on the order of $Dd.$ Moreover,
property~(ii) ensures that $F$ is hard to approximate even on inputs
of Hamming weight $O(\theta),$ putting us in a strong position for
another round of hardness amplification. Finally, property~(iii)
guarantees that the result of constantly many rounds of hardness amplification
is computable by a DNF formula of polynomial size and constant width.

\subsubsection*{Actual implementation}

As one might suspect, the above program is too bold and cannot be
implemented literally. Our actual construction of $h$ achieves~(i)\textendash (iii)
only approximately. In more detail, let $k$ be a sufficiently large
constant. For each $v\in\{e_{1},e_{2},\ldots,e_{N},0^{N}\},$ we construct
a probability distribution $\lambda_{v}$ on $\zoon$ that has all
but a vanishing fraction of its mass on inputs of Hamming weight exactly
$k,$ and moreover any two such distributions $\lambda_{v}$ and $\lambda_{v'}$
are indistinguishable by polynomials of low degree. We are further
able to ensure that an input of Hamming weight $k$ belongs to the
support of at most one of the distributions $\lambda_{v}$. Thus,
the $\lambda_{v}$ are in essence supported on pairwise disjoint sets
of strings of Hamming weight $k,$ and are pairwise indistinguishable
by polynomials of low degree. The decoding function $h$ works by
taking an input $x\in\zoon$ of Hamming weight $k$ and determining
which of the distributions has $x$ in its support\textemdash a highly
efficient computation realizable as a monotone $k$-DNF formula. With
small probability, $h$ will receive as input a string of Hamming
weight larger than $k,$ in which case the decoding may fail.

\subsubsection*{Construction of the $\lambda_{v}$}

Central to our work is the number-theoretic notion of \emph{$m$-discrepancy},
which is a measure of pseudorandomness or aperiodicity of a given
set of integers modulo $m.$ Formally, the $m$-discrepancy of a nonempty
finite set $S\subseteq\ZZ$ is defined as 
\[
\disc_{m}(S)=\max_{k=1,2,\ldots,m-1}\left|\frac{1}{|S|}\sum_{s\in S}\xi^{ks}\right|,
\]
where $\xi$ is a primitive $m$-th root of unity. The construction
of sparse sets with low discrepancy is a well-studied problem in combinatorics
and theoretical computer science. By building on previous work~\cite{AIKPS90aperiodic-set,sherstov18hardest-hs},
we construct a sparse set of integers with small discrepancy in our
regime of interest. For our application, we set the modulus $m=N+1.$

Continuing, let $\binom{[n]}{k}$ denote the family of cardinality-$k$
subsets of $[n]=\{1,2,\ldots,n\}.$ To design the distributions $\lambda_{v},$
we need an explicit coloring $\gamma\colon\binom{[n]}{k}\to[N+1]$
that is \emph{balanced}, in the sense that for nearly all large enough
subsets $A\subseteq\{1,2,\ldots,n\}$ and all $i\in[N+1],$ the family
$\gamma^{-1}(i)$ accounts for almost exactly a $1/(N+1)$ fraction
of all cardinality-$k$ subsets of $A.$ The existence of a highly
balanced coloring follows by the probabilistic method, and we construct
one explicitly using the sparse set of integers with small $(N+1)$-discrepancy
constructed earlier in the proof.

Our next ingredient is a dual polynomial $\omega$ for the OR function,
a staple in approximate degree lower bounds. An important property
of $\omega$ is that it places a constant fraction of its $\ell_{1}$
mass on the point $0^{n}.$ Translating $\omega$ from $0^{n}$ to
a point $z$ of slightly larger Hamming weight results in a new dual
polynomial, call it $\omega_{z}.$ Analogous to $\omega,$ the new
dual polynomial has a constant fraction of its $\ell_{1}$ mass on
$z$ and the rest on inputs that are greater than or equal to $z$
componentwise.

For notational convenience, let us now rename $\gamma$'s range elements
$1,2,\ldots,N+1$ to $e_{1},e_{2},\ldots,e_{N},0^{N},$ respectively.
For $v\in\{e_{1},e_{2},\ldots,e_{N},0^{N}\},$ define $\Phi_{v}$
to be the average of the dual polynomials $\omega_{z}$ where $z$
ranges over all characteristic vectors of the sets in $\gamma^{-1}(v).$
Being a convex combination of dual polynomials, each $\Phi_{v}$ is
a dual object orthogonal to polynomials of low degree. Observe further
that each $\Phi_{v}$ is supported on inputs of Hamming weight at
least $k,$ and any input of Hamming weight exactly $k$ belongs to
the support of exactly one $\Phi_{v}.$ For inputs $x$ of Hamming
weight greater than $k$, a remarkable thing happens: $\Phi_{v}(x)$
is almost the same for all $v.$ We prove this by exploiting the fact
that $\gamma$ is highly balanced. As a result, the ``common part''
of the $\Phi_{v}$ for inputs of Hamming weight greater than $k$
can be subtracted out to obtain a function $\widetilde{\Phi_{v}}$
for each $v\in\{e_{1},e_{2},\ldots,e_{N},0^{N}\}$. While these new
functions are not dual polynomials, the \emph{difference} of any two
of them is since $\widetilde{\Phi_{v}}-\widetilde{\Phi_{v'}}=\Phi_{v}-\Phi_{v'}$.
Put another way, the $\widetilde{\Phi_{v}}$ are pairwise indistinguishable
by low-degree polynomials. By defining the $\widetilde{\Phi_{v}}$
in a somewhat more subtle way, we further ensure that each $\widetilde{\Phi_{v}}$
is nonnegative. The distribution $\lambda_{v}$ can then be taken
to be the normalized function $\widetilde{\Phi_{v}}/\|\widetilde{\Phi_{v}}\|_{1}.$
This construction ensures all the properties that we need: $\lambda_{v}$
has nearly all of its mass on inputs of Hamming weight $k$; an input
of Hamming weight $k$ belongs to the support of at most one distribution
$\lambda_{v}$; and any pair of distributions $\lambda_{v},\lambda_{v'}$
are indistinguishable by a low-degree polynomial. Observe that in
our construction, $\lambda_{v}$ is close to the uniform probability
distribution on the characteristic vectors of the sets in $\gamma^{-1}(v).$

\section{\label{sec:Preliminaries}Preliminaries}

\subsection{General notation}

For a string $x\in\zoon$ and a set $S\subseteq\{1,2,\ldots,n\},$
we let $x|_{S}$ denote the restriction of $x$ to the indices in
$S.$ In other words, $x|_{S}=x_{i_{1}}x_{i_{2}}\ldots x_{i_{|S|}},$
where $i_{1}<i_{2}<\cdots<i_{|S|}$ are the elements of $S.$ The
\emph{characteristic vector} $\1_{S}$ of a set $S\subseteq\{1,2,\ldots,n\}$
is given by
\[
(\1_{S})_{i}=\begin{cases}
1 & \text{if }i\in S,\\
0 & \text{otherwise.}
\end{cases}
\]
Given an arbitrary set $X$ and elements $x,y\in X,$ the Kronecker
delta $\delta_{x,y}$ is defined by
\[
\delta_{x,y}=\begin{cases}
1 & \text{if }x=y,\\
0 & \text{otherwise.}
\end{cases}
\]
For a logical condition $C,$ we use the Iverson bracket
\[
\I[C]=\begin{cases}
1 & \text{if \ensuremath{C} holds,}\\
0 & \text{otherwise.}
\end{cases}
\]
We let $\NN=\{0,1,2,3,\ldots\}$ denote the set of natural numbers.
We use the comparison operators in a unary capacity to denote one-sided
intervals of the real line. Thus, ${<}a,$ ${\leq}a,$ ${>}a,$ ${\geq}a$
stand for $(-\infty,a),$ $(-\infty,a],$ $(a,\infty),$ $[a,\infty),$
respectively. We let $\ln x$ and $\log x$ stand for the natural
logarithm of $x$ and the logarithm of $x$ to base $2,$ respectively.
The term \emph{Euclidean space} refers to $\Re^{n}$ for some positive
integer $n.$ We let $e_{i}$ denote the vector whose $i$-th component
is $1$ and the others are $0.$ Thus, the vectors $e_{1},e_{2},\dots,e_{n}$
form the standard basis for $\Re^{n}.$ For a complex number $x,$
we denote the real part, imaginary part, and complex conjugate of
$x$ as usual by $\realpart(x),$ $\imagpart(x),$ and $\overline{x},$
respectively. We typeset the imaginary unit $\iu$ in boldface to
distinguish it from the index variable $i$. For an arbitrary integer
$a$ and a positive integer $m$, recall that $a\bmod m$ denotes
the unique element of $\{0,1,2,\ldots,m-1\}$ that is congruent to
$a$ modulo $m.$

For a set $X,$ we let $\Re^{X}$ denote the linear space of real-valued
functions on $X.$ The \emph{support} of a function $f\in\Re^{X}$
is denoted $\supp f=\{x\in X:f(x)\ne0\}.$ For real-valued functions
with finite support, we adopt the usual norms and inner product:
\begin{align*}
 & \|f\|_{\infty}=\max_{x\in\supp f}\,|f(x)|,\\
 & \|f\|_{1}=\sum_{x\in\supp f}\,|f(x)|,\\
 & \langle f,g\rangle=\sum_{x\in\supp f\,\cap\,\supp g}f(x)g(x).
\end{align*}
This covers as a special case functions on finite sets. Analogous
to functions, we adopt the familiar norms for vectors $x\in\Re^{n}$
in Euclidean space: $\|x\|_{\infty}=\max_{i=1,\ldots,n}|x_{i}|$ and
$\|x\|_{1}=\sum_{i=1}^{n}|x_{i}|.$ The \emph{tensor product} of $f\in\Re^{X}$
and $g\in\Re^{Y}$ is denoted $f\otimes g\in\Re^{X\times Y}$ and
given by $(f\otimes g)(x,y)=f(x)g(y).$ The tensor product $f\otimes f\otimes\cdots\otimes f$
($n$ times) is abbreviated $f^{\otimes n}.$ We frequently omit the
argument in equations and inequalities involving functions, as in
$\sign p=(-1)^{f}$. Such statements are to be interpreted pointwise.
For example, the statement ``$f\geq2|g|$ on $X$'' means that $f(x)\geq2|g(x)|$
for every $x\in X.$ For vectors $x$ and $y,$ the notation $x\leq y$
means that $x_{i}\leq y_{i}$ for each $i$.

We adopt the standard notation for function composition, with $f\circ g$
defined by $(f\circ g)(x)=f(g(x)).$ In addition, we use the $\circ$
operator to denote the \emph{componentwise composition} of Boolean
functions. Formally, the componentwise composition of $f\colon\zoon\to\{0,1\}$
and $g\colon X\to\{0,1\}$ is the function $f\circ g\colon X^{n}\to\{0,1\}$
given by $(f\circ g)(x_{1},x_{2},\ldots,x_{n})=f(g(x_{1}),g(x_{2}),\ldots,g(x_{n})).$
Componentwise composition is consistent with standard composition,
which in the context of Boolean functions is only defined for $n=1.$
Thus, the meaning of $f\circ g$ is determined by the range of $g$
and is never in doubt.

For a natural number $n,$ we abbreviate $[n]=\{1,2,\ldots,n\}$.
For a set $S$ and an integer $k,$ we let $\binom{S}{k}$ stand for
the family of cardinality-$k$ subsets of $S$:
\[
\binom{S}{k}=\{A\subseteq S:|A|=k\}.
\]
Analogously, for any set $I$, we define 
\[
\binom{S}{I}=\{A\subseteq S:|A|\in I\}.
\]
To illustrate, $\binom{S}{{\leq}k}$ denotes the family of subsets
of $S$ that have cardinality at most $k.$ Analogously, we have the
symbols $\binom{S}{{<}k},\binom{S}{{\geq}k},\binom{S}{{>}k}.$ Throughout
this manuscript, we use brace notation as in $\{z_{1},z_{2},\ldots,z_{n}\}$
to specify \emph{multisets} rather than \emph{sets}, the distinction
being that the number of times an element occurs is taken into account.
The \emph{cardinality} $|Z|$ of a finite multiset $Z$ is defined
to be the total number of element occurrences in $Z$, with each element
counted as many times as it occurs. The equality and subset relations
on multisets are defined analogously, with the number of element occurrences
taken into account. For example, $\{1,1,2\}=\{1,2,1\}$ but $\{1,1,2\}\ne\{1,2\}$.
Similarly, $\{1,2\}\subseteq\{1,1,2\}$ but $\{1,1,2\}\nsubseteq\{1,2\}.$

\subsection{Boolean strings and functions}

We identify the Boolean values ``true'' and ``false'' with $1$
and $0,$ respectively, and view Boolean functions as mappings $X\to\zoo$
for a finite set $X.$ The familiar functions $\OR_{n}\colon\zoon\to\zoo$
and $\AND_{n}\colon\zoon\to\zoo$ are given by $\OR_{n}(x)=\bigvee_{i=1}^{n}x_{i}$
and $\AND_{n}(x)=\bigwedge_{i=1}^{n}x_{i}.$ We abbreviate $\NOR_{n}=\neg\OR_{n}.$
For Boolean strings $x,y\in\zoon,$ we let $x\oplus y$ denote their
bitwise XOR. The strings $x\wedge y$ and $x\vee y$ are defined analogously,
with the binary operator applied bitwise.

For a vector $v\in\NN^{n},$ we define its \emph{weight }$|v|$ to
be $|v|=v_{1}+v_{2}+\cdots+v_{n}.$ If $x\in\zoon$ is a Boolean string,
then $|x|$ is precisely the Hamming weight of $x$. For any sets
$X\subseteq\NN^{n}$ and $W\subseteq\Re,$ we define $X|_{W}$ to
be the subset of vectors in $X$ whose weight belongs to $W$:
\[
X|_{W}=\{x\in X:|x|\in W\}.
\]
In the case of a one-element set $W=\{w\}$, we further shorten $X|_{\{w\}}$
to $X|_{w}.$ For example, $\NN^{n}|_{\leq w}$ denotes the set of
vectors whose $n$ components are natural numbers and sum to at most
$w$, whereas $\zoon|_{w}$ denotes the set of Boolean strings of
length $n$ and Hamming weight exactly $w.$ For a function $f\colon X\to\Re$
on a subset $X\subseteq\zoon,$ we let $f|_{W}$ denote the restriction
of $f$ to $X|_{W}.$ Thus, $f|_{W}$ is a function with domain $X|_{W}$
given by $f|_{W}(x)=f(x).$ A typical instance of this notation would
be $f|_{\leq w}$ for some real number~$w,$ corresponding to the
restriction of $f$ to Boolean strings of Hamming weight at most $w.$

\subsection{Concentration of measure}

Throughout this manuscript, we view probability distributions as real
functions. This convention makes available the shorthand notation
introduced above. In particular, for probability distributions $\mu$
and $\lambda,$ the symbol $\supp\mu$ denotes the support of $\mu$,
and $\mu\otimes\lambda$ denotes the probability distribution given
by $(\mu\otimes\lambda)(x,y)=\mu(x)\lambda(y).$ We use the notation
$\mu\times\lambda$ interchangeably with $\mu\otimes\lambda,$ the
former being more standard for probability distributions. If $\mu$
is a probability distribution on $X,$ we consider $\mu$ to be defined
also on any superset of $X$ with the understanding that $\mu=0$
outside $X.$

We recall the following multiplicative form of the Chernoff bound~\cite{chernoff52bounds}.
\begin{thm}[Chernoff bound]
\label{thm:chernoff}Let $X_{1},X_{2},\ldots,X_{n}\in\{0,1\}$ be
i.i.d.~random variables with $\Exp X_{i}=p.$ Then for all $0\leq\delta\leq1,$
\[
\Prob\left[\left|\sum_{i=1}^{n}X_{i}-pn\right|\geq\delta pn\right]\leq2\exp\left(-\frac{\delta^{2}pn}{3}\right).
\]
\end{thm}

\noindent Theorem~\ref{thm:chernoff} assumes i.i.d.~Bernoulli random
variables. Hoeffding's inequality~\cite{hoeffding-bound}, stated
next, is a more general concentration-of-measure result that applies
to any independent bounded random variables.
\begin{thm}[Hoeffding's inequality]
\label{thm:hoeffding}Let $X_{1},X_{2},\ldots,X_{n}$ be independent
random variables with $X_{i}\in[a_{i},b_{i}].$ Define $p=\sum_{i=1}^{n}\Exp X_{i}.$
Then for all $\delta\geq0,$
\[
\Prob\left[\left|\sum_{i=1}^{n}X_{i}-p\right|\geq\delta\right]\leq2\exp\left(-\frac{2\delta^{2}}{\sum_{i=1}^{n}(b_{i}-a_{i})^{2}}\right).
\]
\end{thm}

\noindent The standard version of Hoeffding's inequality, stated above,
requires $X_{1},X_{2},\ldots,X_{n}$ to be independent. Less known
are Hoeffding's results for dependent random variables, which he obtained
along with Theorem~\ref{thm:hoeffding} in his original paper~\cite{hoeffding-bound}.
We will specifically need the following concentration inequality for
sampling without replacement~\cite[Section~6]{hoeffding-bound}.
\begin{thm}[Hoeffding's sampling without replacement]
\label{thm:hoeffding-without-replacement} Let $\omega_{1},\omega_{2},\ldots,\omega_{N}$
be given reals, with $\omega_{i}\in[a,b]$ for all $i.$ Let $J_{1},J_{2},\ldots,J_{n}\in[N]$
be uniformly random integers that are pairwise distinct. Let $X_{i}=\omega_{J_{i}}$
for $i=1,2,\ldots,n,$ and define $p=\sum_{i=1}^{n}\Exp X_{i}.$ Then
for all $\delta\geq0,$
\[
\Prob\left[\left|\sum_{i=1}^{n}X_{i}-p\right|\geq\delta\right]\leq2\exp\left(-\frac{2\delta^{2}}{n(b-a)^{2}}\right).
\]
\end{thm}

\noindent Hoeffding's two theorems are clearly incomparable. On the
one hand, Theorem~\ref{thm:hoeffding} requires independence and
therefore does not apply to sampling without replacement. On the other
hand, each random variable $X_{i}$ in Theorem~\ref{thm:hoeffding-without-replacement}
must be uniformly distributed on a finite multiset of values, which
must further be the same multiset for all $X_{i}$; none of this is
assumed in Theorem~\ref{thm:hoeffding}.

\noindent %

Finally, we will need a concentration-of-measure result due to Bun
and Thaler~\cite[Lemma~4.7]{bun-thaler17adeg-ac0} for product distributions
on $\NN^{n}$\emph{.}
\begin{lem}[cf.~Bun and Thaler]
 \label{lem:concentration-of-measure} Let $\lambda_{1},\lambda_{2},\ldots,\lambda_{n}$
be distributions on $\NN$ with finite support such that
\begin{align*}
\lambda_{i}(t) & \leq\frac{C\alpha^{t}}{(t+1)^{2}}, &  & t\in\NN,
\end{align*}
where $C\geq0$ and $0\leq\alpha\leq1$. Then for all $T\geq8C\e n(1+\ln n),$
\[
\Prob_{v\sim\lambda_{1}\times\lambda_{2}\times\cdots\times\lambda_{n}}[\|v\|_{1}\geq T]\leq\alpha^{T/2}.
\]
\end{lem}

\noindent Bun and Thaler's result in~\cite[Lemma~4.7]{bun-thaler17adeg-ac0}
differs slightly from the statement above. The proof of Lemma~\ref{lem:concentration-of-measure}
as stated can be found in~\cite[Lemma~3.6]{sherstov-wu18sign-ac0}.
By leveraging Lemma~\ref{lem:concentration-of-measure}, we obtain
the following concentration result for probability distributions that
are supported on the Boolean hypercube, rather than $\NN$, and are
shifted from the origin.
\begin{lem}
\label{cor:concentration-of-measure-Boolean-shifted}Fix integers
$B\geq k\geq0.$ Let $\lambda_{1},\lambda_{2},\ldots,\lambda_{\ell}$
be probability distributions on $\{0,1\}^{B}$ with support contained
in $\zoo^{B}|_{\geq k}.$ Suppose further that
\begin{align*}
\lambda_{i}(\{0,1\}^{B}|_{t}) & \leq\frac{C\alpha^{t-k}}{(t-k+1)^{2}}, &  & i\in[\ell],\;\;t\in\{k,k+1,\ldots,B\},
\end{align*}
where $C\geq0$ and $0\leq\alpha\leq1.$ Then for all $T\geq8C\e\ell(1+\ln\ell)+\ell k,$
\[
\Prob_{(x_{1},\ldots,x_{\ell})\sim\lambda_{1}\times\cdots\times\lambda_{\ell}}\left[\sum_{i=1}^{\ell}|x_{i}|\geq T\right]\leq\alpha^{(T-\ell k)/2}.
\]
\end{lem}

\begin{proof}
For $i=1,2,\ldots,\ell,$ consider the distribution $\mu_{i}$ on
$\{0,1,\ldots,B-k\}$ given by $\mu_{i}(t)=\lambda_{i}(\zoo^{B}|_{t+k}).$
Then 
\begin{align}
\mu_{i}(t) & \leq\frac{C\alpha^{t}}{(t+1)^{2}}, &  & i\in[\ell],\;\;t\geq0.\label{eq:mu-i-concentration}
\end{align}
Moreover, the random variable $|x_{i}|$ with $x_{i}\sim\lambda_{i}$
has the same distribution as the random variable $u_{i}+k$ for $u_{i}\sim\mu_{i}.$
As a result,
\begin{align*}
\Prob_{(x_{1},\ldots,x_{\ell})\sim\lambda_{1}\times\cdots\times\lambda_{\ell}}\left[\sum_{i=1}^{\ell}|x_{i}|\geq T\right] & =\Prob_{u\sim\mu_{1}\times\mu_{2}\times\cdots\times\mu_{\ell}}\left[\sum_{i=1}^{\ell}(u_{i}+k)\geq T\right]\\
 & =\Prob_{u\sim\mu_{1}\times\mu_{2}\times\cdots\times\mu_{\ell}}[\|u\|_{1}\geq T-k\ell]\\
 & \leq\alpha^{(T-\ell k)/2},
\end{align*}
where the last step uses Lemma~\ref{lem:concentration-of-measure}
along with~(\ref{eq:mu-i-concentration}) and the hypothesis that~$T\geq8C\e\ell(1+\ln\ell)+\ell k.$
\end{proof}

\subsection{Orthogonal content}

For a multivariate polynomial $p\colon\Re^{n}\to\Re$, we let $\deg p$
denote the total degree of $p$, i.e., the largest degree of any monomial
of $p.$ We use the terms \emph{degree }and \emph{total degree }interchangeably
in this paper. It will be convenient to define the degree of the zero
polynomial by $\deg0=-\infty.$ For a real-valued function $\phi$
supported on a finite subset of $\Re^{n}$, the \emph{orthogonal content
of $\phi,$} denoted $\orth\phi$, is the minimum degree of a real
polynomial $p$ for which $\langle\phi,p\rangle\ne0.$ We adopt the
convention that $\orth\phi=\infty$ if no such polynomial exists.
It is clear that $\orth\phi\in\NN\cup\{\infty\},$ with the extremal
cases $\orth\phi=0\;\Leftrightarrow\;\langle\phi,1\rangle\ne0$ and
$\orth\phi=\infty\;\Leftrightarrow\;\phi=0.$ Additional facts about
orthogonal content are given by the following two propositions.
\begin{prop}
\label{prop:orth}Let $X$ and $Y$ be nonempty finite subsets of
Euclidean space. Then:
\begin{enumerate}
\item \label{item:orth-sum}$\orth(\phi+\psi)\geq\min\{\orth\phi,\orth\psi\}$
for all $\phi,\psi\colon X\to\Re;$
\item \label{item:orth-tensor}$\orth(\phi\otimes\psi)=\orth(\phi)+\orth(\psi)$
for all $\phi\colon X\to\Re$ and $\psi\colon Y\to\Re.$
\end{enumerate}
\end{prop}

\noindent A proof of Proposition~\ref{prop:orth} can be found in~\cite[Proposition~2.1]{sherstov-wu18sign-ac0}.
\begin{prop}
\label{prop:expect-out}Define $V=\{0^{N},e_{1},e_{2},\ldots,e_{N}\}\subseteq\Re^{N}.$
Fix functions $\phi_{v}\colon X\to\Re$ $(v\in V),$ where $X$ is
a finite subset of Euclidean space. Suppose that 
\begin{align}
\orth(\phi_{u}-\phi_{v}) & \geq D, &  & u,v\in V,\label{eq:orth-phi-u-phi-v}
\end{align}
where $D$ is a positive integer. Then for every polynomial $p\colon X^{\ell}\to\Re,$
the mapping $z\mapsto\langle\bigotimes_{i=1}^{\ell}\phi_{z_{i}},p\rangle$
is a polynomial on $V^{\ell}$ of degree at most $(\deg p)/D.$
\end{prop}

\begin{proof}
By linearity, it suffices to consider factored polynomials $p(x_{1},\ldots,x_{\ell})=\prod_{i=1}^{\ell}p_{i}(x_{i}),$
where each $p_{i}$ is a nonzero polynomial on $X.$ In this setting,
\begin{align}
\left\langle \bigotimes_{i=1}^{\ell}\phi_{z_{i}},p\right\rangle  & =\prod_{i=1}^{\ell}\left\langle \phi_{z_{i}},p_{i}\right\rangle .\label{eq:factored-inner-product-e1e2}
\end{align}
By (\ref{eq:orth-phi-u-phi-v}), we have $\langle\phi_{0^{N}},p_{i}\rangle=\langle\phi_{e_{1}},p_{i}\rangle=\langle\phi_{e_{2}},p_{i}\rangle=\cdots=\langle\phi_{e_{N}},p_{i}\rangle$
for any index $i$ with $\deg p_{i}<D.$ As a result, polynomials
$p_{i}$ with $\deg p_{i}<D$ do not contribute to the degree of the
right-hand side of~(\ref{eq:factored-inner-product-e1e2}) as a function
of $z.$ For the other polynomials $p_{i}$, the inner product $\langle\phi_{z_{i}},p_{i}\rangle$
is a linear polynomial in $z_{i},$ namely,
\begin{multline*}
\langle\phi_{z_{i}},p_{i}\rangle=z_{i,1}\langle\phi_{e_{1}},p_{i}\rangle+z_{i,2}\langle\phi_{e_{2}},p_{i}\rangle+\cdots+z_{i,N}\langle\phi_{e_{N}},p_{i}\rangle\\
+\left(1-\sum_{j=1}^{N}z_{i,j}\right)\langle\phi_{0^{N}},p_{i}\rangle.
\end{multline*}
Thus, polynomials $p_{i}$ with $\deg p_{i}\geq D$ contribute at
most $1$ each to the degree. Summarizing, the right-hand side of~(\ref{eq:factored-inner-product-e1e2})
is a real polynomial in $z_{1},z_{2},\ldots,z_{\ell}$ of degree at
most $|\{i:\deg p_{i}\geq D\}|\leq\frac{\deg p}{D}.$
\end{proof}
\noindent Proposition~\ref{prop:expect-out} generalizes an analogous
result in~\cite[Proposition~2.2]{sherstov-wu18sign-ac0}, where the
special case $N=1$ was treated.

\subsection{Polynomial approximation}

For a real number $\epsilon\geq0$ and a function $f\colon X\to\Re$
on a finite subset $X$ of Euclidean space, the\emph{ $\epsilon$-approximate
degree of $f$} is denoted $\deg_{\epsilon}(f)$ and is defined to
be the minimum degree of a polynomial $p$ such that $\|f-p\|_{\infty}\leq\epsilon.$
For $\epsilon<0$, it will be convenient to define $\deg_{\epsilon}(f)=+\infty$
since no polynomial satisfies $\|f-p\|_{\infty}\leq\epsilon$ in this
case. We focus on the approximate degree of Boolean functions $f\colon X\to\zoo.$
In this setting, the standard choice of the error parameter is $\epsilon=1/3.$
This choice is without loss of generality since $\deg_{\epsilon}(f)=\Theta(\deg_{1/3}(f))$
for every Boolean function $f$ and every constant $0<\epsilon<1/2.$
In what follows, we refer to $1/3$-approximate degree simply as ``approximate
degree.'' The notion of approximate degree has the following dual
characterization~\cite{sherstov07quantum,sherstov11quantum-sdpt}.
\begin{fact}
\label{fact:adeg-dual} Let $f\colon X\to\Re$ be given, for a finite
set $X\subset\Re^{n}$. Let $d\geq0$ be an integer and $\epsilon\geq0$
a real number. Then $\deg_{\epsilon}(f)\geq d$ if and only if there
exists a function $\psi\colon X\to\Re$ such that 
\begin{align*}
 & \langle f,\psi\rangle>\epsilon\|\psi\|_{1},\\
 & \orth\psi\geq d.
\end{align*}
\end{fact}

\noindent This characterization of approximate degree can be verified
using linear programming duality, cf.~\cite{sherstov07quantum,sherstov11quantum-sdpt}.
We now recall a variant of approximate degree for one-sided approximation.
For a Boolean function $f\colon X\to\zoo$ and $\epsilon\geq0,$ the
\emph{one-sided $\epsilon$-approximate degree} of $f$ is denoted
$\onedeg_{\epsilon}(f)$ and defined to be the minimum degree of a
real polynomial $p$ such that 
\begin{align*}
f(x)-\epsilon & \leq p(x)\leq f(x)+\epsilon, &  & x\in f^{-1}(0),\\
f(x)-\epsilon & \leq p(x), &  & x\in f^{-1}(1).
\end{align*}

\noindent We refer to any such polynomial as a \emph{one-sided approximant
}for $f$ with error $\epsilon.$ As usual, the canonical setting
of the error parameter is $\epsilon=1/3.$ In the pathological case
$\epsilon<0$, it will be convenient to define $\onedeg_{\epsilon}(f)=+\infty$.
Observe the asymmetric treatment of $f^{-1}(0)$ and $f^{-1}(1)$
in this formalism. In particular, the one-sided approximate degree
of Boolean functions is in general not invariant under negation. One-sided
approximate degree enjoys the following dual characterization~\cite{bun-thaler13amplification}. 
\begin{fact}
\noindent \label{fact:onedeg-dual}Let $f\colon X\to\zoo$ be given,
for a finite set $X\subset\Re^{n}$. Let $d\geq0$ be an integer and
$\epsilon\geq0$ a real number. Then $\onedeg_{\epsilon}(f)\geq d$
if and only if there exists a function $\psi\colon X\to\Re$ such
that 
\begin{align*}
 & \langle f,\psi\rangle>\epsilon\|\psi\|_{1},\\
 & \orth\psi\geq d,\\
 & \psi(x)\geq0\;\text{whenever }f(x)=1.
\end{align*}
\end{fact}

\subsection{Dual polynomials}

\noindent Facts~\ref{fact:adeg-dual} and~\ref{fact:onedeg-dual}
make it possible to prove lower bounds on approximate degree in a
constructive manner, by exhibiting a dual object $\psi$ that serves
as a witness. This object is referred to as a \emph{dual polynomial.}
Often, a dual polynomial for a composed function $f$ can be constructed
by combining dual objects for various components of $f.$ Of particular
importance in the study of $\classAC^{0}$ is the dual object for
the OR function. The first dual polynomial for OR was constructed
by Špalek~\cite{spalek08dual-or}, with many refinements and generalizations
obtained in follow-up work~\cite{bun-thaler13and-or-tree,sherstov14sign-deg-ac0,sherstov15asymmetry,bun-thaler17adeg-ac0,BKT17poly-strikes-back,sherstov-wu18sign-ac0}.
We will use the following construction from~\cite[Lemma~B.2]{sherstov-wu18sign-ac0}.
\begin{lem}
\label{lem:SW-OR-dual}Let $\epsilon$ be given, $0<\epsilon<1$.
Then for some constant $c=c(\epsilon)\in(0,1)$ and every integer
$n\geq1,$ there is an $($explicitly given$)$ function $\omega\colon\{0,1,2,\dots,n\}\to\Re$
such that
\begin{align*}
 & \omega(0)>\frac{1-\epsilon}{2}\cdot\|\omega\|_{1},\\
 & |\omega(t)|\leq\frac{1}{ct^{2}\,2^{ct/\sqrt{n}}}\cdot\|\omega\|_{1} &  & (t=1,2,\ldots,n),\\
 & (-1)^{t}\omega(t)\geq0 &  & (t=0,1,2,\dots,n),\\
 & \orth\omega\geq c\sqrt{n}.
\end{align*}
\end{lem}

A useful tool in the construction of dual polynomials is the following
lemma due to Razborov and Sherstov~\cite{RS07dc-dnf}.
\begin{lem}[Razborov and Sherstov]
\label{lem:razborov-sherstov}Fix integers $D$ and $n,$ where $0\leq D<n.$
Then there is an $($explicitly given$)$ function $\zeta\colon\zoon\to\Re$
such that
\begin{align}
 & \supp\zeta\subseteq\zoon|_{\leq D}\cup\{1^{n}\},\label{eq:RS-zeta-support}\\
 & \zeta(1^{n})=1,\label{eq:RS-zeta-unit-mass}\\
 & \|\zeta\|_{1}\leq1+2^{D}\binom{n}{D},\label{eq:RS-zeta-ell1}\\
 & \orth\zeta>D.\label{eq:RS-zeta-orth}
\end{align}
\end{lem}

\noindent In more detail, this result corresponds to taking $k=D$
and $\zeta=(-1)^{n}g$ in the proof of Lemma~3.2 of~\cite{RS07dc-dnf}.
We will need the following natural generalization of Lemma~\ref{lem:razborov-sherstov}. 
\begin{lem}
\label{lem:razborov-sherstov-generalized}Fix integers $D$ and $B,$
where $0\leq D<B.$ Let $y\in\zoo^{B}$ be a string with $|y|>D.$
Then there is an $($explicitly given$)$ function $\zeta_{y}\colon\zoo^{B}\to\Re$
such that
\begin{align}
 & \supp\zeta_{y}\subseteq\{x:x\leq y\text{ and }|x|\leq D\}\cup\{y\},\label{eq:zeta-y-supp}\\
 & \zeta_{y}(y)=1,\label{eq:zeta-y-point-mass}\\
 & \|\zeta_{y}\|_{1}\leq1+2^{D}\binom{B}{D},\label{eq:zeta-y-ell1}\\
 & \orth\zeta_{y}>D.\label{eq:zeta-y-orth}
\end{align}
\end{lem}

\begin{proof}
Set $n=|y|.$ Lemma~\ref{lem:razborov-sherstov} gives an explicit
function $\zeta\colon\zoon\to\Re$ that satisfies~(\ref{eq:RS-zeta-support})\textendash (\ref{eq:RS-zeta-orth}).
Define $\zeta_{y}\colon\zoo^{B}\to\Re$ by
\[
\zeta_{y}(x)=\zeta(x|_{S})\prod_{i\notin S}(1-x_{i}),
\]
where $S=\{i:y_{i}=1\}$. Then~(\ref{eq:zeta-y-point-mass}) and~(\ref{eq:zeta-y-ell1})
are immediate from~(\ref{eq:RS-zeta-unit-mass}) and~(\ref{eq:RS-zeta-ell1}),
respectively. Property~(\ref{eq:zeta-y-orth}) follows from~(\ref{eq:RS-zeta-orth})
in light of Proposition~\ref{prop:orth}~\ref{item:orth-tensor}.
To verify the remaining property~(\ref{eq:zeta-y-supp}), fix any
input $x$ with $\zeta_{y}(x)\ne0$. Then the definition of $\zeta_{y}$
implies that $x|_{\overline{S}}$ is the zero vector, whereas (\ref{eq:RS-zeta-support})
implies that $x|_{S}$ is either $1^{n}$ or a string of Hamming weight
at most $D.$ In the former case, we have $x=y$; in the latter case,
$x\leq y$ and $|x|\leq D.$
\end{proof}
\noindent Informally, Lemmas~\ref{lem:razborov-sherstov} and~\ref{lem:razborov-sherstov-generalized}
are useful when one needs to adjust a dual object's metric properties
while preserving its orthogonality to low-degree polynomials. These
lemmas play a basic role in several recent papers~\cite{RS07dc-dnf,bun-thaler17adeg-ac0,BKT17poly-strikes-back,BT18ac0-large-error,sherstov-wu18sign-ac0}
as well as our work. For the reader's benefit, we encapsulate this
procedure as Lemma~\ref{lem:zero-high-Hamming-weight} below and
provide a detailed proof.
\begin{lem}
\label{lem:zero-high-Hamming-weight}Let $\Phi\colon\zoo^{B}\to\Re$
be given. Fix integers $T\geq D\geq0$. Then there is an $($explicitly
given$)$ function $\tilde{\Phi}\colon\zoo^{B}\to\Re$ such that
\begin{align}
 & \supp\tilde{\Phi}\subseteq\zoo^{B}|_{\leq T},\label{eq:Phi-tilde-support}\\
 & \orth(\Phi-\tilde{\Phi})>D,\label{eq:Phi-tilde-Phi-orth}\\
 & \|\Phi-\tilde{\Phi}\|_{1}\leq\left(1+2^{D}\binom{B}{D}\right)\sum_{x:|x|>T}|\Phi(x)|.\label{eq:Phi-tilde-Phi-ell1}
\end{align}
\end{lem}

\begin{proof}[Proof \emph{(adapted from~\cite{RS07dc-dnf,bun-thaler17adeg-ac0,BKT17poly-strikes-back,BT18ac0-large-error,sherstov-wu18sign-ac0}).}]
 For $T\geq B,$ the lemma holds trivially with $\tilde{\Phi}=\Phi.$
In what follows, we treat the complementary case $T<B.$

For each $y\in\zoo^{B}|_{>T}$, Lemma~\ref{lem:razborov-sherstov-generalized}
constructs a function $\zeta_{y}\colon\zoo^{B}\to\Re$ that obeys~(\ref{eq:zeta-y-supp})\textendash (\ref{eq:zeta-y-orth}).
Define 
\[
\tilde{\Phi}=\Phi-\sum_{y\in\zoo^{B}|_{>T}}\Phi(y)\zeta_{y}.
\]
Then for $x\in\zoo^{B}|_{>T}$, properties~(\ref{eq:zeta-y-supp})
and~(\ref{eq:zeta-y-point-mass}) force $\zeta_{y}(x)=\delta_{x,y}$
and consequently $\tilde{\Phi}(x)=\Phi(x)-\Phi(x)=0.$ This settles~(\ref{eq:Phi-tilde-support}).
Property~(\ref{eq:Phi-tilde-Phi-orth}) is justified by
\[
\orth(\Phi-\tilde{\Phi})=\orth\left(\sum_{y\in\zoo^{B}|_{>T}}\Phi(y)\zeta_{y}\right)\geq\min_{y\in\zoo^{B}|_{>T}}\orth\zeta_{y}>D,
\]
where the last two steps use Proposition~\ref{prop:orth}\ref{item:orth-sum}
and~(\ref{eq:zeta-y-orth}), respectively. The final property~(\ref{eq:Phi-tilde-Phi-ell1})
can be derived as follows:
\begin{align*}
\|\Phi-\tilde{\Phi}\|_{1} & =\left\Vert \sum_{y\in\zoo^{B}|_{>T}}\Phi(y)\zeta_{y}\right\Vert _{1}\\
 & \leq\sum_{y\in\zoo^{B}|_{>T}}|\Phi(y)|\|\zeta_{y}\|_{1}\\
 & \leq\left(1+2^{D}\binom{B}{D}\right)\sum_{y\in\zoo^{B}|_{>T}}|\Phi(y)|,
\end{align*}
where the last two steps use the triangle inequality and~(\ref{eq:zeta-y-ell1}),
respectively.
\end{proof}

\subsection{Symmetrization}

Let $S_{n}$ denote the symmetric group on $n$ elements. For a permutation
$\sigma\in S_{n}$ and an arbitrary sequence $x=(x_{1},x_{2},\ldots,x_{n}),$
we adopt the shorthand $\sigma x=(x_{\sigma(1)},x_{\sigma(2)},\ldots,x_{\sigma(n)}).$
A function $f(x_{1},x_{2},\ldots,x_{n})$ is called \emph{symmetric}
if it is invariant under permutation of the input variables: $f(x_{1},x_{2},\ldots,x_{n})=f(x_{\sigma(1)},x_{\sigma(2)},\ldots,x_{\sigma(n)})$
for all $x$ and $\sigma.$ Symmetric functions on $\zoon$ are intimately
related to univariate polynomials, as was first observed by Minsky
and Papert in their \emph{symmetrization argument}~\cite{minsky88perceptrons}.
\begin{prop}[Minsky and Papert]
\label{prop:minsky-papert}Let $p\colon\Re^{n}\to\Re$ be a given
polynomial. Then the mapping
\[
t\mapsto\Exp_{\substack{x\in\zoon|_{t}}
}\;p(x)
\]
is a univariate polynomial on $\{0,1,2,\ldots,n\}$ of degree at most
$\deg p.$
\end{prop}

\noindent 

\noindent The next result, proved in~\cite[Corollary~2.13]{sherstov-wu18sign-ac0},
generalizes Minsky and Papert's symmetrization to the setting when
$x_{1},x_{2},\ldots,x_{n}$ are vectors rather than bits.
\begin{fact}[Sherstov and Wu]
\label{fact:ambainis-symmetrization}Let $p\colon(\Re^{N})^{\theta}\to\Re$
be a given polynomial. Then the mapping
\begin{equation}
v\mapsto\Exp_{\substack{x\in\{0^{N},e_{1},e_{2},\ldots,e_{N}\}^{\theta}:\\
x_{1}+x_{2}+\cdots+x_{\theta}=v
}
}\;\;p(x)\label{eq:ambainis-mapping}
\end{equation}
is a polynomial on $\NN^{N}|_{\leq\theta}$ of degree at most $\deg p.$
\end{fact}

\noindent Minsky and Papert's symmetrization corresponds to $N=1$
in Fact~\ref{fact:ambainis-symmetrization}.

\subsection{\label{subsec:Number-theoretic-preliminaries}Number theory}

For positive integers $a$ and $b$ that are relatively prime, we
let $(a^{-1})_{b}\in\{1,2,\ldots,b-1\}$ denote the multiplicative
inverse of $a$ modulo $b.$ The following fact is well-known and
straightforward to verify;~see, e.g.,~\cite[Fact~2.8]{sherstov18hardest-hs}.
\begin{fact}
\label{fact:rel-prime}For any positive integers $a$ and $b$ that
are relatively prime,
\[
\frac{(a^{-1})_{b}}{b}+\frac{(b^{-1})_{a}}{a}-\frac{1}{ab}\in\mathbb{Z}.
\]
\end{fact}

The \emph{prime counting function} $\pi(x)$ for a real argument $x\geq0$
evaluates to the number of prime numbers less than or equal to $x.$
In this manuscript, it will be clear from the context whether $\pi$
refers to $3.14159\ldots$ or the prime counting function. The asymptotic
growth of the latter is given by the \emph{prime number theorem},
which states that $\pi(n)\sim n/\ln n.$ The following explicit bound
on $\pi(n)$ is due to Rosser~\cite{rosser41primes}. 
\begin{fact}[Rosser]
\label{fact:PNT}For $n\geq55,$
\[
\frac{n}{\ln n+2}<\pi(n)<\frac{n}{\ln n-4}.
\]
\end{fact}

\noindent 

\noindent The number of distinct prime divisors of a natural number
$n$ is denoted $\nu(n)$. The following first-principles bound on
$\nu(n)$ is asymptotically tight for infinitely many $n;$ see~\cite[Fact~2.11]{sherstov18hardest-hs}
for details.
\begin{fact}
\label{fact:num-prime-factors}The number of distinct prime divisors
of $n$ obeys
\[
(\nu(n)+1)!\leq n.
\]
In particular,
\[
\nu(n)\leq(1+o(1))\frac{\ln n}{\ln\ln n}.
\]
\end{fact}

\noindent 

\section{\label{sec:Balanced-colorings}Balanced colorings}

For integers $n\geq k\geq1$ and $r\geq1,$ consider a mapping $\gamma\colon\binom{[n]}{k}\to[r].$
We refer to any such $\gamma$ as a \emph{coloring of $\binom{[n]}{k}$
with $r$ colors}. An important ingredient in our work is the construction
of a \emph{balanced} coloring, in the following technical sense.
\begin{defn}
Let $\gamma\colon\binom{[n]}{k}\to[r]$ be a given coloring. For a
subset $A\subseteq[n],$ we say that \emph{$\gamma$ is $\epsilon$-balanced
on $A$} iff for each $i\in[r],$
\[
\frac{1-\epsilon}{r}\binom{|A|}{k}\leq\left|\gamma^{-1}(i)\cap\binom{A}{k}\right|\leq\frac{1+\epsilon}{r}\binom{|A|}{k}.
\]
We define $\gamma$ to be \emph{$(\epsilon,\delta,m)$-balanced} iff
\[
\Prob_{A\in\binom{[n]}{\ell}}[\gamma\text{ is }\text{\ensuremath{\epsilon}-balanced on \ensuremath{A}]}\ensuremath{\geq1-\delta}
\]
for all $\ell\in\{m,m+1,\ldots,n\}.$
\end{defn}

As one might expect, a uniformly random coloring is balanced with
high probability; we establish this fact in Section~\ref{subsec:Existence-proof}.
In Sections~\ref{subsec:Discrepancy-defined}\textendash \ref{subsec:An-explicit-balanced}
that follow, we construct a highly balanced coloring based on an integer
set with low discrepancy. The reader who is interested only in the
quantitative aspect of our theorems and is not concerned about explicitness,
may read Section~\ref{subsec:Existence-proof} and skip without loss
of continuity to Section~\ref{sec:hardness-amplification}.

\subsection{\label{subsec:Existence-proof}Existence of balanced colorings}

The next lemma uses the probabilistic method to establish the existence
of balanced colorings with excellent parameters.
\begin{lem}
\label{lem:balanced-coloring-existential}Let $\epsilon,\delta\in(0,1]$
be given. Let $n,m,k,r$ be positive integers with $n\geq m\geq k$
and
\begin{equation}
\binom{m}{k}\geq\frac{3r}{\epsilon^{2}}\cdot\ln\frac{2rn}{\delta}.\label{eq:eps-m-k-n-r}
\end{equation}
Then there exists an $(\epsilon,\delta,m)$-balanced coloring $\gamma\colon\binom{[n]}{k}\to[r]$.
\end{lem}

\begin{proof}
Let $\gamma\colon\binom{[n]}{k}\to[r]$ be a uniformly random coloring.
For fixed $i$ and $A\in\binom{[n]}{{\geq}m},$ the cardinality $|\gamma^{-1}(i)\cap\binom{A}{k}|$
is the sum of $\binom{|A|}{k}$ independent Bernoulli random variables,
each with expected value $1/r.$ As a result,
\begin{align}
 & \Prob_{\gamma}[\text{\ensuremath{\gamma} is not \ensuremath{\epsilon}-balanced on \ensuremath{A}}]\nonumber \\
 & \qquad\qquad=\Prob_{\gamma}\left[\max_{i\in[r]}\left|\left|\gamma^{-1}(i)\cap\binom{A}{k}\right|-\frac{1}{r}\binom{|A|}{k}\right|>\frac{\epsilon}{r}\binom{|A|}{k}\right]\nonumber \\
 & \qquad\qquad\leq r\max_{i\in[r]}\;\;\Prob_{\gamma}\left[\left|\left|\gamma^{-1}(i)\cap\binom{A}{k}\right|-\frac{1}{r}\binom{|A|}{k}\right|>\frac{\epsilon}{r}\binom{|A|}{k}\right]\nonumber \\
 & \qquad\qquad\leq r\cdot2\exp\left(-\frac{\epsilon^{2}}{3r}\binom{|A|}{k}\right)\nonumber \\
 & \qquad\qquad\leq r\cdot2\exp\left(-\frac{\epsilon^{2}}{3r}\binom{m}{k}\right)\nonumber \\
 & \qquad\qquad\leq\frac{\delta}{n},\label{eq:gamma-not-balanced-on-A}
\end{align}
where the second step applies the union bound over $i\in[r],$ the
third step uses the Chernoff bound (Theorem~\ref{thm:chernoff}),
and the fifth step uses~(\ref{eq:eps-m-k-n-r}). Now
\begin{align*}
 & \Exp_{\gamma}\;\max_{\ell\in\{m,m+1,\ldots,n\}}\Prob_{A\in\binom{[n]}{\ell}}[\text{\ensuremath{\gamma} is not \ensuremath{\epsilon}-balanced on \ensuremath{A}]}\\
 & \qquad\qquad\leq\Exp_{\gamma}\;\;\sum_{\ell=m}^{n}\Prob_{A\in\binom{[n]}{\ell}}[\text{\ensuremath{\gamma} is not \ensuremath{\epsilon}-balanced on \ensuremath{A}]}\\
 & \qquad\qquad=\sum_{\ell=m}^{n}\Exp_{A\in\binom{[n]}{\ell}}\Prob_{\gamma}[\text{\ensuremath{\gamma} is not \ensuremath{\epsilon}-balanced on \ensuremath{A}}]\\
 & \qquad\qquad\leq\sum_{\ell=m}^{n}\frac{\delta}{n}\\
 & \qquad\qquad\leq\delta,
\end{align*}
where the next-to-last step uses~(\ref{eq:gamma-not-balanced-on-A}).
We conclude that there exists a coloring $\gamma$ with
\[
\max_{\ell\in\{m,m+1,\ldots,n\}}\Prob_{A\in\binom{[n]}{\ell}}[\text{\ensuremath{\gamma} is not \ensuremath{\epsilon}-balanced on \ensuremath{A}]}\leq\delta,
\]
which is the definition of an $(\epsilon,\delta,m)$-balanced coloring.
\end{proof}
For our purposes, the following consequence of Lemma~\ref{lem:balanced-coloring-existential}
will be sufficient.
\begin{cor}
\label{cor:balanced-coloring-existential}Let $n,m,k,r$ be positive
integers with $n\geq m\geq k^{2}.$ Let $\epsilon\in(0,1]$ be given
with 
\[
\epsilon\geq\frac{3r\sqrt{k\ln(n+1)}}{m^{k/4}}.
\]
Then there exists an $(\epsilon,\epsilon,m)$-balanced coloring $\gamma\colon\binom{[n]}{k}\to[r].$
\end{cor}

\begin{proof}
We have
\begin{align*}
\frac{3r}{\epsilon^{2}}\cdot\ln\frac{2rn}{\epsilon} & \leq\frac{3rm^{k/2}}{9r^{2}k\ln(n+1)}\cdot\ln\left(\frac{2rn}{3r\sqrt{k\ln(n+1)}}\cdot m^{k/4}\right)\\
 & \leq\frac{m^{k/2}}{3rk\ln(n+1)}\cdot\ln(n\cdot m^{k/4})\\
 & \leq\frac{m^{k/2}}{3rk\ln(n+1)}\cdot2k\ln n\\
 & \leq m^{k/2}\\
 & \leq\left(\frac{m}{k}\right)^{k}\\
 & \leq\binom{m}{k},
\end{align*}
where the next-to-last step uses the hypothesis $m\geq k^{2}.$ By
Lemma~\ref{lem:balanced-coloring-existential}, we conclude that
there is an $(\epsilon,\epsilon,m)$-balanced coloring $\gamma\colon\binom{[n]}{k}\to[r].$
\end{proof}
\noindent In Sections~\ref{subsec:Discrepancy-defined}\textendash \ref{subsec:An-explicit-balanced}
below, we will give an explicit coloring with parameters essentially
matching Corollary~\ref{cor:balanced-coloring-existential}.

\subsection{\label{subsec:Discrepancy-defined}Discrepancy defined}

Discrepancy is a measure of pseudorandomness or aperiodicity of a
multiset of integers with respect to a given modulus $M.$ Formally,
let $M\geq2$ be a given integer. The \emph{$M$-discrepancy} of a
nonempty multiset $Z=\{z_{1},z_{2},\ldots,z_{n}\}$ of arbitrary integers
is defined as
\[
\disc_{M}(Z)=\max_{k=1,2,\ldots,M-1}\left|\frac{1}{n}\sum_{j=1}^{n}\omega^{kz_{j}}\right|,
\]
where $\omega$ is a primitive $M$-th root of unity; the right-hand
side is obviously the same for any such $\omega$. Equivalently, we
may write
\[
\disc_{M}(Z)=\max_{\omega\ne1:\omega^{M}=1}\left|\frac{1}{n}\sum_{j=1}^{n}\omega^{z_{j}}\right|,
\]
where the maximum is over $M$-th roots of unity $\omega$ other than
$1$. Yet another way to think of $M$-discrepancy is in terms of
the discrete Fourier transform on $\ZZ_{M}.$ Specifically, consider
the \emph{frequency vector $(f_{0},f_{1},\ldots,f_{M-1})$} of $Z$,
where $f_{j}$ is the total number of element occurrences in $Z$
that are congruent to $j$ modulo $M.$ Applying the discrete Fourier
transform to $(f_{j})_{j=0}^{M-1}$ produces the sequence $(\sum_{j=0}^{M-1}f_{j}\exp(-2\pi\iu kj/M))_{k=0}^{M-1}=(\sum_{j=1}^{n}\exp(-2\pi\iu kz_{j}/M))_{k=0}^{M-1},$
which is a permutation of $(n,\sum_{j=1}^{n}\omega^{z_{j}},\ldots,\sum_{j=1}^{n}\omega^{(M-1)z_{j}})$
for a primitive $M$-th root of unity $\omega.$ Thus, the $M$-discrepancy
of $Z$ coincides up to a normalizing factor with the largest absolute
value of a nonconstant Fourier coefficient of the frequency vector
of $Z.$ The notion of $m$-discrepancy has a long history in combinatorics
and theoretical computer science; see~\cite{sherstov18hardest-hs}
for a bibliographic overview.
\begin{lem}[Discrepancy under sampling without replacement]
\label{lem:discrepancy-subset}Fix integers $n\geq t\geq1$ and $M\geq2.$
Let $Z=\{z_{1},z_{2},\ldots,z_{n}\}$ be a multiset of integers. Then
for all $\alpha\in[0,1],$
\[
\Prob_{S\in\binom{[n]}{t}}[\disc_{M}(\{z_{i}:i\in S\})-\disc_{M}(Z)\geq\alpha]\leq4M\exp\left(-\frac{t\alpha^{2}}{8}\right),
\]
where $\{z_{i}:i\in S\}$ is understood to be a multiset of cardinality
$t$.
\end{lem}

\begin{proof}
Fix an $M$-th root of unity $\omega.$ Then $\realpart(\omega^{z_{1}}),\realpart(\omega^{z_{2}}),\ldots,\realpart(\omega^{z_{n}})$
range in $[-1,1]$. Now, let $S\in\binom{[n]}{t}$ be a uniformly
random subset. Then the Hoeffding inequality for sampling without
replacement (Theorem~\ref{thm:hoeffding-without-replacement}) implies
that
\[
\Prob_{S\in\binom{[n]}{t}}\left[\left|\frac{1}{t}\sum_{j\in S}\realpart(\omega^{z_{j}})-\frac{1}{n}\sum_{j=1}^{n}\realpart(\omega^{z_{j}})\right|\geq\frac{\alpha}{2}\right]\leq2\exp\left(-\frac{t\alpha^{2}}{8}\right).
\]
Analogously,
\[
\Prob_{S\in\binom{[n]}{t}}\left[\left|\frac{1}{t}\sum_{j\in S}\imagpart(\omega^{z_{j}})-\frac{1}{n}\sum_{j=1}^{n}\imagpart(\omega^{z_{j}})\right|\geq\frac{\alpha}{2}\right]\leq2\exp\left(-\frac{t\alpha^{2}}{8}\right).
\]
Combining these two equations shows that for every $M$-th root of
unity $\omega,$
\begin{equation}
\Prob_{S\in\binom{[n]}{t}}\left[\left|\frac{1}{t}\sum_{j\in S}\omega^{z_{j}}-\frac{1}{n}\sum_{j=1}^{n}\omega^{z_{j}}\right|\geq\alpha\right]\leq4\exp\left(-\frac{t\alpha^{2}}{8}\right).\label{eq:w-close}
\end{equation}
Now
\begin{align}
 & \disc_{M}(\{z_{i}:i\in S\})-\disc_{M}(Z)\nonumber \\
 & \qquad\qquad=\max_{\omega}\left|\frac{1}{t}\sum_{j\in S}\omega^{z_{j}}\right|-\max_{\omega}\left|\frac{1}{n}\sum_{j=1}^{n}\omega^{z_{j}}\right|\nonumber \\
 & \qquad\qquad\leq\max_{\omega}\left\{ \left|\frac{1}{t}\sum_{j\in S}\omega^{z_{j}}\right|-\left|\frac{1}{n}\sum_{j=1}^{n}\omega^{z_{j}}\right|\right\} \nonumber \\
 & \qquad\qquad\leq\max_{\omega}\left|\frac{1}{t}\sum_{j\in S}\omega^{z_{j}}-\frac{1}{n}\sum_{j=1}^{n}\omega^{z_{j}}\right|,\label{eq:w-w-close}
\end{align}
where the maximum in all equations is taken over $M$-th roots of
unity $\omega\ne1.$ Using~(\ref{eq:w-close}) and the union bound
over $\omega,$ we see that the right-hand side of~(\ref{eq:w-w-close})
is bounded by $\alpha$ with probability at least $1-4M\exp(-t\alpha^{2}/8).$
\end{proof}

\subsection{\label{subsec:An-explicit-construction}A low-discrepancy set}

The construction of sparse integer sets with small discrepancy relative
to a given modulus $M$ is a well-studied problem. There is an inherent
trade-off between the size of the set and the discrepancy it achieves,
and different works have focused on different regimes depending on
the application at hand. We work in a regime not considered previously:
for any constant $\epsilon>0$, we construct a set of cardinality
at most $M^{\epsilon}$ that has $M$-discrepancy at most $M^{-\delta}$
for some constant $\delta=\delta(\epsilon)>0.$ We construct such
a set based on the following result.
\begin{thm}[cf.~\cite{AIKPS90aperiodic-set,sherstov18hardest-hs}]
\label{thm:ajtai-iteration}Fix an integer $R\geq1$ and reals $P\geq2$
and $\Delta\geq1$. Let $M$ be an integer with
\[
M\geq P^{2}(R+1).
\]
Fix a set $S_{p}\subseteq\{1,2,\ldots,p-1\}$ for each prime $p\in(P/2,P]$
with $p\nmid M$. Suppose further that the cardinalities of any two
sets from among the $S_{p}$ differ by a factor of at most $\Delta.$
Consider the multiset
\begin{multline}
S=\{(r+s\cdot(p^{-1})_{M})\bmod M:\\
\qquad r=1,\ldots,R;\quad p\in(P/2,P]\text{ prime with }p\nmid M;\quad s\in S_{p}\}.\label{eq:low-discrepancy-set}
\end{multline}
Then the elements of $S$ are pairwise distinct and nonzero. Moreover,
if $S\ne\varnothing$ then
\[
\disc_{M}(S)\leq\frac{c}{\sqrt{R}}+\frac{c\log M}{\log\log M}\cdot\frac{\log P}{P}\cdot\Delta+\max_{p}\{\disc_{p}(S_{p})\}
\]
for some $($explicitly given$)$ constant $c\geq1$ independent of
$P,R,M,\Delta.$
\end{thm}

\noindent Ajtai et al.~\cite{AIKPS90aperiodic-set} proved a special
case of Theorem~\ref{thm:ajtai-iteration} for $M$ prime and $\Delta=1.$
Their argument was generalized in~\cite[Theorem~3.6]{sherstov18hardest-hs}
to arbitrary moduli $M$, again in the setting of $\Delta=1$. The
treatment in~\cite{sherstov18hardest-hs} in turn readily generalizes
to any $\Delta\geq1,$ and for the reader's convenience we provide
a complete proof of Theorem~\ref{thm:ajtai-iteration} in Appendix~\ref{app:ajtai-iteration}.
With this result in hand, we obtain the low-discrepancy set with the
needed parameters:
\begin{thm}[Explicit low-discrepancy set]
\label{thm:explicit-low-discrepancy-set}For all integers $M\geq2$
and $t\geq2,$ there is an $($explicitly given$)$ nonempty set $S\subseteq\{1,2,\ldots,M\}$
with
\begin{align}
 & |S|\leq t,\\
 & \disc_{M}(S)\leq\frac{C^{*}\log t}{t^{1/4}}\cdot\frac{\log M}{1+\log\log M},\label{eq:disc-M-t}
\end{align}
where $C^{*}\geq1$ is an $($explicitly given$)$ absolute constant
independent of $M$ and $t.$
\end{thm}

\begin{proof}
Facts~\ref{fact:PNT} and~\ref{fact:num-prime-factors} imply that
\begin{align}
 & \pi(P)-\pi\left(\frac{P}{2}\right)\geq\frac{P}{C\log P} &  & \text{for all }P\geq C,\label{eq:many-primes-1}\\
 & \nu(M)\leq\frac{C\log M}{1+\log\log M} &  & \text{for all }M\geq2,\label{eq:few-prime-factors-1}
\end{align}
for some integer $C\geq1$ that is an absolute constant. Moreover,
$C$ can be easily calculated from the explicit bounds in Facts~\ref{fact:PNT}
and~\ref{fact:num-prime-factors}. We will show that the theorem
holds for some constant $C^{*}\geq4C^{2}.$

For $t\geq M,$ the theorem is trivial since the set $S=\{1,2,\ldots,M\}$
achieves $\disc_{M}(S)=0.$ Also, if the right-hand side of~(\ref{eq:disc-M-t})
exceeds $1$, then (\ref{eq:disc-M-t}) holds trivially for the set
$S=\{1\}.$ In what follows, we treat the remaining case when
\begin{align}
 & t<M,\label{eq:t-less-than-M}\\
 & \frac{4C^{2}\log t}{t^{1/4}}\cdot\frac{\log M}{1+\log\log M}\le1.\label{eq:rhs-less-than-1}
\end{align}
The latter condition forces
\begin{equation}
t\geq\max\{81,C^{8}\}.\label{eq:t-81}
\end{equation}
Set $P=\lfloor t^{1/4}\rfloor$ and $R=\lfloor\sqrt{t}-1\rfloor.$
Then~(\ref{eq:t-less-than-M}) and~(\ref{eq:t-81}) imply that $P\geq\max\{3,C\}$,
$R\geq1,$ and $M\geq P^{2}(R+1).$ As a result, Theorem~\ref{thm:ajtai-iteration}
is applicable with the sets $S_{p}=\{1,2,\ldots,p-1\}$ for prime
$p\in(P/2,P].$ The discrepancy of these sets is given by $\disc_{p}(S_{p})=1/(p-1)$.
Define $S$ by~(\ref{eq:low-discrepancy-set}). The interval $(P/2,P]$
contains $\pi(P)-\pi(P/2)$ prime numbers, of which at most $\nu(M)$
are divisors of $M.$ We have
\begin{align*}
\pi(P)-\pi\left(\frac{P}{2}\right)-\nu(M) & \geq\frac{P}{C\log P}-\frac{C\log M}{1+\log\log M}\\
 & \geq\frac{t^{1/4}}{C\log t}-\frac{C\log M}{1+\log\log M}\\
 & >0,
\end{align*}
where the first step uses~(\ref{eq:many-primes-1}), (\ref{eq:few-prime-factors-1}),
and $P\geq C$, and the last step uses~(\ref{eq:rhs-less-than-1}).
We conclude that $(P/2,P]$ contains a prime that does not divide
$M,$ which in turn implies that $S$ is nonempty. Continuing, $P\geq3$
forces $\Delta\leq(P-1)/(\lceil P/2\rceil-1)\leq3$ in the notation
of Theorem~\ref{thm:ajtai-iteration}. As a result, Theorem~\ref{thm:ajtai-iteration}
guarantees~(\ref{eq:disc-M-t}) for a large enough constant $C^{*}.$
We note that $C^{*}$ can be easily calculated from the constant $c$
in Theorem~\ref{thm:ajtai-iteration}. Since $|S|\leq RP^{2}\leq t$
by definition, the proof is complete.
\end{proof}

\subsection{\label{subsec:Discrepancy-and-balanced}Discrepancy and balanced
colorings}

We will leverage the low-discrepancy integer set in Theorem~\ref{thm:explicit-low-discrepancy-set}
to construct a balanced coloring of $\binom{[n]}{k}.$ For this, we
now develop a connection between these two notions of pseudorandomness.
We will henceforth denote the modulus by $r$ since in our construction,
the modulus is set equal to the number of colors in the coloring of
$\binom{[n]}{k}.$ We start with a technical lemma.
\begin{lem}
\label{lem:equipartition}Fix integers $\ell,k,r$ with $\ell\geq k\geq1$
and $r\geq2.$ Let $Z=\{z_{1},z_{2},\ldots,z_{\ell}\}$ be a multiset
of integers. Then for all $\alpha\in[0,1],$
\begin{multline*}
\max_{a\in\ZZ}\left|\Prob_{S\in\binom{[\ell]}{k}}\left[\sum_{i\in S}z_{i}\equiv a\pmod r\right]-\frac{1}{r}\right|\\
\leq4rk\exp\left(-\frac{\lfloor\ell/k\rfloor\alpha^{2}}{8}\right)+(\disc_{r}(Z)+\alpha)^{k}.
\end{multline*}
\end{lem}

\begin{proof}
Let $\omega$ be a primitive $r$-th root of unity. Then
\begin{align*}
\Prob_{S\in\binom{[\ell]}{k}}\left[\sum_{i\in S}z_{i}\equiv a\pmod r\right] & =\Exp_{S\in\binom{[\ell]}{k}}\I\left[\sum_{i\in S}z_{i}\equiv a\pmod r\right]\\
 & =\Exp_{S\in\binom{[\ell]}{k}}\frac{1}{r}\sum_{t=0}^{r-1}\omega^{t(\sum_{i\in S}z_{i}-a)}\\
 & =\frac{1}{r}\sum_{t=0}^{r-1}\Exp_{S\in\binom{[\ell]}{k}}\omega^{t(\sum_{i\in S}z_{i}-a)}\\
 & =\frac{1}{r}+\frac{1}{r}\sum_{t=1}^{r-1}\Exp_{S\in\binom{[\ell]}{k}}\omega^{t(\sum_{i\in S}z_{i}-a)}\\
 & =\frac{1}{r}+\frac{1}{r}\sum_{t=1}^{r-1}\Exp_{i_{1},i_{2},\ldots,i_{k}}\omega^{t(z_{i_{1}}+z_{i_{2}}+\cdots+z_{i_{k}}-a)},
\end{align*}
where the final expectation is taken over a uniformly random tuple
of indices $i_{1},i_{2},\ldots,i_{k}\in[\ell]$ that are pairwise
distinct. Therefore,
\begin{align}
\left|\Prob_{S\in\binom{[\ell]}{k}}\left[\sum_{i\in S}z_{i}\equiv a\pmod r\right]-\frac{1}{r}\right| & =\left|\frac{1}{r}\sum_{t=1}^{r-1}\Exp_{i_{1},i_{2},\ldots,i_{k}}\omega^{t(z_{i_{1}}+z_{i_{2}}+\cdots+z_{i_{k}}-a)}\right|\nonumber \\
 & \leq\frac{1}{r}\sum_{t=1}^{r-1}\left|\Exp_{i_{1},i_{2},\ldots,i_{k}}\omega^{t(z_{i_{1}}+z_{i_{2}}+\cdots+z_{i_{k}}-a)}\right|\nonumber \\
 & =\frac{1}{r}\sum_{t=1}^{r-1}\left|\Exp_{i_{1},i_{2},\ldots,i_{k}}\prod_{j=1}^{k}\omega^{tz_{i_{j}}}\right|.\label{eq:prob-generate-a}
\end{align}
We now introduce conditioning to make $i_{1},i_{2},\ldots,i_{k}$
independent random variables. Specifically, $i_{1},i_{2},\ldots,i_{k}$
can be generated by the following two-step procedure:
\begin{enumerate}
\item pick uniformly random sets $S_{1},S_{2},\ldots,S_{k}\in\binom{[\ell]}{\lfloor\ell/k\rfloor}$
that are pairwise disjoint;
\item for $j=1,2,\ldots,k,$ pick $i_{j}$ uniformly at random from among
the elements of $S_{j}$.
\end{enumerate}
By symmetry, this procedure generates every tuple $(i_{1},i_{2},\ldots,i_{k})$
of pairwise distinct integers with equal probability. Importantly,
conditioning on $S_{1},S_{2},\ldots,S_{k}$ makes $i_{1},i_{2},\ldots,i_{k}$
independent. Now~(\ref{eq:prob-generate-a}) gives
\begin{align}
 & \max_{a\in\ZZ}\left|\Prob_{S\in\binom{[\ell]}{k}}\left[\sum_{i\in S}z_{i}\equiv a\pmod r\right]-\frac{1}{r}\right|\nonumber \\
 & \qquad\qquad\leq\frac{1}{r}\sum_{t=1}^{r-1}\left|\Exp_{i_{1},i_{2},\ldots,i_{k}}\prod_{j=1}^{k}\omega^{tz_{i_{j}}}\right|\nonumber \\
 & \qquad\qquad=\frac{1}{r}\sum_{t=1}^{r-1}\left|\Exp_{S_{1},S_{2},\ldots,S_{k}}\prod_{j=1}^{k}\Exp_{i_{j}\in S_{j}}\omega^{tz_{i_{j}}}\right|\nonumber \\
 & \qquad\qquad\leq\frac{1}{r}\sum_{t=1}^{r-1}\Exp_{S_{1},S_{2},\ldots,S_{k}}\prod_{j=1}^{k}\left|\Exp_{i_{j}\in S_{j}}\omega^{tz_{i_{j}}}\right|\nonumber \\
 & \qquad\qquad\leq\frac{1}{r}\sum_{t=1}^{r-1}\Exp_{S_{1},S_{2},\ldots,S_{k}}\prod_{j=1}^{k}\disc_{r}(\{z_{i}:i\in S_{j}\})\nonumber \\
 & \qquad\qquad\leq\Exp_{S_{1},S_{2},\ldots,S_{k}}\prod_{j=1}^{k}\disc_{r}(\{z_{i}:i\in S_{j}\}),\label{eq:equipartition-pre-bad}
\end{align}
where $\{z_{i}:i\in S_{j}\}$ for each $j$ is a multiset of cardinality
$\lfloor\ell/k\rfloor.$

Let $B_{j}$ be the event that $\{z_{i}:i\in S_{j}\}$ has $r$-discrepancy
greater than $\disc_{r}(Z)+\alpha,$ and let $B=B_{1}\vee B_{2}\vee\cdots\vee B_{k}$.
Conditioned on $B,$ we get $\prod_{j}\disc_{r}(\{z_{i}:i\in S_{j}\})\leq1$
since $r$-discrepancy is at most $1.$ Conditioned on $\overline{B},$
we have by definition that $\prod_{j}\disc_{r}(\{z_{i}:i\in S_{j}\})\leq(\disc_{r}(Z)+\alpha)^{k}.$
Thus,
\begin{align}
\Exp_{S_{1},S_{2},\ldots,S_{k}}\prod_{j=1}^{k}\disc_{r}(\{z_{i} & :i\in S_{j}\})\leq\Prob_{S_{1},S_{2},\ldots,S_{k}}[B]+(\disc_{r}(Z)+\alpha)^{k}.\label{eq:exp-bad-event}
\end{align}
Recall that $S_{1},S_{2},\ldots,S_{k}$ are identically distributed,
namely, each $S_{j}$ has the distribution of a uniformly random subset
of $[\ell]$ of cardinality $\lfloor\ell/k\rfloor.$ As a result,
Lemma~\ref{lem:discrepancy-subset} guarantees that $B_{j}$ occurs
with probability at most $4r\exp(-\lfloor\ell/k\rfloor\alpha^{2}/8)$.
Applying the union bound over all $j,$ 
\begin{equation}
\Prob_{S_{1},S_{2},\ldots,S_{k}}[B]\leq4rk\exp\left(-\frac{\lfloor\ell/k\rfloor\alpha^{2}}{8}\right).\label{eq:bad-partition}
\end{equation}
Combining~(\ref{eq:equipartition-pre-bad})\textendash (\ref{eq:bad-partition})
concludes the proof.
\end{proof}
We are now in a position to give our general transformation of a low-discrepancy
integer set into a balanced coloring of $\binom{[n]}{k}.$
\begin{thm}[From a low-discrepancy set to a balanced coloring]
\label{thm:discrepancy-to-coloring}Let $n,m,k,r$ be integers with
$n\geq m\geq k\geq1$ and $r\geq2.$ Let $Z=\{z_{1},z_{2},\ldots,z_{n}\}$
be a multiset of integers. Define $\gamma\colon\binom{[n]}{k}\to[r]$
by
\begin{equation}
\gamma(S)=1+\left(\left(\sum_{i\in S}z_{i}\right)\bmod r\right).\label{eq:gamma-defined}
\end{equation}
Let $\beta,\zeta\in[0,1]$ be arbitrary. Then $\gamma$ is $(\epsilon,\delta,m)$-balanced,
where 
\begin{align*}
\epsilon & =4r^{2}k\exp\left(-\frac{\lfloor m/k\rfloor\zeta^{2}}{8}\right)+r(\disc_{r}(Z)+\beta+\zeta)^{k},\\
\delta & =4r\exp\left(-\frac{m\beta^{2}}{8}\right).
\end{align*}
\end{thm}

\begin{proof}
Let $\ell\in\{m,m+1,\ldots,n\}$ be arbitrary. Then Lemma~\ref{lem:discrepancy-subset}
implies that for all but a $\delta$ fraction of the sets $A\in\binom{[n]}{\ell},$
\begin{equation}
\disc_{r}(\{z_{i}:i\in A\})\leq\disc_{r}(Z)+\beta.\label{eq:disc-similar}
\end{equation}
It remains to prove that $\gamma$ is $\epsilon$-balanced on every
set $A\in\binom{[n]}{\ell}$ that satisfies~(\ref{eq:disc-similar}).
We have
\begin{align*}
 & \max_{a\in[r]}\left|\frac{|\gamma^{-1}(a)\cap\binom{A}{k}|}{\binom{|A|}{k}}-\frac{1}{r}\right|\\
 & \qquad\qquad=\max_{a\in[r]}\left|\Prob_{S\in\binom{A}{k}}[\gamma(S)=a]-\frac{1}{r}\right|\\
 & \qquad\qquad=\max_{a\in\ZZ}\left|\Prob_{S\in\binom{A}{k}}\left[\sum_{i\in S}z_{i}\equiv a\pmod r\right]-\frac{1}{r}\right|\\
 & \qquad\qquad\leq4rk\exp\left(-\frac{\lfloor\ell/k\rfloor\zeta^{2}}{8}\right)+(\disc_{r}(\{z_{i}:i\in A\})+\zeta)^{k}\\
 & \qquad\qquad\leq4rk\exp\left(-\frac{\lfloor m/k\rfloor\zeta^{2}}{8}\right)+(\disc_{r}(Z)+\beta+\zeta)^{k}\\
 & \qquad\qquad=\frac{\epsilon}{r},
\end{align*}
where the second step uses the definition of $\gamma,$ the third
step applies Lemma~\ref{lem:equipartition}, the fourth step uses~(\ref{eq:disc-similar})
and~$\ell\geq m$, and the fifth step uses the definition of $\epsilon$.
We have shown that $\gamma$ is $\epsilon$-balanced on $A$, thereby
completing the proof.
\end{proof}

\subsection{\label{subsec:An-explicit-balanced}An explicit balanced coloring}

Theorem~\ref{thm:discrepancy-to-coloring} transforms any integer
set with small $r$-discrepancy into a balanced coloring with $r$
colors. We now apply this transformation to the low-discrepancy integer
set constructed earlier, resulting in an explicit balanced coloring.
\begin{thm}[Explicit balanced coloring]
\label{thm:explicit-coloring} Let $n,m,k,r$ be integers with $n/2\geq m\geq k\geq1$
and $r\geq2.$ Let $\beta,\zeta\in[0,1]$ be arbitrary. Then there
is an $($explicitly given$)$ integer $n'\in(n/2,n]$ and an $($explicitly
given$)$ $(\epsilon,\delta,m)$-balanced coloring $\gamma\colon\binom{[n']}{k}\to[r],$
where 
\begin{align}
\epsilon & =4r^{2}k\exp\left(-\frac{\lfloor m/k\rfloor\zeta^{2}}{8}\right)+r\left(\frac{C^{*}\log n}{n^{1/4}}\cdot\frac{\log r}{1+\log\log r}+\beta+\zeta\right)^{k},\label{eq:epsilon-explicit-coloring}\\
\delta & =4r\exp\left(-\frac{m\beta^{2}}{8}\right),\label{eq:delta-explicit-coloring}
\end{align}
and $C^{*}\geq1$ is the absolute constant from Theorem~\emph{\ref{thm:explicit-low-discrepancy-set}}.
\end{thm}

\begin{proof}
By hypothesis, $n\geq2.$ Invoke Theorem~\ref{thm:explicit-low-discrepancy-set}
with $M=r$ and $t=n$ to obtain an explicit nonempty set $S\subseteq\{1,2,\ldots,r\}$
with
\[
|S|\leq n,
\]
\[
\disc_{r}(S)\leq\frac{C^{*}\log n}{n^{1/4}}\cdot\frac{\log r}{1+\log\log r}.
\]
Let $Z$ be the union of $\lfloor n/|S|\rfloor$ copies of $S.$ Then
$\disc_{r}(Z)=\disc_{r}(S)$ by the definition of $r$-discrepancy.
Letting $n'=|Z|,$ we claim that $n'\in(n/2,n].$ Indeed, the upper
bound is justified by $n'=|S|\cdot\lfloor n/|S|\rfloor\leq n,$ whereas
the lower bound is the arithmetic mean of the bounds $n'\geq|S|$
and $n'>|S|(n/|S|-1).$

Now, let $z_{1},z_{2},\ldots,z_{n'}$ be the elements of $Z$ and
define $\gamma\colon\binom{[n']}{k}\to[r]$ by (\ref{eq:gamma-defined}).
Then Theorem~\ref{thm:discrepancy-to-coloring} implies that $\gamma$
is $(\epsilon,\delta,m)$-balanced with $\epsilon,\delta$ given by~(\ref{eq:epsilon-explicit-coloring})
and~(\ref{eq:delta-explicit-coloring}), respectively.
\end{proof}
Taking $\beta=\zeta=m^{-1/4}$ in Theorem~\ref{thm:explicit-coloring},
we obtain:
\begin{cor}[Explicit balanced coloring]
\label{cor:explicit-coloring} Let $n,m,k,r$ be integers with $n/2\geq m\geq k\geq1$
and $r\geq2.$ Then there is an $($explicitly given$)$ integer $n'\in(n/2,n]$
and an $($explicitly given$)$ $(\epsilon,\delta,m)$-balanced coloring
$\gamma\colon\binom{[n']}{k}\to[r],$ where 
\begin{align}
\epsilon & =4r^{2}k\exp\left(-\frac{\sqrt{m}}{16k}\right)+r\left(\frac{3C^{*}\log^{2}(n+r)}{m^{1/4}}\right)^{k},\label{eq:epsilon-explicit-coloring-1}\\
\delta & =4r\exp\left(-\frac{\sqrt{m}}{8}\right),\label{eq:delta-explicit-coloring-1}
\end{align}
and $C^{*}\geq1$ is the absolute constant from Theorem~\emph{\ref{thm:explicit-low-discrepancy-set}}.
\end{cor}

\noindent The parameters in Corollary~\ref{cor:explicit-coloring}
generously meet our requirements. In our setting of interest, the
integers $n,m,r$ are polynomially related. Thus, we obtain an $(m^{-K},m^{-K},m)$-balanced
coloring for any desired constant $K\geq1$ by invoking Corollary~\ref{cor:explicit-coloring}
with a large enough constant $k=k(K)$.

\section{\label{sec:hardness-amplification}Hardness amplification}

In Section~\ref{sec:Balanced-colorings}, we laid the foundation
for our main result by constructing an explicit integer set with small
discrepancy and transforming it into a highly balanced coloring of
$\binom{[n]}{k}$. In this section, we use this coloring to design
a hardness amplification method for approximate degree and its one-sided
variant.

\subsection{Pseudodistributions from balanced colorings}

Recall from the introduction that our approach centers around encoding
the vectors $e_{1},e_{2},\ldots,e_{N},0^{N}$ as $n$-bit strings
with $n\ll N$ so as to make the decoding easy for circuits but hard
for low-degree polynomials. The construction of this code requires
several steps. As a first step, we show how to convert any balanced
coloring of $\binom{[n]}{k}$ with $r$ colors into an explicit sequence
of functions $\phi_{1},\phi_{2},\ldots,\phi_{r}\colon\zoon\to\Re$
that are almost everywhere nonnegative, are supported almost entirely
on pairwise disjoint sets of strings of Hamming weight $k$, and are
pairwise indistinguishable by low-degree polynomials. We call them
\emph{pseudodistributions} to highlight the fact that each $\phi_{i}$
has $\ell_{1}$ norm approximately~$1$, nearly all of it coming
from the points where $\phi_{i}$ is nonnegative.
\begin{thm}
\label{thm:encoding-1-r}Let $\epsilon,\delta\in[0,1)$ be given.
Let $n,m,k,r$ be positive integers with $n\geq m>k$. Let $\gamma\colon\binom{[n]}{k}\to[r]$
be a given $(\epsilon,\delta,m)$-balanced coloring. Then there are
$($explicitly given$)$ functions $\phi_{1},\phi_{2},\ldots,\phi_{r}\colon\zoon\to\Re$
with the following properties.
\begin{enumerate}[topsep=3mm,labelsep=3mm,parsep=3mm,label=\textup{(\roman*)}]
\item \textbf{\label{enu:Support}Support:} $\supp\phi_{i}\subseteq\{x\in\zoon:|x|=k\text{ or }|x|\geq m\};$
\item \textbf{\label{enu:Essential-support}Essential support: $\zoon|_{k}\cap\supp\phi_{i}=\{\1_{S}:S\in\gamma^{-1}(i)\};$}
\item \textbf{\label{enu:Nonnegativity}Nonnegativity:} $\phi_{i}\geq0$
on $\zoon|_{k};$
\item \textbf{\label{enu:Normalization}Normalization:} $\sum_{x:|x|=k}\phi_{i}(x)=1;$
\item \textbf{\label{enu:Tail-bound}Tail bound:} $\sum_{x:|x|\ne k}|\phi_{i}(x)|\leq(8\epsilon+4r\delta)/(1-\epsilon);$
\item \textbf{\label{enu:Pointwise-bound}Graded bound: }for some absolute
constant $c'\in(0,1),$
\begin{align*}
\sum_{x:|x|=\ell}|\phi_{i}(x)| & \leq\frac{\epsilon+r\delta}{1-\epsilon}\cdot\frac{m^{2}}{c'\ell^{2}}\cdot\exp\left(-\frac{c'(\ell-k)}{\sqrt{nm}}\right), &  & \ell>k;
\end{align*}
\item \textbf{\label{enu:Orthogonality}Orthogonality: }for some absolute
constant $c''\in(0,1),$
\begin{align*}
 & \orth(\phi_{i}-\phi_{j})\ge c''\sqrt{\frac{n}{m}}, &  & i,j\in[r].
\end{align*}
\end{enumerate}
\end{thm}

\begin{proof}
Define
\begin{align}
\Delta & =m-k,\label{eq:Delta}\\
D & =\left\lfloor \frac{n-k}{\Delta}\right\rfloor .\label{eq:D}
\end{align}
Setting $\epsilon=1/2$ in Lemma~\ref{lem:SW-OR-dual} gives an explicit
function $\omega\colon\{0,1,2,\ldots,D\}\to\Re$ with
\begin{align}
 & \omega(0)>\frac{1}{4}\|\omega\|_{1},\label{eq:construction-omega-at-zero}\\
 & |\omega(t)|\leq\frac{1}{ct^{2}\,2^{ct/\sqrt{D}}}\cdot\|\omega\|_{1} &  & (t=1,2,\ldots,D),\label{eq:construction-omega-upper-bound}\\
 & \orth\omega\geq c\sqrt{D},\label{eq:construction-omega-orthog}
\end{align}
where $0<c<1$ is an absolute constant. For convenience of notation,
we will extend $\omega$ to all of $\Re$ by setting $\omega(t)=0$
for $t\notin\{0,1,2,\ldots,D\}.$ With this extension,~(\ref{eq:construction-omega-upper-bound})
gives
\begin{align}
|\omega(t)| & \leq\frac{1}{ct^{2}\,2^{ct/\sqrt{D}}}\cdot\|\omega\|_{1}, &  & t\in[1,\infty).\label{eq:omega-construction-upper-extended}
\end{align}

For $S\in\binom{[n]}{k},$ define an auxiliary dual object $\phi_{S}\colon\zoon\to\Re$
by
\begin{equation}
\phi_{S}(x)=\binom{n-k}{|x|-k}^{-1}\omega(0)^{-1}\omega\left(\frac{|x|-k}{\Delta}\right)\prod_{i\in S}x_{i}.\label{eq:phi-S-definition}
\end{equation}
Then
\begin{align}
 & \phi_{S}(\1_{S})=1, &  & S\in\binom{[n]}{k}.\label{eq:phi-S-1-S}
\end{align}
Since $\phi_{S}(x)=0$ unless $x|_{S}=1^{k},$ we see that $\1_{S}$
is in fact the only input of Hamming weight $k$ at which $\phi_{S}$
is nonzero:
\begin{align}
\phi_{S}(\1_{T}) & =\delta_{S,T}, &  & S,T\in\binom{[n]}{k}.\label{eq:phi_S-1-T}
\end{align}
Since $\supp\omega\subseteq\{0,1,2,\ldots,D\},$ the only inputs $x$
other than $\1_{S}$ in the support of $\phi_{S}$ have Hamming weight
$|x|\in\{$$k+\Delta,k+2\Delta,\ldots,k+D\Delta\},$ so that in particular
$|x|\geq m.$ In summary,
\begin{align}
\supp\phi_{S} & \subseteq\{x:x=\1_{S}\text{ or }|x|\geq m\}, &  & S\in\binom{[n]}{k},\label{eq:phi-S-supp}\\
\supp\phi_{S} & \subseteq\bigcup_{i=0}^{D}\{x:|x|=k+i\Delta\}, &  & S\in\binom{[n]}{k}.\label{eq:phi-i-supp-multiples-of-Delta}
\end{align}

We now turn to the construction of the $\phi_{i}$. By definition
of an $(\epsilon,\delta,m)$-balanced coloring, the given coloring
$\gamma\colon\binom{[n]}{k}\to[r]$ satisfies
\begin{multline}
\Prob_{A\in\binom{[n]}{\ell}}\left[\frac{1-\epsilon}{r}\binom{|A|}{k}\leq\left|\gamma^{-1}(i)\cap\binom{A}{k}\right|\leq\frac{1+\epsilon}{r}\binom{|A|}{k}\right]\geq1-\delta,\\
\ell=m,m+1,\ldots,n.\label{eq:balanced-coloring-property}
\end{multline}
Since $\delta<1,$ taking $\ell=n$ in this equation leads to
\begin{align}
\left||\gamma^{-1}(i)|-\frac{1}{r}\binom{n}{k}\right| & \leq\frac{\epsilon}{r}\binom{n}{k}, &  & i\in[r],\label{eq:gamma-equipartition}
\end{align}
and in particular
\begin{align}
|\gamma^{-1}(i)| & \geq\frac{1-\epsilon}{r}\binom{n}{k}, &  & i\in[r].\label{eq:alpha-fair-1}
\end{align}
For $i=1,2,\ldots,r,$ we define $\phi_{i}\colon\zoon\to\Re$ by
\[
\phi_{i}(x)=\Exp_{S\in\gamma^{-1}(i)}\phi_{S}(x)-\I[|x|\geq m]\Exp_{S\in\binom{[n]}{k}}\phi_{S}(x).
\]
This definition is legitimate since $\gamma^{-1}(i)\ne\varnothing$
for every $i$ due to (\ref{eq:alpha-fair-1}) and $\epsilon<1.$
\begin{claim}
\label{claim:balance}For all $i\in[r]$ and $\ell\in\{m,m+1,\ldots,n\},$
\begin{align*}
 & \Exp_{A\in\binom{[n]}{\ell}}\left|\Prob_{S\in\gamma^{-1}(i)}[S\subseteq A]-\Prob_{S\in\binom{[n]}{k}}[S\subseteq A]\right|\leq\frac{2\epsilon+r\delta}{1-\epsilon}\binom{n}{k}^{-1}\binom{\ell}{k}.
\end{align*}
\end{claim}

\begin{proof}
Fix $i\in[r]$ and $\ell\in\{m,m+1,\ldots,n\}$ arbitrarily for the
remainder of the proof. Let $A\in\binom{[n]}{\ell}$ be uniformly
random. If $\gamma$ is $\epsilon$-balanced on $A,$ then by definition
\begin{align*}
\left|\left|\gamma^{-1}(i)\cap\binom{A}{k}\right|-\frac{1}{r}\binom{|A|}{k}\right| & \leq\frac{\epsilon}{r}\binom{|A|}{k}.
\end{align*}
If $\gamma$ is not $\epsilon$-balanced on $A,$ we have the trivial
bound
\begin{align*}
\left|\left|\gamma^{-1}(i)\cap\binom{A}{k}\right|-\frac{1}{r}\binom{|A|}{k}\right| & \leq\binom{|A|}{k}.
\end{align*}
Combining these two equations, we arrive at 
\begin{align}
\left|\left|\gamma^{-1}(i)\cap\binom{A}{k}\right|-\frac{1}{r}\binom{|A|}{k}\right| & \leq\left(\frac{\epsilon}{r}+Y_{A}\right)\binom{|A|}{k}\label{eq:gamma-eps-balanced-on-A}
\end{align}
for all $A,$ where $Y_{A}$ is the indicator random variable for
the event that $\gamma$ is not $\epsilon$-balanced on $A.$ Since
$\gamma$ is $(\epsilon,\delta,m)$-balanced, we further have
\begin{equation}
\Exp_{A\in\binom{[n]}{\ell}}Y_{A}\leq\delta.\label{eq:balance-is-common}
\end{equation}
Now
\begin{align}
 & \left|\Prob_{S\in\gamma^{-1}(i)}[S\subseteq A]-\Prob_{S\in\binom{[n]}{k}}[S\subseteq A]\right|\nonumber \\
 & \qquad\qquad=\left|\frac{|\gamma^{-1}(i)\cap\binom{A}{k}|}{|\gamma^{-1}(i)|}-\frac{\binom{|A|}{k}}{\binom{n}{k}}\right|\nonumber \\
 & \qquad\qquad=\frac{1}{|\gamma^{-1}(i)|}\binom{n}{k}^{-1}\left|\left|\gamma^{-1}(i)\cap\binom{A}{k}\right|\binom{n}{k}-|\gamma^{-1}(i)|\binom{|A|}{k}\right|\nonumber \\
 & \qquad\qquad\leq\frac{r}{1-\epsilon}\binom{n}{k}^{-2}\left|\left|\gamma^{-1}(i)\cap\binom{A}{k}\right|\binom{n}{k}-|\gamma^{-1}(i)|\binom{|A|}{k}\right|\nonumber \\
 & \qquad\qquad\leq\frac{r}{1-\epsilon}\binom{n}{k}^{-2}\left|\left|\gamma^{-1}(i)\cap\binom{A}{k}\right|\binom{n}{k}-\frac{1}{r}\binom{|A|}{k}\binom{n}{k}\right|\nonumber \\
 & \qquad\qquad\quad\qquad\hfill+\frac{r}{1-\epsilon}\binom{n}{k}^{-2}\left|\frac{1}{r}\binom{|A|}{k}\binom{n}{k}-|\gamma^{-1}(i)|\binom{|A|}{k}\right|\nonumber \\
 & \qquad\qquad\leq\frac{r}{1-\epsilon}\binom{n}{k}^{-2}\left(\frac{\epsilon}{r}+\frac{\epsilon}{r}+Y_{A}\right)\binom{|A|}{k}\binom{n}{k}\nonumber \\
 & \qquad\qquad=\frac{r}{1-\epsilon}\binom{n}{k}^{-1}\left(\frac{2\epsilon}{r}+Y_{A}\right)\binom{\ell}{k},\label{eq:intermediate-gamma-A}
\end{align}
where the third step is valid by~(\ref{eq:alpha-fair-1}), the fourth
step applies the triangle inequality, the fifth step uses~(\ref{eq:gamma-equipartition})
and~(\ref{eq:gamma-eps-balanced-on-A}), and the last step uses $|A|=\ell.$
It remains to pass to expectations with respect to $A$: 
\begin{align*}
 & \Exp_{A\in\binom{[n]}{\ell}}\left|\Prob_{S\in\gamma^{-1}(i)}[S\subseteq A]-\Prob_{S\in\binom{[n]}{k}}[S\subseteq A]\right|\\
 & \qquad\qquad\leq\Exp_{A\in\binom{[n]}{\ell}}\frac{r}{1-\epsilon}\binom{n}{k}^{-1}\left(\frac{2\epsilon}{r}+Y_{A}\right)\binom{\ell}{k}\\
 & \qquad\qquad=\frac{r}{1-\epsilon}\binom{n}{k}^{-1}\left(\frac{2\epsilon}{r}+\Exp_{A\in\binom{[n]}{\ell}}Y_{A}\right)\binom{\ell}{k}\\
 & \qquad\qquad\leq\frac{2\epsilon+r\delta}{1-\epsilon}\binom{n}{k}^{-1}\binom{\ell}{k},
\end{align*}
where the last step uses~(\ref{eq:balance-is-common}).
\end{proof}
\begin{claim}
\label{claim:pointwise-ell}For each $i\in[r]$ and $\ell\in\{m,m+1,\ldots,n\}$,
\[
\sum_{x:|x|=\ell}|\phi_{i}(x)|\leq\frac{2\epsilon+r\delta}{1-\epsilon}\cdot\left|\frac{\omega(\frac{\ell-k}{\Delta})}{\omega(0)}\right|.
\]
\end{claim}

\begin{proof}
Fix $i\in[r]$ and $\ell\in\{m,m+1,\ldots,n\}$ arbitrarily for the
remainder of the proof. Consider any input $x=\1_{A}$ with $|A|=\ell.$
In this case, the definition of $\phi_{i}$ simplifies to
\begin{align*}
\phi_{i}(\1_{A}) & =\Exp_{S\in\gamma^{-1}(i)}\phi_{S}(\1_{A})-\Exp_{S\in\binom{[n]}{k}}\phi_{S}(\1_{A}).
\end{align*}
Recall from~(\ref{eq:phi-S-definition}) that
\[
\phi_{S}(\1_{A})=\frac{\omega(\frac{\ell-k}{\Delta})}{\omega(0)\binom{n-k}{\ell-k}}\cdot\I[S\subseteq A].
\]
As a result,
\[
\phi_{i}(\1_{A})=\frac{\omega(\frac{\ell-k}{\Delta})}{\omega(0)\binom{n-k}{\ell-k}}\left(\Prob_{S\in\gamma^{-1}(i)}[S\subseteq A]-\Prob_{S\in\binom{[n]}{k}}[S\subseteq A]\right).
\]
Passing to absolute values and summing over $A\in\binom{[n]}{\ell},$
we obtain
\begin{align*}
\sum_{A\in\binom{[n]}{\ell}}|\phi_{i}(\1_{A})| & =\left|\frac{\omega(\frac{\ell-k}{\Delta})}{\omega(0)\binom{n-k}{\ell-k}}\right|\sum_{A\in\binom{[n]}{\ell}}\left|\Prob_{S\in\gamma^{-1}(i)}[S\subseteq A]-\Prob_{S\in\binom{[n]}{k}}[S\subseteq A]\right|\\
 & \leq\left|\frac{\omega(\frac{\ell-k}{\Delta})}{\omega(0)\binom{n-k}{\ell-k}}\right|\cdot\binom{n}{\ell}\cdot\frac{2\epsilon+r\delta}{1-\epsilon}\cdot\binom{n}{k}^{-1}\binom{\ell}{k}\\
 & =\left|\frac{\omega(\frac{\ell-k}{\Delta})}{\omega(0)}\right|\cdot\frac{2\epsilon+r\delta}{1-\epsilon},
\end{align*}
where the second step applies Claim~\ref{claim:balance}, and the
final step is justified by
\begin{align*}
 & \binom{n-k}{\ell-k}^{-1}\binom{n}{\ell}\binom{n}{k}^{-1}\binom{\ell}{k}\\
 & \qquad=\frac{(\ell-k)!\;(n-\ell)!}{(n-k)!}\cdot\frac{n!}{\ell!\;(n-\ell)!}\cdot\frac{k!\;(n-k)!}{n!}\cdot\frac{\ell!}{k!\;(\ell-k)!}=1. &  & \qedhere
\end{align*}
\end{proof}
We now turn to the verification of properties~\ref{enu:Support}\textendash \ref{enu:Orthogonality}
in the theorem statement.\medskip{}

\textbf{Properties~\ref{enu:Support}\textendash \ref{enu:Normalization}.}
Equation~(\ref{eq:phi-S-supp}) shows that $\phi_{i}$ is a linear
combination of functions whose support is contained in $\{x:|x|=k\text{ or }|x|\geq m\}.$
This settles the support requirement~\ref{enu:Support}. For $T\in\binom{[n]}{k}$,
\begin{align}
\phi_{i}(\1_{T}) & =\Exp_{S\in\gamma^{-1}(i)}\phi_{S}(\1_{T})\nonumber \\
 & =\Exp_{S\in\gamma^{-1}(i)}\delta_{S,T}\nonumber \\
 & =\frac{\I[T\in\gamma^{-1}(i)]}{|\gamma^{-1}(i)|}, &  & T\in\binom{[n]}{k},\label{eq:phi-i-on-1-T}
\end{align}
where the first step is immediate from the defining equation for $\phi_{i},$
and the second step applies~(\ref{eq:phi_S-1-T}). The essential
support property~\ref{enu:Essential-support} and nonnegativity property~\ref{enu:Nonnegativity}
are now immediate from~(\ref{eq:phi-i-on-1-T}). The normalization
requirement~\ref{enu:Normalization} follows by summing~(\ref{eq:phi-i-on-1-T})
over $T\in\binom{[n]}{k}.$

\medskip{}

\textbf{Properties~\ref{enu:Tail-bound} and \ref{enu:Pointwise-bound}.}
The tail bound~\ref{enu:Tail-bound} for $i\in[r]$ can be seen as
follows:
\begin{align*}
\sum_{x:|x|\ne k}|\phi_{i}(x)| & =\sum_{x:|x|\geq m}|\phi_{i}(x)|\\
 & =\sum_{\ell=m}^{n}\,\sum_{x:|x|=\ell}|\phi_{i}(x)|\\
 & \leq\frac{2\epsilon+r\delta}{1-\epsilon}\cdot\sum_{\ell=m}^{n}\left|\omega(0)^{-1}\omega\left(\frac{\ell-k}{\Delta}\right)\right|\\
 & \leq\frac{2\epsilon+r\delta}{1-\epsilon}\cdot\frac{\|\omega\|_{1}}{|\omega(0)|}\\
 & \leq\frac{8\epsilon+4r\delta}{1-\epsilon},
\end{align*}
where the first step uses the support property~\ref{enu:Support},
the third step is valid by Claim~\ref{claim:pointwise-ell}, and
the last step applies~(\ref{eq:construction-omega-at-zero}).

The graded bound~\ref{enu:Pointwise-bound} for $\ell\in(k,m)$ holds
trivially since $\phi_{i}$ vanishes on inputs of Hamming weight in
$(k,m),$ by the support property~\ref{enu:Support}. The validity
of ~\ref{enu:Pointwise-bound} for $\ell\geq m$ is borne out by
\begin{align*}
\sum_{x:|x|=\ell}|\phi_{i}(x)| & \leq\frac{2\epsilon+r\delta}{1-\epsilon}\cdot\left|\omega(0)^{-1}\omega\left(\frac{\ell-k}{\Delta}\right)\right|\\
 & \leq\frac{8\epsilon+4r\delta}{1-\epsilon}\cdot\frac{1}{\|\omega\|_{1}}\cdot\left|\omega\left(\frac{\ell-k}{\Delta}\right)\right|\\
 & \le\frac{8\epsilon+4r\delta}{1-\epsilon}\cdot\frac{1}{c\left(\frac{\ell-k}{\Delta}\right)^{2}2^{c(\ell-k)/(\Delta\sqrt{D})}}\\
 & =\frac{8\epsilon+4r\delta}{1-\epsilon}\cdot\frac{1}{c\left(\frac{\ell-k}{m-k}\right)^{2}\,2^{c(\ell-k)/((m-k)\sqrt{\lfloor(n-k)/(m-k)\rfloor})}}\\
 & \leq\frac{8\epsilon+4r\delta}{1-\epsilon}\cdot\frac{m^{2}}{c\ell^{2}\,2^{c(\ell-k)/\sqrt{nm}}},
\end{align*}
where the first step restates Claim~\ref{claim:pointwise-ell}, the
second step is justified by~(\ref{eq:construction-omega-at-zero}),
the third step appeals to~(\ref{eq:omega-construction-upper-extended}),
and the fourth step substitutes the values from~(\ref{eq:Delta})
and~(\ref{eq:D}).

\medskip{}

\textbf{Property~\ref{enu:Orthogonality}.} To begin with, we claim
that 
\begin{align}
\orth\phi_{S} & \geq c\sqrt{D}, &  & S\in\binom{[n]}{k}.\label{eq:phi-S-orth}
\end{align}
Indeed, let $p$ be a real polynomial on $\zoon$ with $\deg p<c\sqrt{D}$.
By linearity, it suffices to consider polynomials $p$ that factor
as $p(x)=p_{1}(x|_{S})p_{2}(x|_{\overline{S}})$ for some nonzero
polynomials $p_{1},p_{2}$. Now, Minsky and Papert's symmetrization
argument (Proposition~\ref{prop:minsky-papert}) guarantees that
\begin{align}
\Exp_{\substack{y\in\zoo^{n-k}\\
|y|=i
}
}p_{2}(y) & =p_{2}^{*}(i), &  & i=0,1,2,\ldots,n-k,\label{eq:p-2-star}
\end{align}
for some univariate polynomial $p_{2}^{*}$ of degree at most $\deg p_{2}$.
As a result,
\begin{align*}
\langle\phi_{S},p\rangle & =\sum_{\substack{x\in\zoon:\\
x|_{S}=1^{k}
}
}\phi_{S}(x)p(x)\\
 & =\sum_{i=0}^{D}\sum_{\substack{x\in\zoon:\\
|x|=k+i\Delta,\;x|_{S}=1^{k}
}
}\phi_{S}(x)p(x)\\
 & =\sum_{i=0}^{D}\sum_{\substack{x\in\zoon:\\
|x|=k+i\Delta,\;x|_{S}=1^{k}
}
}\binom{n-k}{i\Delta}^{-1}\frac{\omega(i)}{\omega(0)}\cdot p(x)\\
 & =\sum_{i=0}^{D}\;\Exp_{\substack{y\in\zoo^{n-k}:\\
|y|=i\Delta
}
}\left[\frac{\omega(i)}{\omega(0)}\cdot p_{1}(1^{k})p_{2}(y)\right]\\
 & =\frac{p_{1}(1^{k})}{\omega(0)}\sum_{i=0}^{D}\omega(i)p_{2}^{*}(i\Delta)\\
 & =0,
\end{align*}
where the first and third steps use the definition of $\phi_{S},$
the second step is justified by~(\ref{eq:phi-i-supp-multiples-of-Delta}),
the next-to-last step uses~(\ref{eq:p-2-star}), and the last step
is valid by~(\ref{eq:construction-omega-orthog}) since $\deg p_{2}^{*}\leq\deg p_{2}\leq\deg p<c\sqrt{D}.$
This settles~(\ref{eq:phi-S-orth}).

Now the orthogonality requirement~\ref{enu:Orthogonality} can be
seen as follows:
\begin{align*}
\orth(\phi_{i}-\phi_{j}) & =\orth\left(\Exp_{S\in\gamma^{-1}(i)}\phi_{S}-\Exp_{S\in\gamma^{-1}(j)}\phi_{S}\right)\\
 & \geq\min_{S\in\binom{[n]}{k}}\orth\phi_{S}\\
 & \geq c\sqrt{D}\\
 & =c\sqrt{\left\lfloor \frac{n-k}{m-k}\right\rfloor }\\
 & \geq c\sqrt{\left\lfloor \frac{n}{m}\right\rfloor },
\end{align*}
where the second step uses Proposition~\ref{prop:orth}\ref{item:orth-sum},
the third step is valid by~(\ref{eq:phi-S-orth}), the fourth step
applies the definition of $D,$ and the last step uses $n\geq m.$
\end{proof}

\subsection{Encoding via indistinguishable distributions}

As our next step, we will show that the pseudodistributions $\phi_{1},\phi_{2},\ldots,\phi_{r}$
in Theorem~\ref{thm:encoding-1-r} can be turned into actual probability
distributions $\lambda_{1},\lambda_{2},\ldots,\lambda_{r}$ provided
that the underlying coloring of $\binom{[n]}{k}$ is sufficiently
balanced. The resulting distributions $\lambda_{i}$ inherit all the
desirable analytic properties established for the $\phi_{i}$ in Theorem~\ref{thm:encoding-1-r}.
Specifically, the $\lambda_{i}$ are supported almost entirely on
pairwise disjoint sets of inputs of Hamming weight $k$ and are pairwise
indistinguishable by low-degree polynomials.
\begin{thm}
\label{thm:encoding-1-r-normalized}Let $0<\beta<1$ be given. Let
$n,n',m,k,r$ be positive integers with $n\geq n'\geq m>k$. Let $\gamma\colon\binom{[n']}{k}\to[r]$
be a given $(\frac{\beta}{16rm^{2}},\frac{\beta}{16r^{2}m^{2}},m)$-balanced
coloring. Then there are $($explicitly given$)$ probability distributions
$\lambda_{1},\lambda_{2},\ldots,\lambda_{r}$ on $\zoon$ such that
\begin{align}
 & \supp\lambda_{i}\subseteq\{x\in\zoon:|x|=k\text{ or }|x|\geq m\}, &  & i\in[r],\label{eq:mu-i-support}\\
 & \zoon|_{k}\cap\supp\lambda_{i}=\{\1_{S}:S\in\gamma^{-1}(i)\}, &  & i\in[r],\label{eq:mu-i-essential-support}\\
 & \lambda_{i}(\zoon|_{k})\geq1-\beta, &  & i\in[r],\label{eq:mu-i-kth-level}\\
 & \lambda_{i}(\zoon|_{\ell})\leq\frac{\exp(-c(\ell-k)/\sqrt{n'm})}{c(\ell-k+1)^{2}}, &  & i\in[r],\;\ell\geq k,\label{eq:mu-i-graded-bound}\\
 & \orth(\lambda_{i}-\lambda_{j})\geq c\sqrt{\frac{n'}{m},} &  & i,j\in[r],\label{eq:mu-i-orth}
\end{align}
where $c\in(0,1)$ is an absolute constant, independent of $n,n',m,k,r,\beta.$
\end{thm}

\begin{proof}
By hypothesis, $\gamma$ is $(\epsilon,\delta,m)$-balanced with
\begin{align*}
 & \epsilon=\frac{\beta}{16rm^{2}},\\
 & \delta=\frac{\beta}{16r^{2}m^{2}}.
\end{align*}
Applying Theorem~\ref{thm:encoding-1-r} with these parameters gives
functions $\phi_{1},\phi_{2},\ldots,\phi_{r}\colon\zoo^{n'}\to\Re$
that obey
\begin{align}
 & \supp\phi_{i}\subseteq\{x\in\zoo^{n'}:|x|=k\text{ or }|x|\geq m\},\label{eq:phi-supp-restated}\\
 & \zoo^{n'}|_{k}\cap\supp\phi_{i}=\{\1_{S}:S\in\gamma^{-1}(i)\},\label{eq:phi-essential-support-restated}\\
 & \phi_{i}\geq0\quad\text{on }\zoo^{n'}|_{k},\label{eq:phi-nonneg-restated}\\
 & \sum_{x:|x|=k}\phi_{i}(x)=1,\label{eq:phi-normalization-restated}\\
 & \sum_{x:|x|\ne k}|\phi_{i}(x)|\leq\frac{\beta}{r},\label{eq:phi-tail-restated}\\
 & \sum_{x:|x|=\ell}|\phi_{i}(x)|\leq\frac{\beta}{rc'\ell^{2}}\cdot\exp\left(-\frac{c'(\ell-k)}{\sqrt{n'm}}\right),\qquad\ell>k,\label{eq:phi-graded-restated}\\
 & \orth(\phi_{i}-\phi_{j})\ge c''\sqrt{\frac{n'}{m}}\qquad\text{for all }i,j\in[r],\label{eq:phi-orth-restated}
\end{align}
where $c',c''\in(0,1)$ are the absolute constants defined in Theorem~\ref{thm:encoding-1-r}.
For $i\in[r],$ define $\tilde{\phi}_{i}\colon\zoo^{n'}\to\Re$ by
\begin{equation}
\tilde{\phi}_{i}(x)=\phi_{i}(x)-\I[|x|>k]\min_{j\in[r]}\phi_{j}(x).\label{eq:tilde-phi-def}
\end{equation}
Equation~(\ref{eq:phi-supp-restated}) shows that $\tilde{\phi}_{i}$
is a linear combination of functions whose support is contained in
$\{x\in\zoo^{n'}:|x|=k\text{ or }|x|\geq m\}.$ As a result,
\begin{align}
 & \supp\tilde{\phi}_{i}\subseteq\{x\in\zoo^{n'}:|x|=k\text{ or }|x|\geq m\}.\label{eq:tilde-phi-support}
\end{align}
Since $\tilde{\phi}_{i}=\phi_{i}$ on $\zoo^{n'}|_{k},$ we obtain
from (\ref{eq:phi-essential-support-restated}) and~(\ref{eq:phi-normalization-restated})
that
\begin{align}
 & \zoo^{n'}|_{k}\cap\supp\tilde{\phi}_{i}=\{\1_{S}:S\in\gamma^{-1}(i)\},\label{eq:tilde-phi-disjoint}\\
 & \sum_{x:|x|=k}\tilde{\phi}_{i}(x)=1.\label{eq:tilde-phi-normalization}
\end{align}
In particular,
\begin{equation}
\|\tilde{\phi}_{i}\|_{1}\geq1.\label{eq:tilde-phi-ell1-norm}
\end{equation}
We further claim that
\begin{align}
\tilde{\phi}_{i}(x) & \geq0, &  & x\in\zoo^{n'}.\label{eq:tilde-phi-nonnegative}
\end{align}
Indeed, the nonnegativity of $\tilde{\phi}_{i}(x)$ for $x\in\zoo^{n'}|_{k}$
follows from $\tilde{\phi}_{i}(x)=\phi_{i}(x)$ and~(\ref{eq:phi-nonneg-restated}),
whereas the nonnegativity of $\tilde{\phi}_{i}(x)$ for $x\in\zoo^{n'}|_{>k}$
follows from~(\ref{eq:tilde-phi-def}) via $\tilde{\phi}_{i}(x)=\phi_{i}(x)-\min_{j\in[r]}\phi_{j}(x)\geq\phi_{i}(x)-\phi_{i}(x)\geq0.$

On $\zoo^{n'}|_{>k},$ we have
\begin{align*}
\tilde{\phi}_{i} & =\phi_{i}-\min_{j\in[r]}\phi_{j}=\max_{j\in[r]}\{\phi_{i}-\phi_{j}\}\leq\max_{j\in[r]}|\phi_{i}-\phi_{j}|\leq\sum_{j=1}^{r}|\phi_{j}|.
\end{align*}
This conclusion is also valid on $\zoo^{n'}|_{<k}$ due to~(\ref{eq:tilde-phi-support}).
Thus, 
\begin{align}
\tilde{\phi}_{i}(x) & \leq\sum_{j=1}^{r}|\phi_{j}(x)|, &  & |x|\ne k.\label{eq:tilde-phi-non-k}
\end{align}
Summing over $x$ gives
\begin{align}
\sum_{x:|x|\ne k}\tilde{\phi}_{i}(x) & \leq\sum_{x:|x|\ne k}\sum_{j=1}^{r}|\phi_{j}(x)|\nonumber \\
 & =\sum_{j=1}^{r}\sum_{x:|x|\ne k}|\phi_{j}(x)|\nonumber \\
 & \leq\beta,\label{eq:tilde-phi-non-k-sum}
\end{align}
where the third step applies~(\ref{eq:phi-tail-restated}).

For all $i\in[r]$ and $\ell\in\{m,m+1,\ldots,n'\}$, we have the
graded bound
\begin{align}
\sum_{x:|x|=\ell}\tilde{\phi}_{i}(x) & \leq\sum_{j=1}^{r}\sum_{x:|x|=\ell}|\phi_{j}(x)|\nonumber \\
 & \leq\frac{1}{c'\ell^{2}}\cdot\exp\left(-\frac{c'(\ell-k)}{\sqrt{n'm}}\right),\label{eq:tilde-phi-graded}
\end{align}
where the first step uses~(\ref{eq:tilde-phi-non-k}), and the second
step uses~(\ref{eq:phi-graded-restated}). Finally, for $i,j\in[r],$
we have
\begin{align}
\orth(\tilde{\phi}_{i}-\tilde{\phi}_{j}) & =\orth(\phi_{i}-\phi_{j})\nonumber \\
 & \geq c''\sqrt{\frac{n'}{m}},\label{eq:tilde-phi-orthog}
\end{align}
where the first step uses the definition~(\ref{eq:tilde-phi-def}),
and the second step uses~(\ref{eq:phi-orth-restated}).

Define $c=\min\{c',c''\}.$ Equations~(\ref{eq:tilde-phi-normalization})
and~(\ref{eq:tilde-phi-nonnegative}) show that each $\tilde{\phi}_{i}$
is a nonnegative function and is not identically zero, making it possible
to define a probability distribution $\lambda_{i}$ on $\zoon$ by
\begin{align*}
\lambda_{i}(x) & =\frac{1}{\|\tilde{\phi}_{i}\|_{1}}\tilde{\phi}_{i}(x_{1}x_{2}\ldots x_{n'})\prod_{j=n'+1}^{n}(1-x_{j}).
\end{align*}
In other words, $\lambda_{i}$ is nonzero only on inputs $x$ with
$x_{n'+1}=x_{n'+2}=\cdots=x_{n}=0,$ and on such inputs $\lambda_{i}(x)$
is the properly normalized version of the nonnegative function $\tilde{\phi}_{i}(x_{1}x_{2}\ldots x_{n'})$.
Then properties~(\ref{eq:mu-i-support}) and~(\ref{eq:mu-i-essential-support})
are immediate from~(\ref{eq:tilde-phi-support}) and~(\ref{eq:tilde-phi-disjoint}),
respectively. Property~(\ref{eq:mu-i-kth-level}) follows from
\begin{align*}
\lambda_{i}(\zoon|_{k}) & =\frac{1}{\|\tilde{\phi}_{i}\|_{1}}\sum_{x\in\zoon|_{k}}\tilde{\phi}_{i}(x_{1}x_{2}\ldots x_{n'})\prod_{j=n'+1}^{n}(1-x_{j})\\
 & =\frac{1}{\|\tilde{\phi}_{i}\|_{1}}\sum_{x\in\zoo^{n'}|_{k}}\tilde{\phi}_{i}(x_{1}x_{2}\ldots x_{n'})\\
 & =\frac{1}{\|\tilde{\phi}_{i}\|_{1}}\\
 & \geq\frac{1}{1+\beta}\\
 & \geq1-\beta,
\end{align*}
where the third step uses~(\ref{eq:tilde-phi-normalization}), and
the fourth step uses~(\ref{eq:tilde-phi-normalization}) and~(\ref{eq:tilde-phi-non-k-sum}).
Property~(\ref{eq:mu-i-graded-bound}) is trivial for $\ell=k$ and
follows for $\ell>k$ from
\begin{align*}
\lambda_{i}(\zoon|_{\ell}) & =\frac{1}{\|\tilde{\phi}_{i}\|_{1}}\sum_{x\in\zoon|_{\ell}}\tilde{\phi}_{i}(x_{1}x_{2}\ldots x_{n'})\prod_{j=n'+1}^{n}(1-x_{j})\\
 & =\frac{1}{\|\tilde{\phi}_{i}\|_{1}}\sum_{x\in\zoo^{n'}|_{\ell}}\tilde{\phi}_{i}(x_{1}x_{2}\ldots x_{n'})\\
 & \leq\frac{1}{\|\tilde{\phi}_{i}\|_{1}}\cdot\frac{\exp(-c'(\ell-k)/\sqrt{n'm})}{c'\ell^{2}}\\
 & \leq\frac{\exp(-c(\ell-k)/\sqrt{n'm})}{c\ell^{2}},
\end{align*}
where the third step uses~(\ref{eq:tilde-phi-graded}), and the fourth
step uses~(\ref{eq:tilde-phi-ell1-norm}) and $c=\min\{c',c''\}.$

It remains to verify~(\ref{eq:mu-i-orth}). For this, fix $i,j\in[r]$
arbitrarily. Then $\|\tilde{\phi}_{i}\|_{1}-\|\tilde{\phi}_{j}\|_{1}=\langle\tilde{\phi}_{i},1\rangle-\langle\tilde{\phi}_{j},1\rangle=\langle\tilde{\phi}_{i}-\tilde{\phi}_{j},1\rangle=0,$
where the first step uses~(\ref{eq:tilde-phi-nonnegative}), and
the third step uses~(\ref{eq:tilde-phi-orthog}). We thus see that
\begin{equation}
\|\tilde{\phi}_{i}\|_{1}=\|\tilde{\phi}_{j}\|_{1}.\label{eq:tilde-phi-i-tilde-phi-j-ell1-norm}
\end{equation}
Next, observe that $\lambda_{i}$ can be written as the product of
two functions on disjoint sets of variables, and likewise for $\lambda_{j}.$
Namely, 
\begin{align*}
\lambda_{i} & =\frac{1}{\|\tilde{\phi}_{i}\|_{1}}\;\tilde{\phi}_{i}\otimes\NOR_{n-n'},\\
\lambda_{j} & =\frac{1}{\|\tilde{\phi}_{j}\|_{1}}\;\tilde{\phi}_{j}\otimes\NOR_{n-n'}.
\end{align*}
Now
\begin{align*}
\orth(\lambda_{i}-\lambda_{j}) & =\orth\left(\left(\frac{\tilde{\phi}_{i}}{\|\tilde{\phi}_{i}\|_{1}}-\frac{\tilde{\phi}_{j}}{\|\tilde{\phi}_{j}\|_{1}}\right)\otimes\NOR_{n-n'}\right)\\
 & \geq\orth\left(\frac{\tilde{\phi}_{i}}{\|\tilde{\phi}_{i}\|_{1}}-\frac{\tilde{\phi}_{j}}{\|\tilde{\phi}_{j}\|_{1}}\right)\\
 & =\orth\left(\frac{\tilde{\phi}_{i}-\tilde{\phi}_{j}}{\|\tilde{\phi}_{i}\|_{1}}\right)\\
 & =\orth(\tilde{\phi}_{i}-\tilde{\phi}_{j})\\
 & \geq c''\sqrt{\frac{n'}{m}},
\end{align*}
where the second step uses Proposition~\ref{prop:orth}\ref{item:orth-tensor},
the third step applies~(\ref{eq:tilde-phi-i-tilde-phi-j-ell1-norm}),
and the last step is justified by~(\ref{eq:tilde-phi-orthog}). In
view of $c=\min\{c',c''\},$ this settles~(\ref{eq:mu-i-orth}) and
completes the proof.
\end{proof}

\subsection{Hardness amplification for approximate degree}

We have reached the crux of our proof, a hardness amplification theorem
for approximate degree. Unlike previous work, our hardness amplification
is directly applicable to Boolean functions with sparse input and
does not use componentwise composition or input compression. The theorem
statement below has a large number of parameters, for maximum generality
and black-box integration with the auxiliary results of previous sections.
We will later derive a succinct and easy-to-apply corollary that will
suffice for our hardness amplification purposes.
\begin{thm}
\label{thm:hardness-amplification-twosided}Let $C^{*}\geq1$ and
$c\in(0,1)$ be the absolute constants from Theorems~\emph{\ref{thm:explicit-low-discrepancy-set}}
and~\emph{\ref{thm:encoding-1-r-normalized},} respectively. Fix
a real number $0<\beta<1$ and positive integers $n,m,k,N,\theta,D,T$
such that
\begin{align}
 & n/2\geq m>k,\label{eq:nmk}\\
 & 4(N+1)^{2}k\exp\left(-\frac{\sqrt{m}}{16k}\right)+(N+1)\left(\frac{3C^{*}\log^{2}(n+N+1)}{m^{1/4}}\right)^{k}\nonumber \\
 & \qquad\qquad\qquad\qquad\qquad\qquad\qquad\qquad\qquad\qquad\leq\frac{\beta}{16(N+1)^{2}m^{2}},\label{eq:delta-m-k-n-r-repeated-1-1}\\
 & T\geq\frac{8\e}{c}\cdot\theta(1+\ln\theta)+\theta k,\label{eq:T-theta-k}\\
 & T\geq D.\label{eq:T-greater-than-D}
\end{align}
Define
\begin{equation}
\Delta=\left(1+2^{D}\binom{n\theta}{D}\right)\exp\left(-\frac{c(T-\theta k)}{2\sqrt{nm}}\right).\label{eq:Delta-defined}
\end{equation}
Then there is an $($explicitly given$)$ mapping $H\colon(\zoon)^{\theta}\to\zoo^{N}$
such that:
\begin{enumerate}
\item \label{enu:each-output-bit-of-H}each output bit of $H$ is computable
by a monotone $(k+1)$-DNF formula; 
\item \label{enu:amplify-two-sided}for every $\epsilon\in[0,1]$ and every
$f\colon\zoo^{N}\to\zoo,$ one has
\[
\deg_{\epsilon-\beta\theta-2\Delta}((f\circ H)|_{\leq T})\geq\min\left\{ c\deg_{\epsilon}(f|_{\leq\theta})\sqrt{\frac{n}{2m}},D\right\} .
\]
\end{enumerate}
\end{thm}

\begin{proof}
We may assume that
\begin{equation}
\epsilon-\beta\theta-2\Delta\geq0\label{eq:error-nonnegative}
\end{equation}
since otherwise the left-hand side in the approximate degree lower
bound of~\ref{enu:amplify-two-sided} is by definition $+\infty$.
Define $V\subseteq\Re^{N}$ by $V=\{0^{N},e_{1},e_{2},\ldots,e_{N}\}$
and set $r=N+1.$ In view of~(\ref{eq:nmk}) and~(\ref{eq:delta-m-k-n-r-repeated-1-1}),
Corollary~\ref{cor:explicit-coloring} gives an explicit integer
$n'\in(n/2,n]$ and an explicit $(\frac{\beta}{16r^{2}m^{2}},\frac{\beta}{16r^{2}m^{2}},m)$-balanced
coloring $\gamma\colon\binom{[n']}{k}\to[r]$. Alternatively, if one
is not concerned about explicitness, the existence of $\gamma$ can
be deduced from the much simpler Corollary~\ref{cor:balanced-coloring-existential}.
Specifically, (\ref{eq:delta-m-k-n-r-repeated-1-1}) forces $\sqrt{m}\geq k$
and in particular $n\geq m\geq k^{2}\geq1.$ Moreover,~(\ref{eq:delta-m-k-n-r-repeated-1-1})
implies that $3r\sqrt{k\ln(n+1)}/m^{k/4}\leq\frac{\beta}{16r^{2}m^{2}}.$
Now Corollary~\ref{cor:balanced-coloring-existential} guarantees
the existence of a $(\frac{\beta}{16r^{2}m^{2}},\frac{\beta}{16r^{2}m^{2}},m)$-balanced
coloring $\gamma\colon\binom{[n]}{k}\to[r]$.

Since $n'\geq m>k,$ Theorem~\ref{thm:encoding-1-r-normalized} gives
explicit distributions $\lambda_{0^{N}},\lambda_{e_{1}},\lambda_{e_{2}},\ldots,\lambda_{e_{N}}$
on $\zoon$ such that
\begin{align}
 & \supp\lambda_{v}\subseteq\{x\in\zoon:|x|=k\text{ or }|x|\geq m\}, &  & v\in V,\label{eq:lambda-i-support-1}\\
 & \zoon|_{k}\cap\supp\lambda_{e_{i}}=\{\1_{S}:S\in\gamma^{-1}(i)\}, &  & i\in[N],\label{eq:lambda-i-kth-level-support-via-gamma}\\
 & \zoon|_{k}\cap\supp\lambda_{0^{N}}=\{\1_{S}:S\in\gamma^{-1}(N+1)\},\label{eq:lambda-i-kth-level-support-via-gamma-0}\\
 & \lambda_{v}(\zoon|_{k})\geq1-\beta, &  & v\in V,\label{eq:lambda-i-kth-level-1}\\
 & \lambda_{v}(\zoon|_{t})\leq\frac{\exp(-c(t-k)/\sqrt{nm})}{c(t-k+1)^{2}}, &  & v\in V,\;t\geq k,\label{eq:lambda-i-graded-bound-1}\\
 & \orth(\lambda_{v}-\lambda_{u})\geq c\sqrt{\frac{n}{2m},} &  & v,u\in V.\label{eq:lambda-i-orth-1}
\end{align}
Properties~(\ref{eq:lambda-i-kth-level-support-via-gamma}) and~(\ref{eq:lambda-i-kth-level-support-via-gamma-0})
imply that 
\begin{align}
\zoon|_{k}\cap\supp\lambda_{u}\cap\supp\lambda_{v} & =\varnothing, &  & u,v\in V,\;u\ne v.\label{eq:lambda-i-j-disjoint-1}
\end{align}
For $\mathbf{v}=(\mathbf{v}_{1},\mathbf{v}_{2},\ldots,\mathbf{v}_{\theta})\in V^{\theta}$,
define 
\begin{equation}
\Lambda_{\mathbf{v}}=\bigotimes_{i=1}^{\theta}\lambda_{\mathbf{v}_{i}}.\label{eq:Lambda-v-defined}
\end{equation}
\begin{claim}
\label{claim:tilde-Lambda}For each $\mathbf{v}\in V^{\theta},$ there
is a function $\widetilde{\Lambda_{\mathbf{v}}}\colon(\zoon)^{\theta}\to\Re$
such that
\begin{align}
 & \supp\widetilde{\Lambda_{\mathbf{v}}}\subseteq(\zoon)^{\theta}|_{\leq T},\label{eq:tilde-Lambda-supp}\\
 & \orth(\Lambda_{\mathbf{v}}-\widetilde{\Lambda_{\mathbf{v}}})>D,\label{eq:tilde-Lambda-orth}\\
 & \|\Lambda_{\mathbf{v}}-\widetilde{\Lambda_{\mathbf{v}}}\|_{1}\leq\Delta.\label{eq:tilde-Lambda-ell1}
\end{align}
\end{claim}

\noindent We will settle Claim~\ref{claim:tilde-Lambda}, and all
other claims, after the proof of the theorem.

We now turn to the construction of the monotone mapping $H$ in the
theorem statement. Define $h\colon\zoon\to\zoo^{N}$ by
\begin{align}
(h(z))_{j} & =\bigvee_{S\in\binom{[n]}{k+1}\cup\gamma^{-1}(j)}\;\bigwedge_{s\in S}z_{s}, &  & j=1,2,\ldots,N.\label{eq:definition-h}
\end{align}
Clearly, this is a monotone DNF formula of width $k+1$. Define $H\colon(\zoon)^{\theta}\to\zoo^{N}$
by
\begin{align}
H(x_{1},x_{2},\ldots,x_{\theta}) & =\bigvee_{i=1}^{\theta}h(x_{i}), &  & x_{1},x_{2},\ldots,x_{\theta}\in\zoon,\label{eq:definition-H}
\end{align}
where the right-hand side is the componentwise disjunction of the
Boolean vectors $h(x_{1}),h(x_{2}),\ldots,h(x_{\theta}).$ Observe
that both $h$ and $H$ are monotone and are given explicitly in closed
form in terms of the coloring $\gamma$ constructed at the beginning
of the proof. This settles~\ref{enu:each-output-bit-of-H}.

For~\ref{enu:amplify-two-sided}, fix an arbitrary function $f\colon\zoo^{N}\to\zoo$
and abbreviate
\[
d=\deg_{\epsilon}(f|_{\leq\theta}).
\]
By the dual characterization of approximate degree (Fact~\ref{fact:adeg-dual}),
there is a function $\psi\colon\zoo^{N}|_{\leq\theta}\to\Re$ such
that
\begin{align}
 & \|\psi\|_{1}=1,\label{eq:psi-bounded}\\
 & \langle f,\psi\rangle>\epsilon,\label{eq:psi-correlated}\\
 & \orth\psi\geq d.\label{eq:psi-orthogonal}
\end{align}
Define $\Psi\colon(\zoon)^{\theta}\to\Re$ by
\begin{equation}
\Psi=\sum_{u\in\zoo^{N}|_{\leq\theta}}\psi(u)\Exp_{\substack{\mathbf{v}\in V^{\theta}:\\
\mathbf{v}_{1}+\mathbf{v}_{2}+\cdots+\mathbf{v}_{\theta}=u
}
}\widetilde{\Lambda_{\mathbf{v}}}.\label{eq:Psi-defined}
\end{equation}
We will now use (\ref{eq:nmk})\textendash (\ref{eq:Psi-defined})
to prove a sequence of claims.
\begin{claim}
\label{claim:Psi-properties}One has
\begin{align}
 & \supp\Psi\subseteq(\zoon)^{\theta}|_{\leq T},\label{eq:Psi-support-on-light-inputs}\\
 & \|\Psi\|_{1}\leq1+\Delta,\label{eq:PSI-ell1}\\
 & \orth\Psi\geq\min\left\{ cd\sqrt{\frac{n}{2m}},D\right\} .\label{eq:orth-Psi}
\end{align}
\end{claim}

\begin{claim}
\label{claim:h-on-short-inputs}Let $v\in V$ be given. Then for all
$z\in\zoon|_{k}\cap\supp\lambda_{v},$ one has $h(z)=v.$
\end{claim}

\begin{claim}
\label{claim:H-on-Lambda-tilde}Let $u\in\zoo^{N}|_{\leq\theta}$
and $\mathbf{v}=(\mathbf{v}_{1},\mathbf{v}_{2},\ldots,\mathbf{v}_{\theta})\in V^{\theta}$
be given such that $\mathbf{v}_{1}+\mathbf{v}_{2}+\cdots+\mathbf{v}_{\theta}=u.$
Then
\begin{equation}
|f(u)-\langle\widetilde{\Lambda_{\mathbf{v}}},f\circ H\rangle|\leq\beta\theta+\Delta.
\end{equation}
\end{claim}

\begin{claim}
\label{claim:Psi-correl}One has
\begin{equation}
\langle f\circ H,\Psi\rangle>(\epsilon-\beta\theta-2\Delta)\|\Psi\|_{1}.\label{eq:correl-Psi-f-H}
\end{equation}
\end{claim}

Note from~(\ref{eq:Psi-support-on-light-inputs}) that $\Psi$ is
supported on inputs of Hamming weight at most $T$ and can therefore
be regarded as a function on $(\zoon)^{\theta}|_{\leq T}$. Now the
claimed bound in~\ref{enu:amplify-two-sided} follows by Fact~\ref{fact:adeg-dual}
in view of~(\ref{eq:orth-Psi}) and (\ref{eq:correl-Psi-f-H}). The
proof of the theorem is complete.
\end{proof}
\begin{proof}[Proof of Claim~\emph{\ref{claim:tilde-Lambda}}.]
 Equations (\ref{eq:T-theta-k}), (\ref{eq:lambda-i-support-1}),
and~(\ref{eq:lambda-i-graded-bound-1}) ensure that Lemma~\ref{cor:concentration-of-measure-Boolean-shifted}
is applicable to the distributions $\lambda_{\mathbf{v}_{1}},\lambda_{\mathbf{v}_{2}},\ldots,\lambda_{\mathbf{v}_{\theta}}$
with parameters $\ell=\theta,$ $B=n,$ $C=1/c$, and $\alpha=\exp(-c/\sqrt{nm})$,
whence
\begin{align*}
\Lambda_{\mathbf{v}}((\zoon)^{\theta}|_{>T}) & \leq\exp\left(-\frac{c(T-\theta k)}{2\sqrt{nm}}\right), &  & \mathbf{v}\in V^{\theta}.
\end{align*}
In view of~(\ref{eq:T-greater-than-D}), we can now invoke Lemma~\ref{lem:zero-high-Hamming-weight}
with parameter $B=n\theta$ to obtain a function $\widetilde{\Lambda_{\mathbf{v}}}\colon(\zoon)^{\theta}\to\Re$
that satisfies~(\ref{eq:tilde-Lambda-supp})\textendash (\ref{eq:tilde-Lambda-ell1}).
\end{proof}
\begin{proof}[Proof of Claim~\emph{\ref{claim:Psi-properties}}.]
 Observe from~(\ref{eq:tilde-Lambda-supp}) that $\Psi$ is a linear
combination of functions supported on inputs of Hamming weight at
most $T.$ This settles the support property~(\ref{eq:Psi-support-on-light-inputs}).
Property~(\ref{eq:PSI-ell1}) can be verified as follows:
\begin{align*}
\|\Psi\|_{1} & \leq\sum_{u\in\zoo^{N}|_{\leq\theta}}|\psi(u)|\Exp_{\substack{\mathbf{v}\in V^{\theta}:\\
\mathbf{v}_{1}+\mathbf{v}_{2}+\cdots+\mathbf{v}_{\theta}=u
}
}\|\widetilde{\Lambda{}_{\mathbf{v}}}\|_{1}\\
 & \leq\left(\sum_{u\in\zoo^{N}|_{\leq\theta}}|\psi(u)|\right)\max_{\mathbf{v}\in V^{\theta}}\|\widetilde{\Lambda_{\mathbf{v}}}\|_{1}\\
 & =\|\psi\|_{1}\max_{\mathbf{v}\in V^{\theta}}\|\widetilde{\Lambda_{\mathbf{v}}}\|_{1}\\
 & \leq\|\psi\|_{1}\max_{\mathbf{v}\in V^{\theta}}\{\|\Lambda_{\mathbf{v}}\|_{1}+\|\widetilde{\Lambda_{\mathbf{v}}}-\Lambda_{\mathbf{v}}\|_{1}\}\\
 & \leq1+\Delta,
\end{align*}
where the first and fourth steps apply the triangle inequality, and
the last step uses~(\ref{eq:tilde-Lambda-ell1}) and~(\ref{eq:psi-bounded}).

To settle~(\ref{eq:orth-Psi}), consider an arbitrary polynomial
$P\colon(\zoon)^{\theta}\to\Re$ of degree less than $\min\{cd\sqrt{n/(2m)},D\}.$
Then
\begin{align}
\langle\Psi,P\rangle & =\sum_{u\in\zoo^{N}|_{\leq\theta}}\psi(u)\Exp_{\substack{\mathbf{v}\in V^{\theta}:\\
\mathbf{v}_{1}+\mathbf{v}_{2}+\cdots+\mathbf{v}_{\theta}=u
}
}\langle\widetilde{\Lambda_{\mathbf{v}}},P\rangle\nonumber \\
 & =\sum_{u\in\zoo^{N}|_{\leq\theta}}\psi(u)\Exp_{\substack{\mathbf{v}\in V^{\theta}:\\
\mathbf{v}_{1}+\mathbf{v}_{2}+\cdots+\mathbf{v}_{\theta}=u
}
}[\langle\Lambda_{\mathbf{v}},P\rangle+\langle\widetilde{\Lambda_{\mathbf{v}}}-\Lambda_{\mathbf{v}},P\rangle]\nonumber \\
 & =\sum_{u\in\zoo^{N}|_{\leq\theta}}\psi(u)\Exp_{\substack{\mathbf{v}\in V^{\theta}:\\
\mathbf{v}_{1}+\mathbf{v}_{2}+\cdots+\mathbf{v}_{\theta}=u
}
}\langle\Lambda_{\mathbf{v}},P\rangle,\label{eq:orth-Psi-P-almost-finished}
\end{align}
where the first and second steps use the linearity of inner product,
and the third step is valid by~(\ref{eq:tilde-Lambda-orth}). Equation
(\ref{eq:lambda-i-orth-1}) allows us to invoke Proposition~\ref{prop:expect-out}
with $\ell=\theta$ and $\phi_{v}=\lambda_{v}$ to infer that the
inner product $\langle\Lambda_{\mathbf{v}},P\rangle$ is a polynomial
in $\mathbf{v}$ of degree less than $d.$ As a result, Fact~\ref{fact:ambainis-symmetrization}
implies that the expected value in~(\ref{eq:orth-Psi-P-almost-finished})
is a polynomial in $u$ of degree less than $d.$ In summary, (\ref{eq:orth-Psi-P-almost-finished})
is the inner product of $\psi$ with a polynomial of degree less than
$d$ and is therefore zero by~(\ref{eq:psi-orthogonal}). The proof
of~(\ref{eq:orth-Psi}) is complete.
\end{proof}
\begin{proof}[Proof of Claim~\emph{\ref{claim:h-on-short-inputs}}.]
 Consider an arbitrary string $z\in\zoon|_{k}$. Then
\[
(h(z))_{j}=\bigvee_{S\in\gamma^{-1}(j)}\;\bigwedge_{s\in S}z_{s}=\I[z\in\supp\lambda_{e_{j}}],
\]
where the first step uses the defining equation~(\ref{eq:definition-h})
together with $|z|=k,$ and the second step applies~(\ref{eq:lambda-i-kth-level-support-via-gamma})
along with~$|z|=k.$ Thus, $h(z)$ can be written out explicitly
as
\begin{equation}
h(z)=(\I[z\in\supp\lambda_{e_{1}}],\I[z\in\supp\lambda_{e_{2}}],\ldots,\I[z\in\supp\lambda_{e_{N}}]).\label{eq:h-written-out}
\end{equation}
Now recall from (\ref{eq:lambda-i-j-disjoint-1}) that a string $z$
of Hamming weight $k$ can belong to at most one of the sets $\supp\lambda_{0^{N}},\supp\lambda_{e_{1}},\supp\lambda_{e_{2}},\ldots,\supp\lambda_{e_{N}}.$
As a result, if $z\in\supp\lambda_{e_{i}}$ then $z\notin\supp\lambda_{e_{j}}$
for all $j\ne i$ and consequently $h(z)=e_{i}$ by~(\ref{eq:h-written-out}).
Analogously, if $z\in\supp\lambda_{0^{N}}$ then $z\notin\supp\lambda_{e_{j}}$
for all $j$ and consequently $h(z)=0^{N}$ by~(\ref{eq:h-written-out}).
This settles the claim for all $v\in V.$
\end{proof}
\begin{proof}[Proof of Claim~\emph{\ref{claim:H-on-Lambda-tilde}}.]
 Since $u$ is a Boolean vector, the equality $\mathbf{v}_{1}+\mathbf{v}_{2}+\cdots+\mathbf{v}_{\theta}=u$
forces
\begin{equation}
\mathbf{v}_{1}\vee\mathbf{v}_{2}\vee\cdots\vee\mathbf{v}_{\theta}=u,\label{eq:vvvv-u-Boolean}
\end{equation}
where the disjunction is applied componentwise. For any input $(x_{1},x_{2},\ldots,x_{\theta})$
where $x_{i}\in\zoon|_{k}\cap\supp\lambda_{\mathbf{v}_{i}},$ we have
\[
(f\circ H)(x_{1},x_{2},\ldots,x_{\theta})=f\left(\bigvee_{i=1}^{\theta}h(x_{i})\right)=f\left(\bigvee_{i=1}^{\theta}\mathbf{v}_{i}\right)=f(u),
\]
where the second and third steps use Claim~\ref{claim:h-on-short-inputs}
and (\ref{eq:vvvv-u-Boolean}), respectively. Since $\supp\Lambda_{\mathbf{v}}=\prod_{i=1}^{\theta}\supp\lambda_{\mathbf{v}_{i}},$
we have shown that
\begin{equation}
f\circ H\equiv f(u)\qquad\qquad\text{on }\;\;(\zoon|_{k})^{\theta}\cap\supp\Lambda_{\mathbf{v}}.\label{eq:f-H-simplified}
\end{equation}
Furthermore,
\begin{equation}
\Lambda_{\mathbf{v}}((\zoon|_{k})^{\theta})=\prod_{i=1}^{\theta}\lambda_{\mathbf{v}_{i}}(\zoon|_{k})\geq(1-\beta)^{\theta}\geq1-\beta\theta,\label{eq:Lambda-v-heavy-inputs-rare}
\end{equation}
where the second step uses~(\ref{eq:lambda-i-kth-level-1}). Now
\begin{align*}
|f(u)- & \langle\widetilde{\Lambda_{\mathbf{v}}},f\circ H\rangle|\\
 & \leq|f(u)-\langle\Lambda_{\mathbf{v}},f\circ H\rangle|+|\langle\Lambda_{\mathbf{v}}-\widetilde{\Lambda_{\mathbf{v}}},f\circ H\rangle|\\
 & \leq|f(u)-\langle\Lambda_{\mathbf{v}},f\circ H\rangle|+\|\Lambda_{\mathbf{v}}-\widetilde{\Lambda_{\mathbf{v}}}\|_{1}\\
 & =\left|f(u)-\Exp_{\Lambda_{\mathbf{v}}}f\circ H\right|+\|\Lambda_{\mathbf{v}}-\widetilde{\Lambda_{\mathbf{v}}}\|_{1}\\
 & \leq\Exp_{\Lambda_{\mathbf{v}}}|f(u)-f\circ H|+\|\Lambda_{\mathbf{v}}-\widetilde{\Lambda_{\mathbf{v}}}\|_{1}\\
 & \leq0\cdot\Lambda_{\mathbf{v}}((\zoon|_{k})^{\theta})+1\cdot\Lambda_{\mathbf{v}}(\overline{(\zoon|_{k})^{\theta}})+\|\Lambda_{\mathbf{v}}-\widetilde{\Lambda_{\mathbf{v}}}\|_{1}\\
 & \leq\beta\theta+\|\Lambda_{\mathbf{v}}-\widetilde{\Lambda_{\mathbf{v}}}\|_{1}\\
 & \leq\beta\theta+\Delta,
\end{align*}
where the last three steps use~(\ref{eq:f-H-simplified}), (\ref{eq:Lambda-v-heavy-inputs-rare}),
and~(\ref{eq:tilde-Lambda-ell1}), respectively.
\end{proof}
\begin{proof}[Proof of Claim~\emph{\ref{claim:Psi-correl}}.]
To begin with,
\begin{align}
\langle f,\psi\rangle-\langle f\circ H,\Psi\rangle & =\sum_{u\in\zoo^{N}|_{\leq\theta}}\psi(u)f(u)\nonumber \\
 & \qquad\qquad-\sum_{u\in\zoo^{N}|_{\leq\theta}}\psi(u)\Exp_{\substack{\mathbf{v}\in V^{\theta}:\\
\mathbf{v}_{1}+\mathbf{v}_{2}+\cdots+\mathbf{v}_{\theta}=u
}
}\langle\widetilde{\Lambda_{\mathbf{v}}},f\circ H\rangle\nonumber \\
 & =\sum_{u\in\zoo^{N}|_{\leq\theta}}\psi(u)\Exp_{\substack{\mathbf{v}\in V^{\theta}:\\
\mathbf{v}_{1}+\mathbf{v}_{2}+\cdots+\mathbf{v}_{\theta}=u
}
}[f(u)-\langle\widetilde{\Lambda_{\mathbf{v}}},f\circ H\rangle]\nonumber \\
 & \leq\sum_{u\in\zoo^{N}|_{\leq\theta}}|\psi(u)|\Exp_{\substack{\mathbf{v}\in V^{\theta}:\\
\mathbf{v}_{1}+\mathbf{v}_{2}+\cdots+\mathbf{v}_{\theta}=u
}
}|f(u)-\langle\widetilde{\Lambda_{\mathbf{v}}},f\circ H\rangle|\nonumber \\
 & \leq\|\psi\|_{1}\max_{u\in\zoo^{N}|_{\leq\theta}}\max_{\substack{\mathbf{v}\in V^{\theta}:\\
\mathbf{v}_{1}+\mathbf{v}_{2}+\cdots+\mathbf{v}_{\theta}=u
}
}|f(u)-\langle\widetilde{\Lambda_{\mathbf{v}}},f\circ H\rangle|\nonumber \\
 & \leq\|\psi\|_{1}(\beta\theta+\Delta)\nonumber \\
 & =\beta\theta+\Delta,\label{eq:error-accumulation}
\end{align}
where the last two steps use Claim~\ref{claim:H-on-Lambda-tilde}
and~(\ref{eq:psi-bounded}), respectively. Then
\begin{align*}
\langle f\circ H,\Psi\rangle & >\epsilon-\beta\theta-\Delta\\
 & \geq\frac{\epsilon-\beta\theta-\Delta}{1+\Delta}\cdot\|\Psi\|_{1}\\
 & \geq(\epsilon-\beta\theta-2\Delta)\|\Psi\|_{1},
\end{align*}
where the first step uses~(\ref{eq:psi-correlated}) and~(\ref{eq:error-accumulation}),
the second step is justified by~(\ref{eq:error-nonnegative}) and~(\ref{eq:PSI-ell1}),
and the third step is legitimate since $a/(1+b)\geq a-b$ for all
$a\in[0,1]$ and $b\geq0$. This completes the proof of~(\ref{eq:correl-Psi-f-H}).
\end{proof}

\subsection{Hardness amplification for one-sided approximate degree}

In this section, we will prove that the construction of Theorem~\ref{thm:hardness-amplification-twosided}
amplifies not only approximate degree but also its one-sided variant.
We start with a technical lemma.
\begin{lem}
\label{lem:AytildeAy} Let $n,m,k,\theta,D,T$ be positive integers
with
\begin{align}
T & \geq n+D,\label{eq:Tnd}\\
T & \geq\theta k.\label{eq:Tktheta}
\end{align}
Let $y\in(\zoon)^{\theta}|_{>T}$ be given. Then there exists $\zeta_{y}\colon(\zoon)^{\theta}\to\Re$
such that
\begin{align}
 & \supp\zeta_{y}\subseteq(\zoon)^{\theta}|_{\leq T}\cup\{y\},\label{eq:zeta-support}\\
 & \zeta_{y}(y)=1,\label{eq:zeta-y-y}\\
 & \orth\zeta_{y}>D,\label{eq:zeta-orth}\\
 & \|\zeta_{y}\|_{1}\leq1+2^{D}\binom{n(\theta-1)}{D},\label{eq:zeta-ell1}\\
 & \zeta_{y}=0\quad\text{ on }\quad(\zoon|_{\leq k})^{\theta}.\label{eq:zeta-on-k}
\end{align}
\end{lem}

\begin{proof}
It follows from~(\ref{eq:Tktheta}) that $y=(y_{1},y_{2},\ldots,y_{\theta})$
has a coordinate with Hamming weight greater than $k.$ By symmetry,
we may assume that 
\begin{equation}
|y_{1}|>k.\label{eq:y1-greater-m}
\end{equation}
We have $|y_{2}y_{3}\ldots y_{\theta}|=|y|-|y_{1}|>T-n\geq D,$ where
the second step uses the hypothesis $|y|>T$ along with the trivial
bound $|y_{1}|\leq n,$ whereas the third step is legitimate by~(\ref{eq:Tnd}).
Thanks to the newly obtained inequality $|y_{2}y_{3}\ldots y_{\theta}|>D,$
Lemma~\ref{lem:razborov-sherstov-generalized} is applicable with
$B=n(\theta-1)$ and gives a function $\zeta\colon(\zoon)^{\theta-1}\to\Re$
such that
\begin{align}
 & \supp\zeta\subseteq(\zoon)^{\theta-1}|_{\leq D}\cup\{y_{2}y_{3}\ldots y_{\theta}\},\label{eq:zeta-y-supp-1}\\
 & \zeta(y_{2}y_{3}\ldots y_{\theta})=1,\label{eq:zeta-y-point-mass-1}\\
 & \|\zeta\|_{1}\leq1+2^{D}\binom{n(\theta-1)}{D},\label{eq:zeta-y-ell1-1}\\
 & \orth\zeta>D.\label{eq:zeta-y-orth-1}
\end{align}
We will prove that the claimed properties~(\ref{eq:zeta-support})\textendash (\ref{eq:zeta-on-k})
are enjoyed by the function
\[
\zeta_{y}(x)=\delta_{x_{1},y_{1}}\zeta(x_{2}x_{3}\ldots x_{\theta}).
\]
To verify the support property~(\ref{eq:zeta-support}), fix any
$x$ with $\zeta_{y}(x)\ne0.$ Then necessarily $\delta_{x_{1},y_{1}}=1,$
forcing $x_{1}=y_{1}.$ Now~(\ref{eq:zeta-y-supp-1}) implies that
$x$ either equals $y$ or has Hamming weight at most $|y_{1}|+D.$
Since $|y_{1}|+D\leq n+D\leq T$ by~(\ref{eq:Tnd}), this completes
the proof of~(\ref{eq:zeta-support}).

The remaining properties are straightforward. Property~(\ref{eq:zeta-y-y})
follows from the corresponding property~(\ref{eq:zeta-y-point-mass-1})
of $\zeta$. Likewise, property~(\ref{eq:zeta-orth}) follows from~(\ref{eq:zeta-y-orth-1})
in light of Proposition~\ref{prop:orth}~\ref{item:orth-tensor}.
Property~(\ref{eq:zeta-ell1}) is immediate from~(\ref{eq:zeta-y-ell1-1}).
Finally, (\ref{eq:zeta-on-k}) is a consequence of~(\ref{eq:y1-greater-m}).
\end{proof}
We are now ready to state and prove our hardness amplification result,
which is a far-reaching generalization of Theorem~\ref{thm:hardness-amplification-twosided}.
\begin{thm}
\label{thm:hardness-amplification-onesided}Let $C^{*}\geq1$ and
$c\in(0,1)$ be the absolute constants from Theorems~\emph{\ref{thm:explicit-low-discrepancy-set}}
and~\emph{\ref{thm:encoding-1-r-normalized},} respectively. Fix
a real number $0<\beta<1$ and positive integers $n,m,k,N,\theta,D,T$
such that
\begin{align}
 & n/2\geq m>k,\label{eq:nmk-1}\\
 & 4(N+1)^{2}k\exp\left(-\frac{\sqrt{m}}{16k}\right)+(N+1)\left(\frac{3C^{*}\log^{2}(n+N+1)}{m^{1/4}}\right)^{k}\nonumber \\
 & \qquad\qquad\qquad\qquad\qquad\qquad\qquad\qquad\qquad\qquad\leq\frac{\beta}{16(N+1)^{2}m^{2}},\label{eq:delta-m-k-n-r-repeated-1-1-1}\\
 & T\geq\frac{8\e}{c}\cdot\theta(1+\ln\theta)+\theta k,\label{eq:T-theta-k-1}\\
 & T\geq D+n.\label{eq:T-greater-than-D-plus-n}
\end{align}
Define
\begin{equation}
\Delta=\left(1+2^{D}\binom{n\theta}{D}\right)\exp\left(-\frac{c(T-\theta k)}{2\sqrt{nm}}\right).\label{eq:Delta-defined-1}
\end{equation}
Then there is an $($explicitly given$)$ mapping $H\colon(\zoon)^{\theta}\to\zoo^{N}$
such that:
\begin{enumerate}
\item \label{enu:each-output-bit-of-H-onesided}each output bit of $H$
is computable by a monotone $(k+1)$-DNF formula; 
\item \label{enu:amplify-generalized-two-sided}for every $\epsilon\in[0,1]$
and every $f\colon\zoo^{N}\to\zoo,$ one has
\[
\deg_{\epsilon-\beta\theta-2\Delta}((f\circ H)|_{\leq T})\geq\min\left\{ c\deg_{\epsilon}(f|_{\leq\theta})\sqrt{\frac{n}{2m}},D\right\} ;
\]
\item \label{enu:amplify-generalized-one-sided}for every $\epsilon\in[0,1]$
and every $f\colon\zoo^{N}\to\zoo$ with $f(1^{N})=0,$ one has
\[
\onedeg_{\epsilon-\beta\theta-2\Delta}((f\circ H)|_{\leq T})\geq\min\left\{ c\onedeg_{\epsilon}(f|_{\leq\theta})\sqrt{\frac{n}{2m}},D\right\} .
\]
\end{enumerate}
\end{thm}

\begin{proof}
As in the proof of Theorem~\ref{thm:hardness-amplification-twosided},
we may assume that
\begin{equation}
\epsilon-\beta\theta-2\Delta\geq0\label{eq:error-nonnegative-1}
\end{equation}
since otherwise the left-hand side in the lower bounds of~\ref{enu:amplify-generalized-two-sided}
and \ref{enu:amplify-generalized-one-sided} is by definition $+\infty$.
Define $V\subseteq\Re^{N}$ by $V=\{0^{N},e_{1},e_{2},\ldots,e_{N}\}$
and set $r=N+1.$ Arguing as in the proof of Theorem~\ref{thm:hardness-amplification-twosided},
we obtain an explicit integer $n'\in(n/2,n]$ and an explicit $(\frac{\beta}{16r^{2}m^{2}},\frac{\beta}{16r^{2}m^{2}},m)$-balanced
coloring $\gamma\colon\binom{[n']}{k}\to[r]$, which in turn results
in explicit distributions $\lambda_{0^{N}},\lambda_{e_{1}},\lambda_{e_{2}},\ldots,\lambda_{e_{N}}$
on $\zoon$ such that
\begin{align}
 & \supp\lambda_{v}\subseteq\{x\in\zoon:|x|=k\text{ or }|x|\geq m\}, &  & v\in V,\label{eq:lambda-i-support-1-1}\\
 & \zoon|_{k}\cap\supp\lambda_{e_{i}}=\{\1_{S}:S\in\gamma^{-1}(i)\}, &  & i\in[N],\label{eq:lambda-i-kth-level-support-via-gamma-1}\\
 & \zoon|_{k}\cap\supp\lambda_{0^{N}}=\{\1_{S}:S\in\gamma^{-1}(N+1)\},\label{eq:lambda-i-kth-level-support-via-gamma-0-1}\\
 & \lambda_{v}(\zoon|_{k})\geq1-\beta, &  & v\in V,\label{eq:lambda-i-kth-level-1-1}\\
 & \lambda_{v}(\zoon|_{t})\leq\frac{\exp(-c(t-k)/\sqrt{nm})}{c(t-k+1)^{2}}, &  & v\in V,\;t\geq k,\label{eq:lambda-i-graded-bound-1-1}\\
 & \orth(\lambda_{v}-\lambda_{u})\geq c\sqrt{\frac{n}{2m},} &  & v,u\in V.\label{eq:lambda-i-orth-1-1}
\end{align}
Properties~(\ref{eq:lambda-i-kth-level-support-via-gamma-1}) and~(\ref{eq:lambda-i-kth-level-support-via-gamma-0-1})
imply that 
\begin{align}
\zoon|_{k}\cap\supp\lambda_{u}\cap\supp\lambda_{v} & =\varnothing, &  & u,v\in V,\;u\ne v.\label{eq:lambda-i-j-disjoint-1-1}
\end{align}
For $\mathbf{v}=(\mathbf{v}_{1},\mathbf{v}_{2},\ldots,\mathbf{v}_{\theta})\in V^{\theta}$,
define 
\begin{equation}
\Lambda_{\mathbf{v}}=\bigotimes_{i=1}^{\theta}\lambda_{\mathbf{v}_{i}}.
\end{equation}
 
\begin{claim}
\label{claim:tilde-Lambda-onesided}For each $\mathbf{v}\in V^{\theta},$
there is a function $\widetilde{\Lambda_{\mathbf{v}}}\colon(\zoon)^{\theta}\to\Re$
such that
\begin{align}
 & \supp\widetilde{\Lambda_{\mathbf{v}}}\subseteq(\zoon)^{\theta}|_{\leq T},\label{eq:tilde-Lambda-supp-1}\\
 & \orth(\Lambda_{\mathbf{v}}-\widetilde{\Lambda_{\mathbf{v}}})>D,\label{eq:tilde-Lambda-orth-1}\\
 & \|\Lambda_{\mathbf{v}}-\widetilde{\Lambda_{\mathbf{v}}}\|_{1}\leq\Delta,\label{eq:tilde-Lambda-ell1-1}\\
 & \widetilde{\Lambda_{\mathbf{v}}}=\Lambda_{\mathbf{v}}\qquad\text{on }(\zoon|_{\leq k})^{\theta}.\label{eq:tilde-Lambda-equals-Lambda}
\end{align}
\end{claim}

We will settle Claim~\ref{claim:tilde-Lambda-onesided} after the
proof of the theorem. We now define the monotone mapping $H$ exactly
the same way as in the proof of Theorem~\ref{thm:hardness-amplification-twosided}.
Specifically, define $h\colon\zoon\to\zoo^{N}$ by
\begin{align}
(h(z))_{j} & =\bigvee_{S\in\binom{[n]}{k+1}\cup\gamma^{-1}(j)}\;\bigwedge_{s\in S}z_{s}, &  & j=1,2,\ldots,N.\label{eq:definition-h-1}
\end{align}
Define $H\colon(\zoon)^{\theta}\to\zoo^{N}$ by
\begin{align}
H(x_{1},x_{2},\ldots,x_{\theta}) & =\bigvee_{i=1}^{\theta}h(x_{i}), &  & x_{1},x_{2},\ldots,x_{\theta}\in\zoon,\label{eq:definition-H-1}
\end{align}
where the right-hand side is the componentwise disjunction of the
Boolean vectors $h(x_{1}),h(x_{2}),\ldots,h(x_{\theta}).$ With these
definitions, items~\ref{enu:each-output-bit-of-H-onesided} and~\ref{enu:amplify-generalized-two-sided}
are immediate because they are restatements of Theorem~\ref{thm:hardness-amplification-twosided}~\ref{enu:each-output-bit-of-H},~\ref{enu:amplify-two-sided}.
To prove the remaining item~\ref{enu:amplify-generalized-one-sided},
fix an arbitrary function $f\colon\zoo^{N}\to\zoo$ with
\begin{equation}
f(1^{N})=0,\label{eq:f-at-11111}
\end{equation}
and abbreviate
\[
d=\onedeg_{\epsilon}(f|_{\leq\theta}).
\]
By the dual characterization of one-sided approximate degree (Fact~\ref{fact:onedeg-dual}),
there is a function $\psi\colon\zoo^{N}|_{\leq\theta}\to\Re$ such
that
\begin{align}
 & \|\psi\|_{1}=1,\label{eq:psi-bounded-onnesided}\\
 & \langle f,\psi\rangle>\epsilon,\label{eq:psi-correlated-onesided}\\
 & \orth\psi\geq d,\label{eq:psi-orthogonal-onesided}\\
 & \psi(x)\geq0\quad\text{whenever}\quad f(x)=1.\label{eq:psi-nonneg-on-f1}
\end{align}
Define $\Psi\colon(\zoon)^{\theta}\to\Re$ by
\begin{equation}
\Psi=\sum_{u\in\zoo^{N}|_{\leq\theta}}\psi(u)\Exp_{\substack{\mathbf{v}\in V^{\theta}:\\
\mathbf{v}_{1}+\mathbf{v}_{2}+\cdots+\mathbf{v}_{\theta}=u
}
}\widetilde{\Lambda_{\mathbf{v}}}.\label{eq:Psi-defined-1}
\end{equation}
Equations~(\ref{eq:nmk-1})\textendash (\ref{eq:Psi-defined-1})
subsume the corresponding equations (\ref{eq:nmk})\textendash (\ref{eq:Psi-defined})
in the proof of Theorem~\ref{thm:hardness-amplification-twosided}.
Recall that from (\ref{eq:nmk})\textendash (\ref{eq:Psi-defined}),
we deduced Claims~\ref{claim:Psi-properties}\textendash \ref{claim:Psi-correl}.
As a result, Claims~\ref{claim:Psi-properties}\textendash \ref{claim:Psi-correl}
remain valid here as well. In particular, we have
\begin{align}
 & \supp\Psi\subseteq(\zoon)^{\theta}|_{\leq T},\label{eq:Psi-support-on-light-inputs-1}\\
 & \orth\Psi\geq\min\left\{ cd\sqrt{\frac{n}{2m}},D\right\} ,\label{eq:orth-Psi-1}\\
 & h\equiv v\quad\text{on }\zoon|_{k}\cap\supp\lambda_{v} &  & (v\in V),\label{eq:h-on-short-inputs}\\
 & \langle f\circ H,\Psi\rangle>(\epsilon-\beta\theta-2\Delta)\|\Psi\|_{1}.\label{eq:correl-Psi-f-H-1}
\end{align}
Moreover, we will shortly prove the following new claim.
\begin{claim}
\label{claim:Psi-nonnegative} $\Psi(x)\geq0$ whenever $(f\circ H)(x)=1.$
\end{claim}

\noindent The lower bound on the one-sided approximate degree in~\ref{enu:amplify-generalized-one-sided}
now follows from the dual characterization of one-sided approximate
degree (Fact~\ref{fact:onedeg-dual}) in view of (\ref{eq:Psi-support-on-light-inputs-1})\textendash (\ref{eq:correl-Psi-f-H-1})
and Claim~\ref{claim:Psi-nonnegative}. This completes the proof
of Theorem~\ref{thm:hardness-amplification-onesided}. 
\end{proof}
\begin{proof}[Proof of Claim~\emph{\ref{claim:tilde-Lambda-onesided}.}]
 Fix $\mathbf{v}\in V^{\theta}$ arbitrarily for the remainder of
the proof. Equations (\ref{eq:T-theta-k-1}), (\ref{eq:lambda-i-support-1-1}),
and~(\ref{eq:lambda-i-graded-bound-1-1}) ensure that Lemma~\ref{cor:concentration-of-measure-Boolean-shifted}
is applicable to the distributions $\lambda_{\mathbf{v}_{1}},\lambda_{\mathbf{v}_{2}},\ldots,\lambda_{\mathbf{v}_{\theta}}$
with parameters $\ell=\theta,$ $B=n,$ $C=1/c$, and $\alpha=\exp(-c/\sqrt{nm})$,
whence
\begin{align}
\Lambda_{\mathbf{v}}((\zoon)^{\theta}|_{>T}) & \leq\exp\left(-\frac{c(T-\theta k)}{2\sqrt{nm}}\right).\label{eq:Lambda-v-concentration}
\end{align}
Recall from (\ref{eq:T-theta-k-1}) and (\ref{eq:T-greater-than-D-plus-n})
that $T\geq D+n$ and $T\geq\theta k,$ which makes Lemma~\ref{lem:AytildeAy}
applicable. Define $\widetilde{\Lambda_{\mathbf{v}}}\colon(\zoon)^{\theta}\to\Re$
by
\begin{equation}
\widetilde{\Lambda_{\mathbf{v}}}=\Lambda_{\mathbf{v}}-\sum_{y\in(\zoon)^{\theta}|_{>T}}\Lambda_{\mathbf{v}}(y)\zeta_{y},\label{eq:tilde-Lambda-definition}
\end{equation}
where $\zeta_{y}$ is as given by Lemma~\ref{lem:AytildeAy}. To
verify the support property~(\ref{eq:tilde-Lambda-supp-1}), fix
any input $x$ of Hamming weight $|x|>T.$ For all $y$ in the summation
with $y\ne x$, we have $\zeta_{y}(x)=0$ in view of~(\ref{eq:zeta-support}).
As a result, (\ref{eq:tilde-Lambda-definition}) simplifies to $\widetilde{\Lambda_{\mathbf{v}}}(x)=\Lambda_{\mathbf{v}}(x)-\Lambda_{\mathbf{v}}(x)\zeta_{x}(x).$
In view of~(\ref{eq:zeta-y-y}), we conclude that $\widetilde{\Lambda_{\mathbf{v}}}(x)=0$.

The orthogonality property~(\ref{eq:tilde-Lambda-orth-1}) follows
from
\[
\orth(\Lambda_{\mathbf{v}}-\widetilde{\Lambda_{\mathbf{v}}})\geq\min_{y\in(\zoon)^{\theta}|_{>T}}\orth\zeta_{y}>D,
\]
where the first step uses the defining equation~(\ref{eq:tilde-Lambda-definition})
and Proposition~\ref{prop:orth}~\ref{item:orth-sum}, and the second
step is legitimate by~(\ref{eq:zeta-orth}). 

Property~(\ref{eq:tilde-Lambda-ell1-1}) can be verified as follows:

\begin{align*}
\|\Lambda_{\mathbf{v}}-\widetilde{\Lambda_{\mathbf{v}}}\|_{1} & \leq\sum_{y\in(\zoon)^{\theta}|_{>T}}\Lambda_{\mathbf{v}}(y)\|\zeta_{y}\|_{1}\\
 & \leq\left(1+2^{D}\binom{n(\theta-1)}{D}\right)\sum_{y\in(\zoon)^{\theta}|_{>T}}\Lambda_{\mathbf{v}}(y)\\
 & \leq\left(1+2^{D}\binom{n(\theta-1)}{D}\right)\exp\left(-\frac{c(T-\theta k)}{2\sqrt{nm}}\right)\\
 & \leq\Delta,
\end{align*}
where the first step uses the triangle inequality along with the defining
equation~(\ref{eq:tilde-Lambda-definition}), the second step applies~(\ref{eq:zeta-ell1}),
the third step is valid by~(\ref{eq:Lambda-v-concentration}), and
the fourth step uses the definition~(\ref{eq:Delta-defined-1}).

Finally, (\ref{eq:tilde-Lambda-equals-Lambda}) follows from the definition~(\ref{eq:tilde-Lambda-definition})
in view of~(\ref{eq:zeta-on-k}).
\end{proof}
\begin{proof}[Proof of Claim~\emph{\ref{claim:Psi-nonnegative}.}]
 We will prove the claim in contrapositive form. Specifically, fix
an arbitrary string $x=(x_{1},x_{2},\ldots,x_{\theta})\in(\zoon)^{\theta}$
with $\Psi(x)<0$. Our objective is to deduce that $(f\circ H)(x)=0.$ 

There are two cases to consider. If $|x_{i}|>k$ for some $i,$ then
the defining equation~(\ref{eq:definition-h-1}) implies that $h(x_{i})=1^{N}$.
As a result,
\[
(f\circ H)(x)=f(H(x))=f\left(\bigvee_{i=1}^{\theta}h(x_{i})\right)=f(1^{N})=0,
\]
where the last step uses~(\ref{eq:f-at-11111}).

We now treat the complementary case $x\in(\zoon|_{\leq k})^{\theta}.$
By (\ref{eq:tilde-Lambda-equals-Lambda}) and (\ref{eq:Psi-defined-1}),
\begin{align}
\Psi(x) & =\sum_{u\in\zoo^{N}|_{\leq\theta}}\psi(u)\Exp_{\substack{\mathbf{v}\in V^{\theta}:\\
\mathbf{v}_{1}+\mathbf{v}_{2}+\cdots+\mathbf{v}_{\theta}=u
}
}\Lambda_{\mathbf{v}}(x)\nonumber \\
 & =\sum_{u\in\zoo^{N}|_{\leq\theta}}\psi(u)\Exp_{\substack{\mathbf{v}\in V^{\theta}:\\
\mathbf{v}_{1}+\mathbf{v}_{2}+\cdots+\mathbf{v}_{\theta}=u
}
}\;\prod_{i=1}^{\theta}\lambda_{\mathbf{v}_{i}}(x_{i}).\label{eq:Psi-rewritten}
\end{align}
It follows from $\Psi(x)<0$ that the summation in~(\ref{eq:Psi-rewritten})
contains at least one negative term, corresponding to a string $u\in\zoo^{N}|_{\leq\theta}$.
This forces
\begin{equation}
\psi(u)<0\label{eq:psi-negative-at-u}
\end{equation}
and additionally implies the existence of $\mathbf{v}_{1},\mathbf{v}_{2},\ldots,\mathbf{v}_{\theta}\in V$
with 
\begin{align}
 & x_{i}\in\supp\lambda_{\mathbf{v}_{i}}, &  & i=1,2,\ldots,\theta,\label{eq:x-i-supp-lambda-v-i}\\
 & \sum_{i=1}^{\theta}\mathbf{v}_{i}=u.\label{eq:sum-v-i-equals-u}
\end{align}
Since $x\in(\zoon|_{\leq k})^{\theta}$ in the case under consideration,
it follows from~(\ref{eq:lambda-i-support-1-1}) and (\ref{eq:x-i-supp-lambda-v-i})
that $|x_{i}|=k$ for all $i.$ Now (\ref{eq:h-on-short-inputs})
ensures that $h(x_{i})=\mathbf{v}_{i}$ for all $i,$ which in turn
makes it possible to rewrite~(\ref{eq:sum-v-i-equals-u}) as $\sum_{i=1}^{\theta}h(x_{i})=u.$
Since $u,h(x_{1}),h(x_{2}),\ldots,h(x_{\theta})\in\zoo^{N},$ we conclude
that $\bigvee_{i=1}^{\theta}h(x_{i})=u.$ As a result,
\[
(f\circ H)(x)=f\left(\bigvee_{i=1}^{\theta}h(x_{i})\right)=f(u)=0,
\]
where the last step is immediate from (\ref{eq:psi-nonneg-on-f1})
and (\ref{eq:psi-negative-at-u}).
\end{proof}

\subsection{Specializing the parameters}

Theorems~\ref{thm:hardness-amplification-twosided} and~\ref{thm:hardness-amplification-onesided}
have a large number of parameters that one can adjust to produce various
hardness amplification theorems. We do so in this section. For any
constants $\alpha\in(0,1]$ and $C\geq1,$ we show how to transform
a function $f$ on $\theta^{C}$ bits with approximate degree
\begin{equation}
\deg_{\epsilon}(f|_{\leq\theta})\geq\theta^{1-\alpha}\label{eq:degree-original}
\end{equation}
into a function $F$ on $T^{1+\alpha}$ bits with approximate degree
\begin{equation}
\deg_{\epsilon-\frac{1}{T}}(F|_{\leq T})\geq T^{1-\frac{2}{3}\alpha}.\label{eq:degree-transformed}
\end{equation}
Comparing the exponents in~(\ref{eq:degree-original}) and~(\ref{eq:degree-transformed}),
we see that $F$ is harder to approximate than $f$ relative to the
Hamming weight of the inputs for $F$ and $f$, respectively. Moreover,
we show that $F$ is expressible as $F=f\circ H$ for some mapping
$H$ whose output bits are computable by monotone DNF formulas of
constant width. In particular, if $f$ is a monotone DNF formula of
constant width, then so is $F.$ The formal statement follows.
\begin{cor}
\label{cor:hardness-amplification}Fix reals $\alpha\in(0,1],$ $A\geq1,$
and $C\geq1$ arbitrarily. Then for all large enough integers $\theta,$
there is an $($explicitly given$)$ mapping $H\colon\{0,1\}^{\lfloor T^{1+\alpha}\rfloor}\to\zoo^{\lfloor\theta^{C}\rfloor}$
with $T=\lfloor\theta\log^{2}\theta\rfloor$ such that the output
bits of $H$ are computable by monotone $\lceil50(A+C)/\alpha\rceil$-DNF
formulas and
\begin{equation}
\deg_{\epsilon-\frac{1}{T^{A}}}((f\circ H)|_{\leq T})\geq T^{1-\frac{2}{3}\alpha}\label{eq:transformation-blackbox}
\end{equation}
for every $\epsilon\in[0,1]$ and every function $f\colon\zoo^{\lfloor\theta^{C}\rfloor}\to\zoo$
with $\deg_{\epsilon}(f|_{\leq\theta})\geq\theta^{1-\alpha}$.
\end{cor}

\begin{proof}
Invoke Theorem~\ref{thm:hardness-amplification-twosided} with parameters
\begin{align}
\beta & =\frac{1}{2\theta\lfloor\theta\log^{2}\theta\rfloor^{A}},\label{eq:apply-beta}\\
N & =\lfloor\theta^{C}\rfloor,\\
n & =\lfloor\theta^{\alpha}\rfloor,\\
m & =\lfloor\theta^{\alpha/4}\rfloor,\\
k & =\left\lceil \frac{50(A+C)}{\alpha}\right\rceil -1,\\
D & =\lceil\theta^{1-\frac{5}{8}\alpha}\rceil,\\
T & =\lfloor\theta\log^{2}\theta\rfloor.\label{eq:apply-T}
\end{align}
Provided that $\theta$ is large enough, these parameter settings
satisfy the theorem hypotheses~(\ref{eq:nmk})\textendash (\ref{eq:T-greater-than-D}),
whereas (\ref{eq:Delta-defined}) gives 
\begin{equation}
\Delta\leq\frac{1}{4T^{A}}.\label{eq:apply-Delta}
\end{equation}
As a result, Theorem~\ref{thm:hardness-amplification-twosided} guarantees
that
\begin{equation}
\deg_{\epsilon-\frac{1}{T^{A}}}((f\circ H)|_{\leq T})\geq\frac{c}{2}\cdot\theta^{1-\frac{5}{8}\alpha},\label{eq:conclusion-deg-composition}
\end{equation}
where $c\in(0,1)$ is the absolute constant from Theorem~\ref{thm:encoding-1-r-normalized}
and $H\colon\zoo^{\lfloor T^{1+\alpha}\rfloor}\to\zoo^{\lfloor\theta^{C}\rfloor}$
is an explicit mapping whose output bits are computable by monotone
$\lceil50(A+C)/\alpha\rceil$-DNF formulas. (In fact, $H$ uses only
$n\theta\approx(T/\log^{2}T)^{1+\alpha}$ input bits, but this improvement
is not relevant for our purposes.) Provided that $\theta$ is large
enough relative to the absolute constant $c,$ we infer~(\ref{eq:transformation-blackbox})
immediately from~(\ref{eq:conclusion-deg-composition}).
\end{proof}
Analogously, we have the following hardness amplification result for
\emph{one-sided} approximate degree.
\begin{cor}
\label{cor:hardness-amplification-onesided}Fix reals $\alpha\in(0,1],$
$A\geq1,$ and $C\geq1$ arbitrarily. Then for all large enough integers
$\theta,$ there is an $($explicitly given$)$ mapping $H\colon\{0,1\}^{\lfloor T^{1+\alpha}\rfloor}\to\zoo^{\lfloor\theta^{C}\rfloor}$
with $T=\lfloor\theta\log^{2}\theta\rfloor$ such that the output
bits of $H$ are computable by monotone $\lceil50(A+C)/\alpha\rceil$-DNF
formulas and
\begin{equation}
\deg_{\epsilon-\frac{1}{T^{A}}}^{+}((f\circ H)|_{\leq T})\geq T^{1-\frac{2}{3}\alpha}\label{eq:transformation-blackbox-onesided}
\end{equation}
for every $\epsilon\in[0,1]$ and every function $f\colon\zoo^{\lfloor\theta^{C}\rfloor}\to\zoo$
such that $\deg_{\epsilon}^{+}(f|_{\leq\theta})\geq\theta^{1-\alpha}$
and $f(1^{\lfloor\theta^{C}\rfloor})=0.$
\end{cor}

\begin{proof}
The proof is the same, mutatis mutandis, as that of Corollary~\ref{cor:hardness-amplification}.
Specifically, invoke Theorem~\ref{thm:hardness-amplification-onesided}
with parameters~(\ref{eq:apply-beta})\textendash (\ref{eq:apply-T}).
Provided that $\theta$ is large enough, these parameter settings
satisfy the theorem hypotheses~(\ref{eq:nmk-1})\textendash (\ref{eq:T-greater-than-D-plus-n}),
whereas (\ref{eq:Delta-defined-1}) gives~(\ref{eq:apply-Delta}).
As a result, Theorem~\ref{thm:hardness-amplification-onesided} guarantees
that
\begin{equation}
\deg_{\epsilon-\frac{1}{T^{A}}}^{+}((f\circ H)|_{\leq T})\geq\frac{c}{2}\cdot\theta^{1-\frac{5}{8}\alpha},\label{eq:conclusion-deg-composition-1}
\end{equation}
where $c\in(0,1)$ is the absolute constant from Theorem~\ref{thm:encoding-1-r-normalized}
and $H\colon\zoo^{\lfloor T^{1+\alpha}\rfloor}\to\zoo^{\lfloor\theta^{C}\rfloor}$
is an explicit mapping whose output bits are computable by monotone
$\lceil50(A+C)/\alpha\rceil$-DNF formulas. Provided that $\theta$
is large enough relative to the absolute constant $c,$ this settles~(\ref{eq:transformation-blackbox-onesided}).
\end{proof}

\section{Main results}

In this section, we will settle our main results on approximate degree
and present their applications to communication complexity.

\subsection{Approximate degree of DNF and CNF formulas}

We will start with the two-sided case. Our proof here amounts to taking
the trivial one-variable formula $x_{1}$ and iteratively applying
the hardness amplification of Corollary~\ref{cor:hardness-amplification}.
\begin{thm}
\label{thm:adeg-dnf}For every $\delta\in(0,1]$ and $\Delta\geq1,$
there is a constant $c\geq1$ and an $($explicitly given$)$ family
$\{f_{n}\}_{n=1}^{\infty}$ of functions $f_{n}\colon\{0,1\}^{n}\to\zoo$
such that each $f_{n}$ is computable by a monotone $c$-DNF formula
and satisfies
\begin{align}
\deg_{\frac{1}{2}-\frac{1}{n^{\Delta}}}(f_{n}) & \geq\frac{1}{c}\cdot n^{1-\delta}, &  & n=1,2,3,\ldots.\label{eq:hard-dnf}
\end{align}
\end{thm}

\begin{proof}
Let $K\geq1$ be the smallest integer such that 
\begin{align}
 & \frac{1-(2/3)^{K}}{1+(2/3)^{K-1}}>1-\delta.\label{eq:K-defined}
\end{align}
Define
\begin{equation}
A=2\Delta+3.\label{eq:A-Delta}
\end{equation}
Now, let $n\geq1$ be any large enough integer. Define $T_{0},T_{1},T_{2},\ldots,T_{K}$
recursively by $T_{0}=\lfloor n/\log^{2K}n\rfloor$ and $T_{i}=\lfloor T_{i-1}\log^{2}T_{i-1}\rfloor$
for $i\geq1.$ Thus,
\begin{align}
T_{i} & \leq\frac{n}{\log^{2(K-i)}n}, &  & i=0,1,2,\ldots,K,\label{eq:T-i-upper-bound}\\
T_{i} & \sim\frac{n}{\log^{2(K-i)}n}, &  & i=0,1,2,\ldots,K,\label{eq:T-i-asymptotic}
\end{align}
where $\sim$ denotes equality up to lower-order terms. Provided that
$n$ is larger than a certain constant, inductive application of Corollary~\ref{cor:hardness-amplification}
gives functions
\begin{align}
g_{n,i}\colon\zoo^{\lfloor T_{i}^{1+(2/3)^{i-1}}\rfloor} & \to\zoo, &  & i=0,1,2,\ldots,K,
\end{align}
 such that
\begin{align}
\deg_{\frac{1}{2}-\frac{1}{T_{0}^{A}}-\frac{1}{T_{1}^{A}}-\cdots-\frac{1}{T_{i}^{A}}}(g_{n,i}|_{\leq T_{i}}) & \geq T_{i}^{1-(2/3)^{i}}, &  & i=0,1,2,\ldots,K,\label{eq:inductive-approx-degree}
\end{align}
and each $g_{n,i}$ is an explicitly constructed monotone $c_{i}$-DNF
formula for some constant $c_{i}$ independent of $n.$ In more detail,
the requirement~(\ref{eq:inductive-approx-degree}) for $i=0$ is
equivalent to $\deg_{\frac{1}{2}-\frac{1}{T_{0}^{A}}}(g_{n,0}|_{\leq T_{0}})>0$
and is trivially satisfied by the ``dictator'' function $g_{n,0}(x)=x_{1}$,
whereas for $i\geq1$ the function $g_{n,i}$ is obtained constructively
from $g_{n,i-1}$ by invoking Corollary~\ref{cor:hardness-amplification}
with 
\begin{align*}
\alpha & =\left(\frac{2}{3}\right)^{i-1},\\
C & =1+\left(\frac{2}{3}\right)^{i-2},\\
\theta & =T_{i-1},\\
f & =g_{n,i-1},\\
\epsilon & =\frac{1}{2}-\frac{1}{T_{0}^{A}}-\frac{1}{T_{1}^{A}}-\cdots-\frac{1}{T_{i-1}^{A}}.
\end{align*}

Specializing~(\ref{eq:T-i-upper-bound})\textendash (\ref{eq:inductive-approx-degree})
to $i=K$, the function $g_{n,K}$ is a monotone $c_{K}$-DNF formula
for some constant $c_{K}$ independent of $n,$ takes at most $N:=n^{1+(2/3)^{K-1}}$
input variables, and has approximate degree 
\begin{align*}
\deg_{\frac{1}{2}-\frac{1}{N^{\Delta+1}}}(g_{n,K}) & \geq\deg_{\frac{1}{2}-\frac{1}{T_{0}^{A}}-\frac{1}{T_{1}^{A}}-\cdots-\frac{1}{T_{K}^{A}}}(g_{n,K})\\
 & \geq\deg_{\frac{1}{2}-\frac{1}{T_{0}^{A}}-\frac{1}{T_{1}^{A}}-\cdots-\frac{1}{T_{K}^{A}}}(g_{n,K}|_{\leq T_{K}})\\
 & =\Omega(n^{1-(2/3)^{K}})\\
 & =\omega(N^{1-\delta}),
\end{align*}
where the first and last steps hold for all large enough $n$ due
to~(\ref{eq:A-Delta}) and~(\ref{eq:K-defined}), respectively.
The desired function family $\{f_{n}\}_{n=1}^{\infty}$ can then be
defined by setting 
\[
f_{n}=g_{\lfloor n^{1/(1+(2/3)^{K-1})}\rfloor,K}
\]
for all $n$ larger than a certain constant $n_{0},$ and taking the
remaining functions $f_{1},f_{2},\ldots,f_{n_{0}}$ to be the dictator
function $x\mapsto x_{1}.$
\end{proof}
Theorem~\ref{thm:adeg-dnf} immediately implies Theorems~\ref{thm:MAIN-dnf}
and \ref{thm:MAIN-dnf-low-error} from the introduction. We now move
on to the one-sided case.
\begin{thm}
\label{thm:adeg-dnf-onesided}For every $\delta\in(0,1]$ and $\Delta\geq1,$
there is a constant $c\geq1$ and an $($explicitly given$)$ family
$\{f_{n}\}_{n=1}^{\infty}$ of functions $f_{n}\colon\{0,1\}^{n}\to\zoo$
such that each $f_{n}$ is computable by a monotone $c$-DNF formula
and satisfies
\begin{align}
\deg_{\frac{1}{2}-\frac{1}{n^{\Delta}}}^{+}(\neg f_{n}) & \geq\frac{1}{c}\cdot n^{1-\delta}, &  & n=1,2,3,\ldots.\label{eq:hard-dnf-onesided}
\end{align}
\end{thm}

\noindent This result subsumes Theorem~\ref{thm:adeg-dnf} and settles
Theorem~\ref{thm:MAIN-one-sided} in the introduction. The proof
below makes repeated use of the following observation: if one applies
Corollary~\ref{cor:hardness-amplification-onesided} to a function
$f$ that is the \emph{negation} of a constant-width monotone DNF
formula, then the resulting composition $f\circ H$ is again the \emph{negation}
of a constant-width monotone DNF formula. This is easy to see by writing
$\neg(f\circ H)=(\neg f)\circ H$ and noting that both $\neg f$ and
$H$ are computable by constant-width monotone DNF formulas. 
\begin{proof}[Proof of Theorem~\emph{\ref{thm:adeg-dnf-onesided}}.]
 Much of the proof is identical to that of Theorem~\ref{thm:adeg-dnf}.
As before, let $K\geq1$ be the smallest integer such that 
\begin{align}
 & \frac{1-(2/3)^{K}}{1+(2/3)^{K-1}}>1-\delta.\label{eq:K-defined-1}
\end{align}
Define
\begin{equation}
A=2\Delta+3.\label{eq:A-Delta-1}
\end{equation}
Now, let $n\geq1$ be any large enough integer. Define $T_{0},T_{1},T_{2},\ldots,T_{K}$
recursively by $T_{0}=\lfloor n/\log^{2K}n\rfloor$ and $T_{i}=\lfloor T_{i-1}\log^{2}T_{i-1}\rfloor$
for $i\geq1.$ Thus,
\begin{align}
T_{i} & \leq\frac{n}{\log^{2(K-i)}n}, &  & i=0,1,2,\ldots,K,\label{eq:T-i-upper-bound-1}\\
T_{i} & \sim\frac{n}{\log^{2(K-i)}n}, &  & i=0,1,2,\ldots,K,\label{eq:T-i-asymptotic-1}
\end{align}
where $\sim$ denotes equality up to lower-order terms. Provided that
$n$ is larger than a certain constant, inductive application of Corollary~\ref{cor:hardness-amplification-onesided}
gives functions
\begin{align}
g_{n,i}\colon\zoo^{\lfloor T_{i}^{1+(2/3)^{i-1}}\rfloor} & \to\zoo, &  & i=0,1,2,\ldots,K,
\end{align}
 such that
\begin{align}
\deg_{\frac{1}{2}-\frac{1}{T_{0}^{A}}-\frac{1}{T_{1}^{A}}-\cdots-\frac{1}{T_{i}^{A}}}^{+}(\neg g_{n,i}|_{\leq T_{i}}) & \geq T_{i}^{1-(2/3)^{i}}, &  & i=0,1,2,\ldots,K,\label{eq:inductive-approx-degree-1}
\end{align}
and each $g_{n,i}$ is an explicitly constructed monotone $c_{i}$-DNF
formula for some constant $c_{i}$ independent of $n.$ In more detail,
the requirement~(\ref{eq:inductive-approx-degree-1}) for $i=0$
is equivalent to $\deg_{\frac{1}{2}-\frac{1}{T_{0}^{A}}}^{+}(\neg g_{n,0}|_{\leq T_{0}})>0$
and is trivially satisfied by the ``dictator'' function $g_{n,0}(x)=x_{1}$.
For $i\geq1$, we obtain $g_{n,i}$ from $g_{n,i-1}$ by applying
Corollary~\ref{cor:hardness-amplification-onesided} with 
\begin{align*}
\alpha & =\left(\frac{2}{3}\right)^{i-1},\\
C & =1+\left(\frac{2}{3}\right)^{i-2},\\
\theta & =T_{i-1},\\
f & =\neg g_{n,i-1},\\
\epsilon & =\frac{1}{2}-\frac{1}{T_{0}^{A}}-\frac{1}{T_{1}^{A}}-\cdots-\frac{1}{T_{i-1}^{A}}.
\end{align*}
This appeal to Corollary~\ref{cor:hardness-amplification-onesided}
is legitimate because $g_{n,i-1}$ is a monotone DNF formula and therefore
its negation $f=\neg g_{n,i-1}$ evaluates to $0$ on the all-ones
input. 

Specializing~(\ref{eq:T-i-upper-bound-1})\textendash (\ref{eq:inductive-approx-degree-1})
to $i=K$, the function $g_{n,K}$ is a monotone $c_{K}$-DNF formula
for some constant $c_{K}$ independent of $n,$ takes at most $N:=n^{1+(2/3)^{K-1}}$
input variables, and has one-sided approximate degree 
\begin{align*}
\deg_{\frac{1}{2}-\frac{1}{N^{\Delta+1}}}^{+}(\neg g_{n,K}) & \geq\deg_{\frac{1}{2}-\frac{1}{T_{0}^{A}}-\frac{1}{T_{1}^{A}}-\cdots-\frac{1}{T_{K}^{A}}}^{+}(\neg g_{n,K})\\
 & \geq\deg_{\frac{1}{2}-\frac{1}{T_{0}^{A}}-\frac{1}{T_{1}^{A}}-\cdots-\frac{1}{T_{K}^{A}}}^{+}(\neg g_{n,K}|_{\leq T_{K}})\\
 & =\Omega(n^{1-(2/3)^{K}})\\
 & =\omega(N^{1-\delta}),
\end{align*}
where the first and last steps hold for all large enough $n$ due
to~(\ref{eq:A-Delta-1}) and~(\ref{eq:K-defined-1}), respectively.
The desired function family $\{f_{n}\}_{n=1}^{\infty}$ can then be
defined by setting 
\[
f_{n}=g_{\lfloor n^{1/(1+(2/3)^{K-1})}\rfloor,K}
\]
for all $n$ larger than a certain constant $n_{0},$ and taking the
remaining functions $f_{1},f_{2},\ldots,f_{n_{0}}$ to be the dictator
function $x\mapsto x_{1}.$
\end{proof}

\subsection{Quantum communication complexity}

Using the pattern matrix method, we will ``lift'' our approximate
degree results to a near-optimal lower bound on the communication
complexity of DNF formulas in the two-party quantum model. Before
we can apply the pattern matrix method, there is a technicality to
address with regard to the representation of Boolean values as real
numbers. In this paper, we have followed the standard convention of
representing ``true'' and ``false'' as $1$ and $0,$ respectively.
There is another common encoding, inspired by Fourier analysis and
used in the pattern matrix method~\cite{sherstov07quantum,sherstov13directional},
whereby ``true'' and ``false'' are represented as $-1$ and $1,$
respectively. To switch back and forth between these representations,
we will use the following proposition.
\begin{prop}
\label{prop:adeg-scaling-translation}For any function $f\colon X\to\Re$
on a finite subset $X$ of Euclidean space, and any reals $\epsilon\geq0$
and $c\ne0,$
\begin{align*}
 & \deg_{\epsilon}(f+c)=\deg_{\epsilon}(f),\\
 & \deg_{|c|\epsilon}(cf)=\deg_{\epsilon}(f).
\end{align*}
\end{prop}

\begin{proof}
For any polynomial $p,$ we have the following equivalences:
\begin{align*}
 & \|f-p\|_{\infty}\leq\epsilon\qquad\Leftrightarrow\qquad\|(f+c)-(p+c)\|_{\infty}\leq\epsilon,\\
 & \|f-p\|_{\infty}\leq\epsilon\qquad\Leftrightarrow\qquad\|cf-cp\|_{\infty}\leq|c|\epsilon,
\end{align*}
where the second line uses $c\ne0.$
\end{proof}
\noindent As a corollary, we can relate in a precise way the approximate
degree of a Boolean function $f\colon X\to\zoo$ and the approximate
degree of the associated $\pm1$-valued function $f'\colon X\to\pomo$
given by $f'=$$(-1)^{f}.$
\begin{cor}
\label{cor:adeg-zoo-pomo}For any Boolean function $f\colon X\to\zoo$
and any $\epsilon\geq0,$
\[
\deg_{\epsilon}((-1)^{f})=\deg_{\epsilon/2}(f).
\]
\end{cor}

\begin{proof}
Since $f$ is Boolean-valued, we have the equality of functions $(-1)^{f}=1-2f.$
Now $\deg_{\epsilon}((-1)^{f})=\deg_{\epsilon}(1-2f)=\deg_{\epsilon}(-2f)=\deg_{2\cdot\epsilon/2}(-2f)=\deg_{\epsilon/2}(f),$
where the second and fourth steps apply Proposition~\ref{prop:adeg-scaling-translation}.
\end{proof}
Corollary~\ref{cor:adeg-zoo-pomo} makes it easy to convert approximate
degree results between the $0,1$ representation and $\pm1$ representation.
For communication complexity, no conversion is necessary in the first
place:
\begin{align}
 & Q_{\epsilon}^{*}(F)=Q_{\epsilon}^{*}((-1)^{F}), &  & F\colon X\times Y\to\zoo,\label{eq:quantum-pomon-zoon}
\end{align}
where $Q_{\epsilon}^{*}$ denotes $\epsilon$-error quantum communication
complexity with arbitrary prior entanglement. This equality holds
because the representation of ``true'' and ``false'' in a communication
protocol is a purely syntactic matter, and one can relabel the output
values $0,1$ as $1,-1$, respectively, without affecting the protocol's
correctness or communication cost. We note that~(\ref{eq:quantum-pomon-zoon})
and Corollary~\ref{cor:adeg-zoo-pomo} pertain to the encoding of
the \emph{output }of a Boolean function $f$. How ``true'' and ``false''
bits are represented in the \emph{input} to $f$ is immaterial both
for communication complexity and approximate degree because the bijection
$(0,1)\leftrightarrow(1,-1)$ is a linear map.

We are now in a position to prove the promised communication lower
bounds. The pattern matrix method for two-party quantum communication
is given by the following theorem~\cite[Theorem~1.1]{sherstov07quantum}.
\begin{thm}[Sherstov]
\label{thm:pattern-matrix-quantum}Let $f\colon\zoo^{t}\to\zoo$
be given. Define $F\colon\zoo^{4t}\times\zoo^{4t}\to\zoo$ by 
\[
F(x,y)=f\left(\bigvee_{i=1}^{4}(x_{1,i}\wedge y_{1,i}),\ldots,\bigvee_{i=1}^{4}(x_{t,i}\wedge y_{t,i})\right).
\]
 Then for all $\alpha\in[0,1)$ and $\beta<\alpha/2,$
\[
Q_{\beta}^{*}(F)\geq\frac{1}{4}\deg_{\alpha/2}(f)-\frac{1}{2}\log\left(\frac{3}{\alpha-2\beta}\right).
\]
\end{thm}

\noindent The original statement in~\cite[Theorem~1.1]{sherstov07quantum}
uses the $\pm1$ representation for the range of $f$ and $F.$ We
translated it to the $0,1$ representation, as stated in Theorem~\ref{thm:pattern-matrix-quantum},
by applying (\ref{eq:quantum-pomon-zoon}) to $F$ and Corollary~\ref{cor:adeg-zoo-pomo}
to $f.$ By combining Theorems~\ref{thm:adeg-dnf} and~\ref{thm:pattern-matrix-quantum},
we obtain our main result on the quantum communication complexity
of DNF formulas:
\begin{thm}
\label{thm:quantum-cc} For all $\delta\in(0,1]$ and $A\geq1,$ there
is a constant $c\geq1$ and an $($explicitly given$)$ family $\{F_{n}\}_{n=1}^{\infty}$
of two-party communication problems $F_{n}\colon\{0,1\}^{n}\times\zoon\to\zoo$
such that each $F_{n}$ is computable by a monotone $c$-DNF formula
and satisfies
\begin{align}
Q_{\frac{1}{2}-\frac{1}{n^{A}}}^{*}(F_{n}) & =\Omega(n^{1-\delta}).\label{eq:Q-dnf}
\end{align}
\end{thm}

\begin{proof}
Theorem~\ref{thm:adeg-dnf} gives a constant $c'\geq1$ and an explicit
family $\{f_{n}\}_{n=1}^{\infty}$ of functions $f_{n}\colon\{0,1\}^{n}\to\zoo$
such that each $f_{n}$ is computable by a monotone $c'$-DNF formula
and satisfies
\begin{align}
\deg_{\frac{1}{2}-\frac{1}{n^{2A}}}(f_{n}) & \geq\frac{1}{c'}\cdot n^{1-\delta}, &  & n=1,2,3,\ldots.\label{eq:hard-dnf-3}
\end{align}
For $n\geq4,$ define $F_{n}\colon\zoon\times\zoon\to\zoo$ by
\[
F_{n}(x,y)=f_{\lfloor n/4\rfloor}\left(\bigvee_{i=1}^{4}(x_{1,i}\wedge y_{1,i}),\ldots,\bigvee_{i=1}^{4}(x_{\lfloor n/4\rfloor,i}\wedge y_{\lfloor n/4\rfloor,i})\right),
\]
where we index the strings $x$ and $y$ as arrays of $\lfloor n/4\rfloor\times4$
bits. Clearly, $F_{n}$ is computable by a monotone $2c'$-DNF formula.
We now invoke the pattern matrix method for quantum communication
(Theorem~\ref{thm:pattern-matrix-quantum}) with parameters
\begin{align*}
\alpha & =1-\frac{2}{\lfloor n/4\rfloor^{2A}},\\
\beta & =\frac{1}{2}-\frac{1}{n^{A}},\\
f & =f_{\lfloor n/4\rfloor},
\end{align*}
which satisfy $\beta<\alpha/2$ for all $n\geq24.$ As a result, 
\begin{align*}
Q_{\frac{1}{2}-\frac{1}{n^{A}}}^{*}(F_{n}) & \geq\frac{1}{4}\cdot\deg_{\frac{1}{2}-\frac{1}{\lfloor n/4\rfloor^{2A}}}(f_{\lfloor n/4\rfloor})-\frac{1}{2}\log\left(\frac{3}{\frac{2}{n^{A}}-\frac{2}{\lfloor n/4\rfloor^{2A}}}\right)\\
 & \geq\frac{1}{4}\cdot\frac{1}{c'}\cdot\left\lfloor \frac{n}{4}\right\rfloor ^{1-\delta}-\frac{1}{2}\log\left(\frac{3}{\frac{2}{n^{A}}-\frac{2}{\lfloor n/4\rfloor^{2A}}}\right)
\end{align*}
for all $n\geq24,$ where the first inequality applies the pattern
matrix method, and the second inequality uses~(\ref{eq:hard-dnf-3}).
Now~(\ref{eq:Q-dnf}) follows since $A,c',\delta$ are constants.
\end{proof}
\noindent Theorem~\ref{thm:quantum-cc} settles Theorem~\ref{thm:MAIN-quantum}
from the introduction.

\subsection{Randomized multiparty communication}

We now turn to communication lower bounds for DNF formulas in the
$k$-party number-on-the-forehead model. Analogous to (\ref{eq:quantum-pomon-zoon}),
we have
\begin{align}
 & R_{\epsilon}(F)=R_{\epsilon}((-1)^{F}), &  & F\colon X_{1}\times X_{2}\times\cdots\times X_{k}\to\zoo,\label{eq:multiparty-pomon-zoon}
\end{align}
where $R_{\epsilon}$ denotes $\epsilon$-error number-on-the-forehead
randomized communication complexity. The \emph{$k$-party set disjointness
problem} $\DISJ_{n,k}\colon(\zoon)^{k}\to\zoo$ is given by 
\[
\DISJ_{n,k}(x_{1},x_{2},\ldots,x_{k})=\bigwedge_{j=1}^{n}\bigvee_{i=1}^{k}\overline{x_{i,j}}.
\]
In other words, the problem asks whether there is a coordinate $j$
in which each of the Boolean vectors $x_{1},x_{2},\ldots,x_{k}$ has
a $1.$ If one views $x_{1},x_{2},\ldots,x_{k}$ as the characteristic
vectors of corresponding sets $S_{1},S_{2},\ldots,S_{k}$, then the
set disjointness function evaluates to true if and only if $S_{1}\cap S_{2}\cap\cdots\cap S_{k}=\varnothing.$
For a communication problem $g\colon X_{1}\times X_{2}\times\cdots\times X_{k}\to\zoo$
and a function $f\colon\zoon\to\zoo,$ we view the componentwise composition
$f\circ g$ as a $k$-party communication problem on $X_{1}^{n}\times X_{2}^{n}\times\cdots\times X_{k}^{n}.$
The multiparty pattern matrix method~\cite[Theorem~5.1]{sherstov13directional}
gives a lower bound on the communication complexity of $f\circ\DISJ_{m,k}$
in terms of the approximate degree of $f$:
\begin{thm}[Sherstov]
\label{thm:pattern-matrix-randomized}Let $f\colon\zoo^{n}\to\zoo$
be given. Consider the $k$-party communication problem $F$ defined
by $F=f\circ\DISJ_{m,k}.$ Then for all $\alpha,\beta\geq0$ with
$\beta<\alpha/2,$ one has 
\[
R_{\beta}(F)\geq\frac{\deg_{\alpha/2}(f)}{2}\cdot\log\left(\frac{\sqrt{m}}{C2^{k}k}\right)-\log\frac{1}{\alpha-2\beta},
\]
where $C>0$ is an absolute constant.
\end{thm}

\noindent 

\noindent The actual statement of the pattern matrix method in~\cite[Theorem~5.1]{sherstov13directional}
is for functions $f$ and $F$ with range $\pomo$. Theorem~\ref{thm:pattern-matrix-randomized}
above, stated for functions with range $\{0,1\}$, is immediate from~\cite[Theorem~5.1]{sherstov13directional}
by applying~(\ref{eq:multiparty-pomon-zoon}) to $F$ and Corollary~\ref{cor:adeg-zoo-pomo}
to $f$. We are now ready for our main result on the randomized multiparty
communication complexity of DNF formulas.
\begin{thm}
\label{thm:multiparty-cc} Fix arbitrary constants $\delta\in(0,1]$
and $A\geq1$. Then for all integers $n,k\geq2,$ there is an $($explicitly
given$)$ $k$-party communication problem $F_{n,k}\colon(\zoon)^{k}\to\zoo$
with
\begin{align}
 & R_{1/3}(F_{n,k})\geq\left(\frac{n}{c4^{k}k^{2}}\right)^{1-\delta},\label{eq:R-dnf-const-error}\\
 & R_{\frac{1}{2}-\frac{1}{n^{A}}}(F_{n,k})\geq\frac{n^{1-\delta}}{c4^{k}},\label{eq:R-dnf-large-error}
\end{align}
where $c\geq1$ is a constant independent of $n$ and $k.$ Moreover,
each $F_{n,k}$ is computable by a monotone DNF formula of width $ck$
and size $n^{c}$.
\end{thm}

It will be helpful to keep in mind that the conclusion of Theorem~\ref{thm:multiparty-cc}
is ``monotone'' in $c,$ in the sense that proving Theorem~\ref{thm:multiparty-cc}
for a given constant $c$ proves it for all larger constants as well.
\begin{proof}
Theorem~\ref{thm:adeg-dnf} gives a constant $c'\geq1$ and an explicit
family $\{f_{n}\}_{n=1}^{\infty}$ of functions $f_{n}\colon\{0,1\}^{n}\to\zoo$
such that each $f_{n}$ is computable by a monotone DNF formula of
width $c'$ and satisfies
\begin{align}
\deg_{\frac{1}{2}-\frac{1}{n^{2A/\delta}}}(f_{n}) & \geq\frac{1}{c'}\cdot n^{1-\frac{\delta}{2}}, &  & n=1,2,3,\ldots.\label{eq:hard-dnf-f}
\end{align}
Let $C>0$ be the absolute constant from Theorem~\ref{thm:pattern-matrix-randomized}.
For arbitrary integers $n,k\geq2,$ define
\[
F_{n,k}=\begin{cases}
\AND_{k} & \text{if }n<\lceil C2^{k+1}k\rceil^{2},\\
f_{\lfloor n/\lceil C2^{k+1}k\rceil^{2}\rfloor}\circ\neg\DISJ_{\lceil C2^{k+1}k\rceil^{2},k} & \text{otherwise.}
\end{cases}
\]
We first analyze the cost of representing $F_{n,k}$ as a DNF formula.
If $n<\lceil C2^{k+1}k\rceil^{2},$ then by definition $F_{n,k}$
is a monotone DNF formula of width $k$ and size $1.$ In the complementary
case, $f_{\lfloor n/\lceil C2^{k+1}k\rceil^{2}\rfloor}$ is by construction
a monotone DNF formula of width $c'$ and hence of size at most $n^{c'},$
whereas $\neg\DISJ_{\lceil C2^{k+1}k\rceil^{2},k}$ is by definition
a monotone DNF formula of width $k$ and size at most $\lceil C2^{k+1}k\rceil^{2}\leq n.$
As a result, the composed function $F_{n,k}$ is a monotone DNF formula
of width $c'k$ and size at most $n^{c'}\cdot n^{c'}=n^{2c'}.$ In
particular, the claim in the theorem statement regarding the width
and size of $F_{n,k}$ as a monotone DNF formula is valid for any
constant $c\geq2c'.$

We now turn to the communication complexity of $F_{n,k}.$ Since $F_{n,k}$
is nonconstant, we have the trivial bound
\begin{align}
R_{\epsilon}(F_{n,k}) & \geq1, &  & 0\leq\epsilon<\frac{1}{2}.\label{eq:F-nk-trivial-cc}
\end{align}
We further claim that
\begin{multline}
R_{\epsilon}(F_{n,k})\geq\frac{1}{2c'}\cdot\left\lfloor \frac{n}{\lceil C2^{k+1}k\rceil^{2}}\right\rfloor ^{1-\frac{\delta}{2}}\\
+\log\left(1-\frac{2}{\lfloor n/\lceil C2^{k+1}k\rceil^{2}\rfloor^{2A/\delta}}-2\epsilon\right)\qquad\label{eq:R-F-nk-general-error}
\end{multline}
whenever the logarithmic term is well-defined. For $n<\lceil C2^{k+1}k\rceil^{2},$
this claim is vacuous. In the complementary case $n\geq\lceil C2^{k+1}k\rceil^{2},$
consider the family $\{g_{n}\}_{n=1}^{\infty}$ of functions $g_{n}\colon\zoon\to\zoo$
given by $g_{n}(x_{1},x_{2},\ldots,x_{n})=f_{n}(\neg x_{1},\neg x_{2},\ldots,\neg x_{n})$.
For each $n,$ it is clear that $g_{n}$ and $f_{n}$ have the same
approximate degree. Since $F_{n,k}=g_{\lfloor n/\lceil C2^{k+1}k\rceil^{2}\rfloor}\circ\DISJ_{\lceil C2^{k+1}k\rceil^{2},k},$
one now obtains~(\ref{eq:R-F-nk-general-error}) directly from~(\ref{eq:hard-dnf-f})
and the multiparty pattern matrix method (Theorem~\ref{thm:pattern-matrix-randomized}).

For a sufficiently large constant $c\geq1$, the communication lower
bound~(\ref{eq:R-dnf-const-error}) follows from~(\ref{eq:R-F-nk-general-error})
for $n\geq c4^{k}k^{2}$ and follows from~(\ref{eq:F-nk-trivial-cc})
for $n<c4^{k}k^{2}$.

The proof of~(\ref{eq:R-dnf-large-error}) is more tedious. Take
the constant $c\geq1$ large enough that the following relations hold:
\begin{align}
 & \left\lfloor \frac{n^{\delta}}{\lceil2C\rceil^{2}\log^{2}n}\right\rfloor \geq2n^{\delta/2} &  & \text{for all }n\geq c,\label{eq:many-vars}\\
 & n^{\delta^{2}/4}\geq1+A\log n &  & \text{for all }n\geq c,\label{eq:n-log}\\
 & c\geq2c'.\label{eq:c-c-prime}
\end{align}
If $n^{1-\delta}<c4^{k},$ then~(\ref{eq:R-dnf-large-error}) holds
due to~(\ref{eq:F-nk-trivial-cc}). In what follows, we treat the
complementary case when
\begin{equation}
\frac{n^{1-\delta}}{c4^{k}}\geq1\label{eq:n-ge-c4k}
\end{equation}
and in particular 
\begin{align}
n & \geq c,\label{eq:n-ge-c}\\
k & \leq\log n.\label{eq:k-le-logn}
\end{align}
Then
\begin{align}
\left\lfloor \frac{n}{\lceil C2^{k+1}k\rceil^{2}}\right\rfloor ^{1-\frac{\delta}{2}} & \geq\left\lfloor \frac{n}{\lceil2C\rceil^{2}4^{k}k^{2}}\right\rfloor ^{1-\frac{\delta}{2}}\nonumber \\
 & \geq\left\lfloor \frac{n^{1-\delta}}{4^{k}}\cdot\frac{n^{\delta}}{\lceil2C\rceil^{2}\log^{2}n}\right\rfloor ^{1-\frac{\delta}{2}}\nonumber \\
 & \geq\left\lfloor \frac{n^{1-\delta}}{4^{k}}\cdot2n^{\delta/2}\right\rfloor ^{1-\frac{\delta}{2}}\nonumber \\
 & \geq\left(\frac{n^{1-\delta}}{4^{k}}\cdot n^{\delta/2}\right)^{1-\frac{\delta}{2}}\nonumber \\
 & \geq\frac{n^{1-\delta}}{4^{k}}\cdot n^{\delta^{2}/4}\nonumber \\
 & \geq\frac{n^{1-\delta}}{4^{k}}\cdot(1+A\log n),\label{eq:first-summand}
\end{align}
where the second step uses~(\ref{eq:k-le-logn}), the third step
uses~(\ref{eq:many-vars}) and~(\ref{eq:n-ge-c}), the fourth step
is valid by~(\ref{eq:n-ge-c4k}), and the last step uses~(\ref{eq:n-log})
and~(\ref{eq:n-ge-c}). Continuing,
\begin{align}
 & \log\left(1-\frac{2}{\lfloor n/\lceil C2^{k+1}k\rceil^{2}\rfloor^{2A/\delta}}-2\left(\frac{1}{2}-\frac{1}{n^{A}}\right)\right)\nonumber \\
 & \hspace{2cm}=\log\left(\frac{2}{n^{A}}-\frac{2}{\lfloor n/\lceil C2^{k+1}k\rceil^{2}\rfloor^{2A/\delta}}\right)\nonumber \\
 & \hspace{2cm}\geq\log\left(\frac{2}{n^{A}}-\frac{2}{\lfloor n/(\lceil2C\rceil^{2}4^{k}k^{2})\rfloor^{2A/\delta}}\right)\nonumber \\
 & \hspace{2cm}\geq\log\left(\frac{2}{n^{A}}-\frac{2}{\lfloor n/(\lceil2C\rceil^{2}n^{1-\delta}\log^{2}n)\rfloor^{2A/\delta}}\right)\nonumber \\
 & \hspace{2cm}\geq\log\left(\frac{2}{n^{A}}-\frac{2}{(2n^{\delta/2})^{2A/\delta}}\right)\nonumber \\
 & \hspace{2cm}\geq\log\left(\frac{2}{n^{A}}-\frac{1}{n^{A}}\right)\nonumber \\
 & \hspace{2cm}=-A\log n,\label{eq:second-summand}
\end{align}
where the third step uses~(\ref{eq:n-ge-c4k}) and~(\ref{eq:k-le-logn}),
and the fourth step uses~(\ref{eq:many-vars}) and~(\ref{eq:n-ge-c}).
Now

\begin{align*}
R_{\frac{1}{2}-\frac{1}{n^{A}}}(F_{n,k}) & \geq\frac{1}{2c'}\cdot\frac{n^{1-\delta}}{4^{k}}\cdot(1+A\log n)-A\log n\\
 & \geq\frac{n^{1-\delta}}{c4^{k}}\cdot(1+A\log n)-A\log n\\
 & \geq\frac{n^{1-\delta}}{c4^{k}},
\end{align*}
where the first step substitutes the bounds~(\ref{eq:first-summand})
and~(\ref{eq:second-summand}) into~(\ref{eq:R-F-nk-general-error}),
the second step uses~(\ref{eq:c-c-prime}), and the third step is
valid by~(\ref{eq:n-ge-c4k}). This completes the proof of~(\ref{eq:R-dnf-large-error}).
\end{proof}
Theorem~\ref{thm:multiparty-cc} settles Theorem~\ref{thm:MAIN-multiparty}
from the introduction.
\begin{rem}
\label{rem:quantum-NOF}In this section, we considered $k$-party
number-on-the-forehead bounded-error communication complexity with
\emph{classical} players. The model naturally extends to \emph{quantum}
players, and our lower bound in Theorem~\ref{thm:multiparty-cc}
implies an $\Omega(n^{1-\delta}/4^{k}k)$ communication lower bound
in this quantum $k$-party number-on-the-forehead model for computing
an explicit DNF formula $F\colon(\zoon)^{k}\to\zoo$ of size $n^{O(1)}$
and width $O(k)$ with error probability $\frac{1}{2}-\frac{1}{n^{A}},$
where the constants $\delta>0$ and $A\geq1$ can be set arbitrarily.
In more detail, the multiparty pattern matrix method actually gives
a bound on the \emph{generalized discrepancy} of the composed communication
problem $F$. By the results of~\cite{LSS2009quantum-multiparty},
generalized discrepancy leads in turn to a lower bound on the communication
complexity of $F$ in the quantum $k$-party number-on-the-forehead
model. Quantitatively, the authors of~\cite{LSS2009quantum-multiparty}
show that any classical communication lower bound obtained via generalized
discrepancy carries over to the quantum model with only a factor of
$\Theta(k)$ loss.
\end{rem}

\subsection{Nondeterministic and Merlin\textendash Arthur multiparty communication}

To obtain our results on nondeterminism and Merlin\textendash Arthur
communication, we will now develop a general technique for transforming
lower bounds on one-sided approximate degree into lower bounds in
these communication models. The technique in question is implicit
in the papers~\cite{GS09npconp,sherstov13directional} but has not
been previously formalized in our sought generality.

Consider a $k$-party communication problem $F\colon X_{1}\times X_{2}\times\cdots\times X_{k}\to\zoo,$
for some finite sets $X_{1},X_{2},\ldots,X_{k}$. A fundamental notion
in the study of multiparty communication is that of a \emph{cylinder
intersection~}\cite{bns92}\emph{, }defined as any function $\chi:X_{1}\times X_{2}\times\cdots\times X_{k}\to\zoo$
of the form 
\[
\chi(x_{1},\dots,x_{k})=\prod_{i=1}^{k}\phi_{i}(x_{1},\dots,x_{i-1},x_{i+1},\dots,x_{k})
\]
for some $\phi_{i}:X_{1}\times\cdots X_{i-1}\times X_{i+1}\times\cdots X_{k}\to\zoo,$
$i=1,2,\dots,k.$ In other words, a cylinder intersection is the product
of $k$ Boolean functions, where the $i$-th function does not depend
on the $i$-th coordinate. For a probability distribution $\mu$ on
the domain of $F,$ the \emph{discrepancy of $F$ with respect to
$\mu$} is denoted $\disc_{\mu}(F)$ and defined as 
\begin{align*}
\disc_{\mu}(F)=\max_{\chi}\left|\Exp_{(x_{1},\ldots,x_{k})\sim\mu}(-1)^{F(x_{1},\ldots,x_{k})}\chi(x_{1},\dots,x_{k})\right|,
\end{align*}
where the maximum is taken over all cylinder intersections $\chi.$
This notion of discrepancy was defined by Babai, Nisan, and Szegedy~\cite{bns92}
and is unrelated to the one that we encountered in Section~\ref{subsec:Discrepancy-defined}.
It is of interest to us because of the following theorem~\cite[Theorem~4.1]{GS09npconp},
which gives a lower bound on nondeterministic and Merlin\textendash Arthur
communication complexity in terms of discrepancy.
\begin{thm}[Gavinsky and Sherstov]
\label{thm:GS} Let $F:X\to\zoo$ be a given $k$-party communication
problem, where $X=X_{1}\times\cdots\times X_{k}.$ Fix a function
$H:X\to\zoo$ and a probability distribution $\Pi$ on $X.$ Put 
\begin{align*}
\alpha & =\Pi(F^{-1}(1)\cap H^{-1}(1)),\\
\beta & =\Pi(F^{-1}(1)\cap H^{-1}(0)),\\
Q & =\log\frac{\alpha}{\beta+\disc_{\Pi}(H)}.
\end{align*}
Then 
\begin{align}
 & N(F)\geq Q,\\
 & \ma_{1/3}(F)^{2}\geq\min\left\{ \Omega(Q),\;\Omega\left(\frac{Q}{\log\{2/\alpha\}}\right)^{2}\right\} \label{eq:GS-ma}\\
 & \phantom{\ma_{1/3}(F)^{2}}\geq\Omega(Q)-\left(\log\frac{2}{\alpha}\right)^{2}.\label{eq:GS-ma-simplified}
\end{align}
 
\end{thm}

\noindent We note that the original statement in~\cite{GS09npconp}
is for functions with range $\pomo.$ The above version for $\zoo$
follows immediately because the output values of a communication protocol
serve as textual labels that can be changed at will. Equation~(\ref{eq:GS-ma-simplified}),
which is also not part of the statement in~\cite{GS09npconp}, follows
from~(\ref{eq:GS-ma}) in view of the inequality $(q/a)^{2}\geq2q-a^{2}$
for all reals $q,a$ with $a\ne0.$ (Start with $(q/a-a)^{2}\geq0$
and multiply out the left-hand side.)

We will need yet another notion of discrepancy, introduced in~\cite{sherstov13directional}
and called ``repeated discrepancy.'' Let $G$ be a $k$-party communication
problem on $X=X_{1}\times X_{2}\times\cdots\times X_{k}$. A probability
distribution $\pi$ on the domain of $G$ is called \emph{balanced}
if $\pi(G^{-1}(0))=\pi(G^{-1}(1))=1/2.$ For such $\pi$, the \emph{repeated
discrepancy of $G$ with respect to $\pi$ }is given by
\[
\iterdisc_{\pi}(G)\,=\,\sup_{d,r\in\ZZ^{+}}\;\,\max_{\chi}\left|\Exp_{\ldots,x_{i,j},\ldots}\left[\chi(\dots,x_{i,j},\ldots)\prod_{i=1}^{d}(-1)^{G(x_{i,1})}\right]\right|^{1/d},
\]
where the maximum is over $k$-dimensional cylinder intersections
$\chi$ on $X^{dr}=X_{1}^{dr}\times X_{2}^{dr}\times\cdots\times X_{k}^{dr},$
and the arguments $x_{i,j}$ $(i=1,2,\dots,d;\;j=1,2,\dots,r)$ are
chosen independently according to $\pi$ conditioned on $G(x_{i,1})=G(x_{i,2})=\cdots=G(x_{i,r})$
for each $i$. The repeated discrepancy of a communication problem
is much harder to bound from above than standard discrepancy. The
following result from~\cite[Theorem~4.27]{sherstov13directional}
bounds the repeated discrepancy of set disjointness. 
\begin{thm}[Sherstov]
\label{thm:iter-discrepancy-DISJ} Let $m$ and $k$ be positive
integers. Then there is a balanced probability distribution $\pi$
on the domain of $\DISJ_{m,k}$ such that
\[
\iterdisc_{\pi}(\DISJ_{m,k})\leq\left(\frac{ck2^{k}}{\sqrt{m}}\right)^{1/2},
\]
where $c>0$ is an absolute constant independent of $m,k,\pi.$
\end{thm}

It was shown in~\cite{sherstov13directional} that repeated discrepancy
gives a highly efficient way to transform multiparty communication
protocols into polynomials. For a nonnegative integer $d$ and a function
$f$ on a finite subset of Euclidean space, define 
\[
E(f,d)=\min_{p}\{\|f-p\|_{\infty}:\deg p\leq d\},
\]
where the minimum is taken over polynomials of degree at most $d.$
In other words, $E(f,d)$ stands for the minimum error in an $\ell_{\infty}$-norm
approximation of $f$ by a polynomial of degree at most $d.$ The
following result was proved in~\cite[Theorem~4.2]{sherstov13directional}.
\begin{thm}[Sherstov]
\label{thm:iter-discrepancy-2-approximation} Let $G\colon X\to\zoo$
be a $k$-party communication problem, where $X=X_{1}\times X_{2}\times\cdots\times X_{k}$.
For an integer $n\geq1$ and a balanced probability distribution $\pi$
on the domain of $G$, consider the linear operator $L_{\pi,n}\colon\Re^{X^{n}}\to\Re^{\zoon}$
given by 
\begin{align}
(L_{\pi,n}\chi)(z) & =\Exp_{x_{1}\sim\pi_{z_{1}}}\cdots\Exp_{x_{n}\sim\pi_{z_{n}}}\chi(x_{1},\dots,x_{n}), &  & z\in\zoon,\label{eq:L-pi-n-defined}
\end{align}
where $\pi_{0}$ and $\pi_{1}$ are the probability distributions
induced by $\pi$ on $G^{-1}(0)$ and $G^{-1}(1),$ respectively.
Then for some absolute constant $c>0$ and every $k$-dimensional
cylinder intersection $\chi$ on $X^{n}=X_{1}^{n}\times X_{2}^{n}\times\cdots\times X_{k}^{n},$
\begin{align*}
E(L_{\pi,n}\chi,d-1) & \leq(c\iterdisc_{\pi}(G))^{d}, &  & d=1,2,\dots,n.
\end{align*}
\end{thm}

We are now in a position to derive the promised lower bound on nondeterministic
and Merlin\textendash Arthur communication complexity in terms of
one-sided approximate degree. Our proof combines Theorems~\ref{thm:GS}\textendash \ref{thm:iter-discrepancy-2-approximation}
in a way closely analogous to the proof of~\cite[Theorem~6.9]{sherstov13directional}.
\begin{thm}
\label{thm:pattern-matrix-N-MA}Let $f\colon\zoon\to\zoo$ be given.
Let $m$ and $k$ be positive integers, and put $F=f\circ\DISJ_{m,k}$.
Then for all $\epsilon\in(0,1/2],$
\begin{align}
N(F) & \geq\frac{\onedeg_{\epsilon}(f)}{2}\log\left(\frac{\sqrt{m}}{Ck2^{k}}\right)-\log\frac{1}{\epsilon},\label{eq:N-bound}\\
\ma_{1/3}(F)^{2} & \geq\frac{\onedeg_{\epsilon}(f)}{C}\log\left(\frac{\sqrt{m}}{Ck2^{k}}\right)-\left(\log\frac{2}{\epsilon}\right)^{2},\label{eq:MA-bound}
\end{align}
where $C\geq1$ is an absolute constant, independent of $f,n,m,k,\epsilon.$
\end{thm}

\begin{proof}
Abbreviate $d=\onedeg_{\epsilon}(f)$. Let $X=(\zoom)^{k}$ denote
the domain of $\DISJ_{m,k}$. By Theorem~\ref{thm:iter-discrepancy-DISJ},
there is a probability distribution $\pi$ on $X$ such that
\begin{align}
 & \pi(\DISJ_{m,k}^{-1}(0))=\pi(\DISJ_{m,k}^{-1}(1))=\frac{1}{2},\label{eq:mu-balanced}\\
 & \iterdisc_{\pi}(\DISJ_{m,k})\leq\left(\frac{c'k2^{k}}{\sqrt{m}}\right)^{1/2},\label{eq:mu-rdisc}
\end{align}
where $c'>0$ is an absolute constant independent of $m$ and $k.$
By the dual characterization of one-sided approximate degree (Fact~\ref{fact:onedeg-dual}),
there exists a function $\psi\colon\zoon\to\Re$ such that
\begin{align}
 & \langle f,\psi\rangle>\epsilon,\label{eq:pm-onesided-psi-correl}\\
 & \|\psi\|_{1}=1,\label{eq:pm-onesided-psi-ell1}\\
 & \orth\psi\geq d,\label{eq:pm-onesided-psi-orth}\\
 & \psi\geq0\;\text{ on }f^{-1}(1).\label{eq:pm-onesided-psi-nonneg}
\end{align}

Define $\Psi\colon X^{n}\to\Re$ by
\begin{equation}
\Psi(x)=2^{n}\psi(\DISJ_{m,k}(x_{1}),\ldots,\DISJ_{m,k}(x_{n}))\prod_{i=1}^{n}\pi(x_{i}).\label{eq:Psi-def}
\end{equation}
\begin{claim}
\label{claim:pm-Psi}$\Psi$ satisfies
\begin{align}
 & \langle F,\Psi\rangle>\epsilon,\label{eq:pm-F-Psi-inner}\\
 & \|\Psi\|_{1}=1,\label{eq:pm-Psi-ell1}\\
 & \Psi\geq0\quad\text{ on }F^{-1}(1).\label{eq:pm-Psi-nonneg}
\end{align}
\end{claim}

We will carry on with the theorem proof and settle the claims later.
Equation~(\ref{eq:pm-Psi-ell1}) allows us to write
\begin{equation}
\Psi(x)=\Pi(x)\cdot(-1)^{1-H(x)}\label{eq:Psi-Lambda-H}
\end{equation}
for some Boolean function $H\colon X^{n}\to\zoo$ and a probability
distribution $\Pi$ on $X^{n}.$ Indeed, one can explicitly define
$\Pi(x)=|\Psi(x)|$ and $H(x)=\I[\Psi(x)\geq0].$ 
\begin{claim}
\label{claim:pm-FH}One has
\begin{align}
 & \Pi(F^{-1}(1)\cap H^{-1}(0))=0,\label{eq:Lambda-F1-H0}\\
 & \Pi(F^{-1}(1)\cap H^{-1}(1))>\epsilon.\label{eq:Lambda-F1-H1}
\end{align}
\end{claim}

\begin{claim}
\label{claim:disc-Pi-H}There is an absolute constant $c>0$ such
that
\[
\disc_{\Pi}(H)\leq\left(\frac{ck2^{k}}{\sqrt{m}}\right)^{d/2}.
\]
\end{claim}

The sought communication bounds~(\ref{eq:N-bound}) and~(\ref{eq:MA-bound})
follow from Theorem~\ref{thm:GS} in view of Claims~\ref{claim:pm-FH}
and~\ref{claim:disc-Pi-H}.
\end{proof}
We now settle the claims used in the proof of Theorem~\ref{thm:pattern-matrix-N-MA}.
\begin{proof}[Proof of Claim~\emph{\ref{claim:pm-Psi}.}]
 We have
\begin{align*}
\langle F,\Psi\rangle & =2^{n}\Exp_{x_{1},\ldots,x_{n}\sim\pi}f(\ldots,\DISJ_{m,k}(x_{i}),\ldots)\psi(\ldots,\DISJ_{m,k}(x_{i}),\ldots)\\
 & =2^{n}\Exp_{z\in\zoon}f(z)\psi(z)\\
 & =\langle f,\psi\rangle\\
 & >\epsilon,
\end{align*}
where the second step uses~(\ref{eq:mu-balanced}), and the third
step is legitimate by~(\ref{eq:pm-onesided-psi-correl}). Analogously,
\begin{align*}
\|\Psi\|_{1} & =2^{n}\Exp_{x_{1},\ldots,x_{n}\sim\pi}|\psi(\DISJ_{m,k}(x_{1}),\ldots,\DISJ_{m,k}(x_{n}))|\\
 & =2^{n}\Exp_{z\in\zoon}|\psi(z)|\\
 & =1,
\end{align*}
where the last two steps are valid by~(\ref{eq:mu-balanced}) and~(\ref{eq:pm-onesided-psi-ell1}),
respectively. The final property~(\ref{eq:pm-Psi-nonneg}) can be
seen from the following chain of implications:
\begin{align*}
x\in F^{-1}(1) & \Rightarrow(\DISJ_{m,k}(x_{1}),\ldots,\DISJ_{m,k}(x_{n}))\in f^{-1}(1)\\
 & \Rightarrow\psi(\DISJ_{m,k}(x_{1}),\ldots,\DISJ_{m,k}(x_{n}))\geq0\\
 & \Rightarrow\Psi(x)\geq0,
\end{align*}
where the first and third steps use the definitions of $F$ and $\Psi,$
respectively, and the second step is valid by~(\ref{eq:pm-onesided-psi-nonneg}). 
\end{proof}
\begin{proof}[Proof of Claim~\emph{\ref{claim:pm-FH}.}]
 Fix any point $x\in F^{-1}(1)\cap H^{-1}(0).$ Then~(\ref{eq:pm-Psi-nonneg})
implies that $\Psi(x)\geq0,$ or equivalently $\Pi(x)\cdot(-1)^{1-H(x)}\geq0.$
This forces $\Pi(x)\leq0$ due to $H(x)=0$. Since $\Pi$ is a probability
distribution, we conclude that $\Pi(x)=0$. The proof of (\ref{eq:Lambda-F1-H0})
is complete. The remaining relation~(\ref{eq:Lambda-F1-H1}) can
be seen as follows:
\begin{align*}
\Pi(F^{-1}(1)\cap H^{-1}(1)) & =\Pi(F^{-1}(1)\cap H^{-1}(1))-\Pi(F^{-1}(1)\cap H^{-1}(0))\\
 & =\sum_{X^{n}}\Pi(x)F(x)(-1)^{1-H(x)}\\
 & =\langle\Psi,F\rangle\\
 & >\epsilon,
\end{align*}
where the first step exploits~(\ref{eq:Lambda-F1-H0}), and the last
step applies~(\ref{eq:pm-F-Psi-inner}).
\end{proof}
\begin{proof}[Proof of Claim~\emph{\ref{claim:disc-Pi-H}.}]
 Let $\pi_{0}$ and $\pi_{1}$ be the probability distributions induced
by $\pi$ on $\DISJ_{m,k}^{-1}(0)$ and $\DISJ_{m,k}^{-1}(1),$ respectively,
and let $L_{\pi,n}\colon\Re^{X^{n}}\to\Re^{\zoon}$ be the linear
operator given by~(\ref{eq:L-pi-n-defined}). Then for any cylinder
intersection $\chi\colon X^{n}\to\zoo$, we have
\begin{align}
\left|\Exp_{x\sim\Pi}(-1)^{H(x)}\chi(x)\right| & =\left|\sum_{X^{n}}\Pi(x)(-1)^{H(x)}\chi(x)\right|\nonumber \\
 & =\left|\sum_{X^{n}}\Psi(x)\chi(x)\right|\nonumber \\
 & =\left|2^{n}\Exp_{x_{1},\ldots,x_{n}\sim\pi}\psi(\ldots,\DISJ_{m,k}(x_{i}),\ldots)\chi(x)\right|\nonumber \\
 & =\left|\sum_{z\in\zoon}\psi(z)\Exp_{x_{1}\sim\pi_{z_{1}}}\ldots\Exp_{x_{n}\sim\pi_{z_{n}}}\chi(x)\right|\nonumber \\
 & =|\langle\psi,L_{\pi,n}\chi\rangle|,\label{eq:disc-intermediate}
\end{align}
where the second step uses~(\ref{eq:Psi-Lambda-H}), the third step
invokes the definition~(\ref{eq:Psi-def}), the fourth step is justified
by~(\ref{eq:mu-balanced}), and the last step is valid by the definition
of $L_{\pi,n}$.

For every polynomial $p\colon\zoon\to\Re$ of degree less than $d$,
we have
\begin{align}
|\langle\psi,L_{\pi,n}\chi\rangle| & =|\langle\psi,L_{\pi,n}\chi-p\rangle+\langle\psi,p\rangle|\nonumber \\
 & =|\langle\psi,L_{\pi,n}\chi-p\rangle|\nonumber \\
 & \leq\|\psi\|_{1}\,\|L_{\pi,n}\chi-p\|_{\infty}\nonumber \\
 & =\|L_{\pi,n}\chi-p\|_{\infty},\label{eq:psi-L-chi}
\end{align}
where the second step uses~(\ref{eq:pm-onesided-psi-orth}), the
third step applies H\"{o}lder's inequality, and the fourth step substitutes~(\ref{eq:pm-onesided-psi-ell1}).
Taking the infimum in~(\ref{eq:psi-L-chi}) over all polynomials
$p$ of degree less than $d,$ we arrive at 
\begin{equation}
|\langle\psi,L_{\pi,n}\chi\rangle|\leq E(L_{\pi,n}\chi,d-1).\label{eq:psi-L-chi-E}
\end{equation}
Now
\begin{align*}
\disc_{\Pi}(H) & =\max_{\chi}\left|\Exp_{x\sim\Pi}(-1)^{H(x)}\chi(x)\right|\\
 & \leq\max_{\chi}E(L_{\pi,n}\chi,d-1)\\
 & \leq(c''\iterdisc_{\pi}(\DISJ_{m,k}))^{d}\\
 & \leq\left(c''\left(\frac{c'k2^{k}}{\sqrt{m}}\right)^{1/2}\right)^{d},
\end{align*}
where the first step maximizes over all cylinder intersections $\chi,$
the second step combines~(\ref{eq:disc-intermediate}) and~(\ref{eq:psi-L-chi-E}),
the third step is valid for some absolute constant $c''>0$ by Theorem~\ref{thm:iter-discrepancy-2-approximation},
and the fourth step holds by~(\ref{eq:mu-rdisc}).
\end{proof}
This completes the proof of Theorem~\ref{thm:pattern-matrix-N-MA}.
By combining it with our main result on one-sided approximate degree,
we now obtain our sought lower bounds for nondeterministic and Merlin\textendash Arthur
multiparty communication. 
\begin{thm}
\label{thm:multiparty-cc-N-MA} Let $\delta>0$ be arbitrary. Then
for all integers $n,k\geq2,$ there is an $($explicitly given$)$
 $k$-party communication problem $F_{n,k}\colon(\zoon)^{k}\to\zoo$
with
\begin{align}
 & N(F_{n,k})\leq c\log n,\label{eq:N-dnf}\\
 & N(\neg F_{n,k})\geq\left(\frac{n}{c4^{k}k^{2}}\right)^{1-\delta},\label{eq:N-cnf}\\
 & R_{1/3}(\neg F_{n,k})\geq\left(\frac{n}{c4^{k}k^{2}}\right)^{1-\delta},\label{eq:R-cnf}\\
 & \ma_{1/3}(\neg F_{n,k})\geq\left(\frac{n}{c4^{k}k^{2}}\right)^{\frac{1-\delta}{2}},\label{eq:ma-cnf}
\end{align}
where $c\geq1$ is a constant independent of $n$ and $k.$ Moreover,
each $F_{n,k}$ is computable by a monotone DNF formula of width $ck$
and size $n^{c}$.
\end{thm}

\begin{proof}
Theorem~\ref{thm:adeg-dnf-onesided} gives a constant $c'\geq1$
and an explicit family $\{f_{n}\}_{n=1}^{\infty}$ of functions $f_{n}\colon\{0,1\}^{n}\to\zoo$
such that each $f_{n}$ is computable by a monotone DNF formula of
width $c'$ and satisfies
\begin{align}
\deg_{3/8}^{+}(\neg f_{n}) & \geq\frac{1}{c'}\cdot n^{1-\delta}, &  & n=1,2,3,\ldots.\label{eq:hard-dnf-f-1}
\end{align}
In particular,
\begin{align}
\deg_{3/8}(\neg f_{n}) & \geq\frac{1}{c'}\cdot n^{1-\delta}, &  & n=1,2,3,\ldots.\label{eq:hard-dnf-f-2}
\end{align}
Let $C\geq1$ be the maximum of the absolute constants from Theorems~\ref{thm:pattern-matrix-randomized}
and~\ref{thm:pattern-matrix-N-MA}. For arbitrary integers $n,k\geq2,$
define
\[
F_{n,k}=\begin{cases}
\AND_{k} & \text{if }n<\lceil C2^{k+1}k\rceil^{2},\\
f_{\lfloor n/\lceil C2^{k+1}k\rceil^{2}\rfloor}\circ\neg\DISJ_{\lceil C2^{k+1}k\rceil^{2},k} & \text{otherwise.}
\end{cases}
\]
We first analyze the cost of representing $F_{n,k}$ as a DNF formula.
If $n<\lceil C2^{k+1}k\rceil^{2},$ then by definition $F_{n,k}$
is a monotone DNF formula of width $k$ and size $1.$ In the complementary
case, $f_{\lfloor n/\lceil C2^{k+1}k\rceil^{2}\rfloor}$ is by construction
a monotone DNF formula of width $c'$ and hence of size at most $n^{c'},$
whereas $\neg\DISJ_{\lceil C2^{k+1}k\rceil^{2},k}$ is by definition
a monotone DNF formula of width $k$ and size at most $\lceil C2^{k+1}k\rceil^{2}\leq n.$
As a result, the composed function $F_{n,k}$ is a monotone DNF formula
of width $c'k$ and size at most $n^{c'}\cdot n^{c'}=n^{2c'}.$ In
particular, the claim in the theorem statement regarding the width
and size of $F_{n,k}$ as a monotone DNF formula is valid for any
large enough $c.$ This in turn implies the upper bound in~(\ref{eq:N-dnf}):
consider the nondeterministic protocol in which the parties ``guess''
one of the terms of the DNF formula for $F_{n,k}$ (for a cost of
$\lceil\log n^{2c'}\rceil$ bits), evaluate it (using another $2$
bits of communication), and output the result.

We now turn to the communication lower bounds. Since $F_{n,k}$ is
nonconstant, we have the trivial bounds
\begin{align}
 & N(\neg F_{n,k})\geq1,\label{eq:N-F-nk-trivial-cc}\\
 & R_{1/3}(\neg F_{n,k})\geq1,\\
 & \ma_{1/3}(\neg F_{n,k})\geq1.\label{eq:MA-F-nk-trivial-cc}
\end{align}
We further claim that
\begin{align}
 & N(\neg F_{n,k})\geq\frac{1}{2c'}\cdot\left\lfloor \frac{n}{\lceil C2^{k+1}k\rceil^{2}}\right\rfloor ^{1-\delta}-\log\frac{8}{3},\label{eq:N-general}\\
 & R_{1/3}(\neg F_{n,k})\geq\frac{1}{2c'}\cdot\left\lfloor \frac{n}{\lceil C2^{k+1}k\rceil^{2}}\right\rfloor ^{1-\delta}-\log12,\label{eq:R}\\
 & \ma_{1/3}(\neg F_{n,k})^{2}\geq\frac{1}{Cc'}\cdot\left\lfloor \frac{n}{\lceil C2^{k+1}k\rceil^{2}}\right\rfloor ^{1-\delta}-\left(\log\frac{16}{3}\right)^{2}.\label{eq:MA-general}
\end{align}
For $n<\lceil C2^{k+1}k\rceil^{2},$ these claims are trivial since
communication complexity is nonnegative. In the complementary case
$n\geq\lceil C2^{k+1}k\rceil^{2},$ consider the family $\{g_{n}\}_{n=1}^{\infty}$
of functions $g_{n}\colon\zoon\to\zoo$ given by $g_{n}(x_{1},x_{2},\ldots,x_{n})=\neg f_{n}(\neg x_{1},\neg x_{2},\ldots,\neg x_{n})$.
For each $n,$ it is clear that $g_{n}$ and $\neg f_{n}$ have the
same one-sided approximate degree. Since $\neg F_{n,k}=g_{\lfloor n/\lceil C2^{k+1}k\rceil^{2}\rfloor}\circ\DISJ_{\lceil C2^{k+1}k\rceil^{2},k},$
one now obtains~(\ref{eq:N-general}) and (\ref{eq:MA-general})
directly from~(\ref{eq:hard-dnf-f-1}) and Theorem~\ref{thm:pattern-matrix-N-MA}.
Analogously, $g_{n}$ and $\neg f_{n}$ have the same two-sided approximate
degree for each $n$, and one obtains~(\ref{eq:R}) from (\ref{eq:hard-dnf-f-2})
and Theorem~\ref{thm:pattern-matrix-randomized}. 

For a large enough constant $c\geq1$, the communication lower bounds~(\ref{eq:N-cnf})\textendash (\ref{eq:ma-cnf})
follow from~(\ref{eq:N-general})\textendash (\ref{eq:MA-general})
for $n\geq c4^{k}k^{2}$, and from~(\ref{eq:N-F-nk-trivial-cc})\textendash (\ref{eq:MA-F-nk-trivial-cc})
for $n<c4^{k}k^{2}$.
\end{proof}
Theorem~\ref{thm:multiparty-cc-N-MA} settles Theorem~\ref{thm:MAIN-N-MA}
from the introduction.

\section*{Acknowledgments}

The author is thankful to Justin Thaler and Mark Bun for useful comments
on an earlier version of this paper.

\bibliographystyle{siamplain}
\bibliography{refs}

\appendix

\section{\label{app:ajtai-iteration}Constructing low-discrepancy integer
sets}

\addtocontents{toc}{\setcounter{tocdepth}{-10}} 

The purpose of this appendix is to provide a detailed and self-contained
proof of Theorem~\ref{thm:ajtai-iteration}, restated below.
\begin{thm*}
Fix an integer $R\geq1$ and reals $P\geq2$ and $\Delta\geq1$. Let
$m$ be an integer with
\[
m\geq P^{2}(R+1).
\]
Fix a set $S_{p}\subseteq\{1,2,\ldots,p-1\}$ for each prime $p\in(P/2,P]$
with $p\nmid m$. Suppose further that the cardinalities of any two
sets from among the $S_{p}$ differ by a factor of at most $\Delta.$
Consider the multiset
\begin{multline*}
S=\{(r+s\cdot(p^{-1})_{m})\bmod m:\\
\qquad r=1,\ldots,R;\quad p\in(P/2,P]\text{ prime with }p\nmid m;\quad s\in S_{p}\}.
\end{multline*}
Then the elements of $S$ are pairwise distinct and nonzero. Moreover,
if $S\ne\varnothing$ then
\begin{equation}
\disc_{m}(S)\leq\frac{c}{\sqrt{R}}+\frac{c\log m}{\log\log m}\cdot\frac{\log P}{P}\cdot\Delta+\max_{p}\{\disc_{p}(S_{p})\}\label{eq:disc-S}
\end{equation}
for some $($explicitly given$)$ constant $c\geq1$ independent of
$P,R,m,\Delta.$
\end{thm*}
\noindent The special case $\Delta=1$ in this result was proved in~\cite[Theorem~3.6]{sherstov18hardest-hs},
and that proof applies with cosmetic changes to any $\Delta\geq1.$
As a service to the reader, we provide the complete derivation below;
the treatment here is the same word for word as in~\cite{sherstov18hardest-hs}
except for one minor point of departure to handle arbitrary $\Delta\geq1$.
We use the same notation as in~\cite{sherstov18hardest-hs} and in
particular denote the modulus by lowercase $m,$ as opposed to the
uppercase $M$ in the main body of our paper (Theorem~\ref{thm:ajtai-iteration}).
Analogous to~\cite{sherstov18hardest-hs}, the presentation is broken
down into five key milestones, corresponding to Sections~\ref{subsec:Shorthand-notation}\textendash \ref{subsec:Finishing-the-proof}
below.

\subsection{Exponential notation\label{subsec:Shorthand-notation}}

In the remainder of this manuscript, we adopt the shorthand
\[
e(x)=\exp(2\pi x\iu),
\]
where $\iu$ is the imaginary unit. We will need the following bounds~\cite[Section~6.1]{sherstov18hardest-hs}:
\begin{align}
|1-e(x)| & \leq2\pi x, &  & 0\leq x\leq1,\label{eq:e-close-to-1}\\
|1-e(x)| & \geq4\min(x,1-x), &  & 0\leq x\leq1.\label{eq:e-far-from-1}
\end{align}

Let $\Pcal$ denote the set of prime numbers $p\in(P/2,P]$ with $p\nmid m.$
In this notation, the multiset $S$ is given by
\[
S=\{(r+s\cdot(p^{-1})_{m})\bmod m:p\in\Pcal,\;s\in S_{p},\;r=1,2,\ldots,R\}.
\]
There are precisely $\pi(P)-\pi(P/2)$ primes in $(P/2,P],$ of which
at most $\nu(m)$ are prime divisors of $m.$ Therefore,
\begin{equation}
|\Pcal|\geq\pi(P)-\pi\left(\frac{P}{2}\right)-\nu(m).\label{eq:Pcal-lower}
\end{equation}

\subsection{\label{subsec:Elements-of-Sm-are-nonzero}Elements of \emph{S} are
nonzero and distinct}

As our first step, we verify that the elements of $S$ are nonzero
and distinct modulo $m$. This part of the argument is reproduced
word for word from~\cite[Section~6.2]{sherstov18hardest-hs}.

Specifically, consider any $r\in\{1,2,\ldots,R\}$, any prime $p\in(P/2,P]$
with $p\nmid m,$ and any $s\in S_{p}.$ Then $pr+s\in[1,PR+P-1]\subseteq[1,m).$
This means that $pr+s\not\equiv0\pmod m,$ which in turn implies that
$r+s\cdot(p^{-1})_{m}\not\equiv0\pmod m$.

We now show that the multiset $S$ contains no repeated elements.
For this, consider any $r,r'\in\{1,2,\ldots,R\},$ any primes $p,p'\in\Pcal,$
and any $s\in S_{p}$ and $s'\in S_{p'}$ such that 
\begin{equation}
r+s\cdot(p^{-1})_{m}\equiv r'+s'\cdot(p'^{-1})_{m}\pmod m.\label{eq:modular-equality}
\end{equation}
Our goal is to show that $p=p',r=r',s=s'.$ To this end, multiply
(\ref{eq:modular-equality}) through by $pp'$ to obtain 
\begin{equation}
r\cdot pp'+s\cdot p'\equiv r'\cdot pp'+s'\cdot p\pmod m.\label{eq:modular-equality-pp}
\end{equation}
The left-hand side and right-hand side of (\ref{eq:modular-equality-pp})
are integers in $[1,RP^{2}+(P-1)P]\subseteq[1,m),$ whence
\begin{equation}
r\cdot pp'+s\cdot p'=r'\cdot pp'+s'\cdot p.\label{eq:straight-equality}
\end{equation}
This implies that $p\mid s\cdot p',$ which in view of $s<p$ and
the primality of $p$ and $p'$ forces $p=p'.$ Now (\ref{eq:straight-equality})
simplifies to 
\begin{equation}
r\cdot p+s=r'\cdot p+s',\label{eq:straight-equality-simplified}
\end{equation}
which in turn yields $s\equiv s'\pmod p$. Recalling that $s,s'\in\{1,2,\ldots,p-1\},$
we arrive at $s=s'.$ Finally, substituting $s=s'$ in~(\ref{eq:straight-equality-simplified})
gives $r=r'.$

\subsection{\label{subsec:Correlation-bound-for-k-small}Correlation for \emph{k}
small}

So far, we have shown that the elements of $S$ are distinct and nonzero.
To bound the $m$-discrepancy of this set, we must bound the exponential
sum 
\begin{equation}
\left|\sum_{s\in S}e\left(\frac{k}{m}\cdot s\right)\right|\label{eq:exp-sum}
\end{equation}
for all $k=1,2,\ldots,m-1.$ This subsection and the next provide
two complementary bounds on~(\ref{eq:exp-sum}). The first bound,
presented below, is preferable when $k$ is close to zero modulo $m.$
\begin{claim}
\label{claim:k-small}Let $k\in\{1,2,\ldots,m-1\}$ be given. Then
\begin{multline*}
\left|\sum_{s\in S}e\left(\frac{k}{m}\cdot s\right)\right|\\
\leq\left(\frac{2\pi\min(k,m-k)}{m}+\max_{p\in\Pcal}\{\disc_{p}(S_{p})\}+\frac{\nu(k)+\nu(m-k)}{|\Pcal|}\cdot\Delta\right)|S|.
\end{multline*}
\end{claim}

This claim generalizes the analogous statement in~\cite[Claim 6.10]{sherstov18hardest-hs},
where the special case $\Delta=1$ was considered.
\begin{proof}
Let $\Pcal'$ be the set of those primes in $\Pcal$ that divide neither
$k$ nor $m-k.$ Then clearly 
\begin{equation}
|\Pcal\setminus\Pcal'|\leq\nu(k)+\nu(m-k).\label{eq:P-minus-P'}
\end{equation}
Exactly as in~\cite{sherstov18hardest-hs}, we have
\begin{align}
 & \left|\sum_{s\in S}e\left(\frac{k}{m}\cdot s\right)\right|\nonumber \\
 & \qquad=\left|\sum_{r=1}^{R}\sum_{p\in\Pcal}\sum_{s\in S_{p}}e\left(\frac{k}{m}\cdot(r+s\cdot(p^{-1})_{m})\right)\right|\nonumber \\
 & \qquad\leq\sum_{r=1}^{R}\sum_{p\in\Pcal}\left|\sum_{s\in S_{p}}e\left(\frac{k}{m}\cdot(r+s\cdot(p^{-1})_{m})\right)\right|\nonumber \\
 & \qquad=R\sum_{p\in\Pcal}\left|\sum_{s\in S_{p}}e\left(\frac{ks\cdot(p^{-1})_{m}}{m}\right)\right|\nonumber \\
 & \qquad\le R\sum_{p\in\Pcal'}\left|\sum_{s\in S_{p}}e\left(\frac{ks\cdot(p^{-1})_{m}}{m}\right)\right|+R\sum_{p\in\Pcal\setminus\Pcal'}\left|\sum_{s\in S_{p}}e\left(\frac{ks\cdot(p^{-1})_{m}}{m}\right)\right|\nonumber \\
 & \qquad\leq R\sum_{p\in\Pcal'}\left|\sum_{s\in S_{p}}e\left(\frac{ks\cdot(p^{-1})_{m}}{m}\right)\right|+R\sum_{p\in\Pcal\setminus\Pcal'}|S_{p}|.\label{eq:k-near-zero-decomposition}
\end{align}

We proceed to bound the two summations in~(\ref{eq:k-near-zero-decomposition}).
Bounding the second summation is straightforward:
\begin{align}
R\sum_{p\in\Pcal\setminus\Pcal'}|S_{p}| & \leq R\cdot\frac{|\Pcal\setminus\Pcal'|}{|\Pcal|}\sum_{p\in\Pcal}\Delta|S_{p}|\nonumber \\
 & =\frac{|\Pcal\setminus\Pcal'|}{|\Pcal|}\cdot\Delta|S|\nonumber \\
 & \leq\frac{\nu(k)+\nu(m-k)}{|\Pcal|}\cdot\Delta|S|,\label{eq:k-near-zero-easy-sum}
\end{align}
where the first step is valid because the cardinalities of any two
sets $S_{p}$ differ by a factor of at most $\Delta,$ and the last
step uses~(\ref{eq:P-minus-P'}). This three-line derivation is our
only point of departure from the treatment in~\cite{sherstov18hardest-hs}.

The other summation in~(\ref{eq:k-near-zero-decomposition}) is analyzed
exactly as in~\cite{sherstov18hardest-hs}. For $p\in\Pcal'$ and
$K\in\{k,k-m\},$ we have
\begin{align*}
 & \hspace{-7mm}\left|\sum_{s\in S_{p}}e\left(\frac{ks\cdot(p^{-1})_{m}}{m}\right)\right|\\
 & =\left|\sum_{s\in S_{p}}e\left(\frac{Ks\cdot(p^{-1})_{m}}{m}\right)\right|\\
 & =\left|\sum_{s\in S_{p}}e\left(-\frac{Ks\cdot(m^{-1})_{p}}{p}\right)e\left(\frac{Ks}{pm}\right)\right|\\
 & \leq\left|\sum_{s\in S_{p}}e\left(-\frac{Ks\cdot(m^{-1})_{p}}{p}\right)\left(e\left(\frac{Ks}{pm}\right)-1\right)\right|+\left|\sum_{s\in S_{p}}e\left(-\frac{Ks\cdot(m^{-1})_{p}}{p}\right)\right|\\
 & \leq\left|\sum_{s\in S_{p}}e\left(-\frac{Ks\cdot(m^{-1})_{p}}{p}\right)\left(e\left(\frac{Ks}{pm}\right)-1\right)\right|+\disc_{p}(S_{p})\cdot|S_{p}|\\
 & \leq\sum_{s\in S_{p}}\left|e\left(\frac{Ks}{pm}\right)-1\right|+\disc_{p}(S_{p})\cdot|S_{p}|\\
 & =\sum_{s\in S_{p}}\left|e\left(\frac{|K|s}{pm}\right)-1\right|+\disc_{p}(S_{p})\cdot|S_{p}|\\
 & \leq|S_{p}|\cdot\frac{2\pi|K|}{m}+\disc_{p}(S_{p})\cdot|S_{p}|,
\end{align*}
where the second step uses Fact~\ref{fact:rel-prime} and the relative
primality of $p$ and $m$; the third step applies the triangle inequality;
the fourth step follows from $p\nmid|K|$, and the last step is valid
by~(\ref{eq:e-close-to-1}) and $s<p$. We have shown that
\begin{align*}
\left|\sum_{s\in S_{p}}e\left(\frac{ks\cdot(p^{-1})_{m}}{m}\right)\right| & \leq\frac{2\pi\min(k,m-k)}{m}\cdot|S_{p}|+\disc_{p}(S_{p})\cdot|S_{p}|
\end{align*}
for $p\in\Pcal'.$ Summing over $\Pcal',$
\begin{align}
R\sum_{p\in\Pcal'} & \left|\sum_{s\in S_{p}}e\left(\frac{ks\cdot(p^{-1})_{m}}{m}\right)\right|\nonumber \\
 & \qquad\leq R\sum_{p\in\Pcal'}\left(\frac{2\pi\min(k,m-k)}{m}\cdot|S_{p}|+\disc_{p}(S_{p})\cdot|S_{p}|\right)\nonumber \\
 & \qquad\leq R\sum_{p\in\Pcal}\left(\frac{2\pi\min(k,m-k)}{m}\cdot|S_{p}|+\disc_{p}(S_{p})\cdot|S_{p}|\right)\nonumber \\
 & \qquad\leq\left(\frac{2\pi\min(k,m-k)}{m}+\max_{p\in\Pcal}\{\disc_{p}(S_{p})\}\right)R\sum_{p\in\Pcal}|S_{p}|\nonumber \\
 & \qquad=\left(\frac{2\pi\min(k,m-k)}{m}+\max_{p\in\Pcal}\{\disc_{p}(S_{p})\}\right)|S|.\label{eq:k-near-zero-hard-sum}
\end{align}
By~(\ref{eq:k-near-zero-decomposition})\textendash (\ref{eq:k-near-zero-hard-sum}),
the proof of the claim is complete.
\end{proof}

\subsection{\label{subsec:Correlation-bound-for-k-large}Correlation for $k$
large}

We now present an alternative bound on the exponential sum~(\ref{eq:exp-sum}),
which is preferable to the bound of Claim~\ref{claim:k-small} when
$k$ is far from zero modulo $m.$ This part of the proof is reproduced
verbatim from~\cite[Section~6.4]{sherstov18hardest-hs}.
\begin{claim}
\label{claim:k-large}Let $k\in\{1,2,\ldots,m-1\}$ be given. Then
\[
\left|\sum_{s\in S}e\left(\frac{k}{m}\cdot s\right)\right|\leq\frac{m}{2R\min(k,m-k)}\cdot|S|.
\]
\end{claim}

\begin{proof}[Proof:]
\begin{align*}
\left|\sum_{s\in S}e\left(\frac{k}{m}\cdot s\right)\right| & =\left|\sum_{p\in\Pcal}\sum_{s\in S_{p}}\sum_{r=1}^{R}e\left(\frac{k}{m}\cdot(r+s\cdot(p^{-1})_{m})\right)\right|\\
 & \leq\sum_{p\in\Pcal}\sum_{s\in S_{p}}\left|\sum_{r=1}^{R}e\left(\frac{k}{m}\cdot(r+s\cdot(p^{-1})_{m})\right)\right|\\
 & =\sum_{p\in\Pcal}\sum_{s\in S_{p}}\left|\sum_{r=1}^{R}e\left(\frac{kr}{m}\right)\right|\\
 & =\sum_{p\in\Pcal}\sum_{s\in S_{p}}\frac{|1-e(kR/m)|}{|1-e(k/m)|}\\
 & \leq\sum_{p\in\Pcal}\sum_{s\in S_{p}}\frac{2}{|1-e(k/m)|}\\
 & \leq\sum_{p\in\Pcal}\sum_{s\in S_{p}}\frac{m}{2\min(k,m-k)}\\
 & =\frac{m}{2R\min(k,m-k)}\cdot|S|,
\end{align*}
where the last two steps use~(\ref{eq:e-far-from-1}) and $|S|=R\sum_{p\in\Pcal}|S_{p}|$,
respectively.
\end{proof}

\subsection{\label{subsec:Finishing-the-proof}Finishing the proof}

The remainder of the proof is reproduced without changes from~\cite[Section~6.5]{sherstov18hardest-hs},
except for the use of the updated bound in Claim~\ref{claim:k-small}
for arbitrary $\Delta\geq1.$

Specifically, Facts~\ref{fact:PNT} and~\ref{fact:num-prime-factors}
imply that
\begin{align}
\pi(P)-\pi\left(\frac{P}{2}\right) & \geq\frac{P}{C\log P}\qquad\qquad(P\geq C),\label{eq:many-primes}\\
\max_{k=1,2,\ldots,m}\nu(k) & \leq\frac{C\log m}{\log\log m},\label{eq:few-prime-factors}
\end{align}
where $C\geq1$ is a constant independent of $R,P,m,\Delta.$ Moreover,
$C$ can be easily calculated from the explicit bounds in Facts~\ref{fact:PNT}
and~\ref{fact:num-prime-factors}. We will show that the theorem
conclusion~(\ref{eq:disc-S}) holds with $c=4C^{2}.$ We may assume
that
\begin{align}
 & P\geq C,\label{eq:P-geq-C}\\
 & \frac{C\log m}{\log\log m}\leq\frac{P}{2C\log P},\label{eq:more-primes-than-prime-factors}
\end{align}
since otherwise the right-hand side of~(\ref{eq:disc-S}) exceeds~$1$
and the theorem is trivially true. By~(\ref{eq:Pcal-lower}) and
(\ref{eq:many-primes})\textendash (\ref{eq:more-primes-than-prime-factors}),
we obtain
\[
|\Pcal|\geq\frac{P}{2C\log P},
\]
which along with~(\ref{eq:few-prime-factors}) gives
\begin{align}
\max_{k=1,2,\ldots,m-1}\frac{\nu(k)+\nu(m-k)}{|\Pcal|} & \leq\frac{2C\log m}{\log\log m}\cdot\frac{2C\log P}{P}\nonumber \\
 & =\frac{c\log m}{\log\log m}\cdot\frac{\log P}{P}.\label{eq:ratio-bound}
\end{align}
Claims~\ref{claim:k-small} and~\ref{claim:k-large} ensure that
for every $k=1,2,\ldots,m-1,$
\begin{align*}
\left|\sum_{s\in S}e\left(\frac{k}{m}\cdot s\right)\right| & \leq\left(\min\left(\frac{2\pi\min(k,m-k)}{m},\frac{m}{2R\min(k,m-k)}\right)\right.\\
 & \qquad\qquad\left.+\max_{p\in\Pcal}\{\disc_{p}(S_{p})\}+\frac{\nu(k)+\nu(m-k)}{|\Pcal|}\cdot\Delta\right)|S|\\
 & \leq\left(\sqrt{\frac{\pi}{R}}+\max_{p\in\Pcal}\{\disc_{p}(S_{p})\}+\frac{\nu(k)+\nu(m-k)}{|\Pcal|}\cdot\Delta\right)|S|\\
 & \leq\left(\frac{c}{\sqrt{R}}+\max_{p\in\Pcal}\{\disc_{p}(S_{p})\}+\frac{\nu(k)+\nu(m-k)}{|\Pcal|}\cdot\Delta\right)|S|;
\end{align*}
here we are using the updated bound from Claim~\ref{claim:k-small}
in this paper for general $\Delta$. Substituting the estimate from~(\ref{eq:ratio-bound}),
we conclude that
\begin{multline*}
\max_{k=1,2,\ldots,m-1}\left|\sum_{s\in S}e\left(\frac{k}{m}\cdot s\right)\right|\\
\leq\left(\frac{c}{\sqrt{R}}+\max_{p\in\Pcal}\{\disc_{p}(S_{p})\}+\frac{c\log m}{\log\log m}\cdot\frac{\log P}{P}\cdot\Delta\right)|S|.\qquad\qquad
\end{multline*}
This conclusion is equivalent to~(\ref{eq:disc-S}). The proof of
Theorem~\ref{thm:ajtai-iteration} is complete.

\addtocontents{toc}{\setcounter{tocdepth}{1}}  
\end{document}